\pdfoutput=1
\documentclass[11pt,a4paper]{article}
\usepackage[margin=1in]{geometry}

\newcommand{\arxiv}[1]{#1}
\newcommand{\conference}[1]{}

\makeatletter
\arxiv{
  
}
\makeatother

\usepackage{cite}

\usepackage[utf8]{inputenc}
\usepackage{graphicx}
\graphicspath{{figures/}}
\usepackage{amsmath,amssymb,amsfonts}
\usepackage{mathtools}
\usepackage{amsthm}
\usepackage{bm}
\usepackage{mathrsfs}
\usepackage{textcomp}
\usepackage{xcolor, colortbl}
\usepackage{comment}
\usepackage{pifont}
\usepackage{enumitem}
\usepackage{fontawesome5}

\allowdisplaybreaks

\let\oldhat\hat
\renewcommand{\hat}[1]{%
  \ifx#1f
    \oldhat{x}%
  \else
    \oldhat{#1}%
  \fi
}

\newcommand{\cmmnt}[1]{}

\usepackage[colorlinks=true, linkcolor=blue, citecolor=blue, urlcolor=blue]{hyperref}
\usepackage{cleveref}

\crefname{section}{Appendix}{appendices}
\Crefname{section}{Appendix}{Appendices}

\usepackage[noend]{algpseudocode}
\usepackage{algorithm}
\usepackage{algorithmicx}

\usepackage{tikz}
\usetikzlibrary{positioning}
\usepackage{pgfplots}
\usepgfplotslibrary{fillbetween}
\pgfplotsset{compat=1.18}
\usetikzlibrary{calc, arrows.meta, shapes, decorations.pathreplacing, decorations.pathmorphing}
\usepackage{subcaption}

\newtheorem{theorem}{Theorem}

\newtheorem{lemma}{Lemma}

\newtheorem{assumption}{Assumption}
\newtheorem{proposition}{Proposition}

\makeatletter
\AtBeginDocument{%
  \def\url#1{}%
  \def\online#1{}%
  \def\IEEEurl#1#2{}%
  \def\IEEEurl@#1#2{}%
}
\makeatother

\makeatletter
\AtBeginDocument{
  \providecommand{\bstctlcite}[1]{}
  
  \def\IEEEurl#1#2{}

}
\def\online#1{}
\makeatother

\title{Conformal Prediction Regions for Continuous-Time Trajectories under Random Sampling}
\author{
  Joaquin Alvarez\footnote{Department of Computer Science, University of Southern California}, 
  Matteo Sesia\footnote{Departments of Data Sciences and Operations, and of Computer Science, University of Southern California}, 
  Jyotirmoy V. Deshmukh\footnotemark[1], and 
  Lars Lindemann\footnote{Automatic Control Laboratory, ETH Zürich}
}
\date{}
\begin{document}
\maketitle
\begingroup
\renewcommand{\thefootnote}{}
\footnotetext{This work was supported by the National Science Foundation under the Grant IIS-SLES-2417075 and by a USC-Capital One CREDIF award.}
\addtocounter{footnote}{-1}
\endgroup

\begin{abstract}
Uncertainty quantification for continuous-time trajectories is a prerequisite in many safety-critical engineering domains. However, a major challenge in data-driven uncertainty quantification is that calibration trajectories are sampled only at discrete, often sparse, and random intervals. Standard conformal prediction methods typically fail to provide guarantees in between  sampling times. In this work, we introduce a new technique to obtain valid conformal prediction regions for continuous-time trajectories that are sampled at discrete and possibly random times. To accomplish this goal, we make three contributions: (1) we provide an algorithm that leverages regularity properties of the underlying trajectories to obtain valid prediction regions in between samples, (2) we provide methods that estimate valid bounds on the aforementioned regularity properties from an additional high-frequency calibration dataset, and (3) we introduce and compare several algorithms to deal with random sampling times. Finally, we present experiments where we validate that our methods achieve valid coverage across the entire continuous trajectory.
\end{abstract}

\section{Introduction}

Ensuring that essential safety properties hold over a continuous-time domain is a fundamental requirement across various engineering domains. In practice, however, trajectories are often observed only at discrete, irregular intervals. This introduces a fundamental risk: a system may appear safe at times when observations are made, but violate safety in between  consecutive observations. For example, networked control systems exchange information digitally and may suffer from packet loss \cite{hespanha2007survey, packet_loss, packet_loss_bernoulli}, aviation safety margins can degrade due to atmospheric blocking or GPS spoofing \cite{schmolter2024space,Tsardanidis_2025}, and medical decision support systems frequently operate on inconsistent patient data \arxiv{\cite{rubanova2019latent, tassopoulou2025uncertaintycalibrated}}\conference{\cite{tassopoulou2025uncertaintycalibrated}}. Another example is ordinary differential equation solvers, which compute solutions at discrete time instances using adaptive step sizes \cite{gustafsson1988pi, zhao2021performance}.

Our work is grounded in conformal prediction (CP), a statistical method for data-driven uncertainty quantification \cite{angelopoulos2025theoreticalfoundationsconformalprediction}. CP was used to obtain discrete-time prediction regions for trajectories under deterministic sampling \cite{NEURIPS2021_312f1ba2, cleaveland2024conformalpredictionregionstime, sun2024copulaconformalpredictionmultistep,hashemi2024statistical,pmlr-v230-galvao-lopes24a}. These methods address the fundamental challenge that trajectories are dependent across time, thereby violating exchangeability assumptions typically needed to apply CP. However, this line of work is generally limited to discrete time domains. The authors in \cite{11007767,zhang2024safetycriticalcontrolofflineonlineneural} provide continuous time prediction regions using regularity properties of the underlying trajectories, but assume knowledge of regularity constants and do not permit random sampling. Recently, \cite{tassopoulou2025uncertaintycalibrated} introduced a method that extends \cite{NEURIPS2021_312f1ba2, cleaveland2024conformalpredictionregionstime, sun2024copulaconformalpredictionmultistep,hashemi2024statistical,pmlr-v230-galvao-lopes24a} for  randomly-sampled trajectories. However, \cite{tassopoulou2025uncertaintycalibrated} requires a pre-specified discrete-time grid and does not provide continuous-time prediction regions.

Motivated by these limitations, we provide a CP technique for obtaining continuous-time prediction regions under random sampling. Our method is fully data-driven and provides statistical guarantees without making any model assumptions. With this work, we address three fundamental challenges. Firstly, we extend statistical trajectory guarantees to continuous time domains, despite observing trajectories only at discrete, possibly sparse random time
points, eliminating the blind spots typically left by standard methods. Secondly, we introduce a calibration method to infer trajectory regularity from finite samples via downsampling high-frequency trajectories and leveraging a stochastic dominance assumption, enabling our method to be fully data-driven. Finally, we introduce different approaches to obtain valid prediction regions when sampling times are random.

\section{Related Work}
The aforementioned works from \cite{NEURIPS2021_312f1ba2, cleaveland2024conformalpredictionregionstime, sun2024copulaconformalpredictionmultistep,hashemi2024statistical,pmlr-v230-galvao-lopes24a} provide discrete-time prediction regions in batch settings, therefore requiring trajectory calibration data. In  online settings, the violation of exchangeability inherent to trajectory data is addressed by adaptive conformal inference \cite{NEURIPS2021_gibbs, angelopoulos2023conformal, pmlr-v162-zaffran22a}. Adaptive conformal inference builds prediction regions from a single trajectory by calibrating over past observation windows. While adaptive conformal inference guarantees asymptotic coverage, it does not provide simultaneous coverage and has not been studied for continuous-time trajectories \cite{NEURIPS2021_312f1ba2, cleaveland2024conformalpredictionregionstime, sun2024copulaconformalpredictionmultistep,hashemi2024statistical,pmlr-v230-galvao-lopes24a}. 

Recent work has integrated conformal prediction into safety-critical control  of continuous-time systems, see e.g., \cite{Hsu_2025, wei2025conformalcontractionrobustnonlinear}. However, a common assumption  is the availability of continuous-time calibration trajectories. We view our work as complementary as we provide continuous-time prediction regions without making this assumption. 

In addition, reachability analysis of stochastic systems has been studied using dynamic programming approaches \cite{esfahani2016stochastic,abate2008probabilistic}. While dynamic programming methods are exact, they are typically computationally expensive and require system knowledge, here knowledge of the underlying randomness. Other methods decoupled stochastic from deterministic parts of the trajectories, but typically make distributional assumptions, such as being sub-Gaussian, see e.g.,  \cite{jafarpourchen,liu2025safetyverificationnonlinearstochastic}, but also require knowledge of disturbance characteristics.

Beyond these works, Gaussian processes were used to  model inter-sample behavior, explicitly capturing uncertainty in between observations \cite{roberts2013gaussian, ensinger2024exact}. However, Gaussian processes impose strong structural assumptions and require distributional priors. Functional data analysis can deal with irregularly sampled trajectories \cite{gervini_irregular}  by introducing model parameters that smooth the trajectories in between observations, as well as for applying CP to continuous-time processes using projections into basis functions \cite{lei2013conformalpredictionapproachexplore}.

\section{Problem Formulation}\label{problem_formulation}

We consider continuous-time trajectories $X: [0, T^*] \to \mathcal{X} \subset \mathbb{R}^d$, where $X_t$ denotes the state of $X$ at time $t$. We assume that $X$ is observed at a set of discrete sampling times, denoted by $\boldsymbol{T} \coloneqq \{t_1, \dots, t_N\} \subset [0, T^*]$, where $N$ denotes the number of sampling times  which may itself be a random variable. We model the randomness as $(X, \boldsymbol{T}) \sim P_{X, \boldsymbol{T}}$.\arxiv{\footnote{While $T$ is standard notation for the prediction horizon, we write $T^*$ to avoid confusion with the set of sampling times $\boldsymbol{T}$.}}

To distinguish the continuous-time trajectory from our actual observations, we define the discretely sampled trajectory as $\boldsymbol{X} \coloneqq (X_{t_1}, X_{t_2}, \dots, X_{t_N})$. We can then compactly represent an observable data pair as $\boldsymbol{Z} \coloneqq (\boldsymbol{X}, \boldsymbol{T})\sim P_{\boldsymbol{Z}}$, where   the distribution $P_{\boldsymbol{Z}}$ is induced from $P_{X,\boldsymbol{T}}$.

Suppose that we have $m$ independent and identically distributed (i.i.d.) pairs $\{(X^{(i)}, \boldsymbol{T}^{(i)})\}_{i=1}^{m} \sim P_{X,\boldsymbol{T}}$. From these, we form the observable calibration dataset $\mathcal{D}_{\mathrm{cal}} := \{\boldsymbol{Z}^{(i)}\}_{i=1}^{m}$, where  $\boldsymbol{Z}^{(i)}:= (\boldsymbol{X}^{(i)}, \boldsymbol{T}^{(i)})$ with $\boldsymbol{X}^{(i)}:=(X_{t_1^{(i)}}, \dots, X_{t_{N_i}^{(i)}})$ and $\boldsymbol{T}^{(i)} := \{t_1^{(i)}, \dots, t_{N_i}^{(i)}\}$. Our objective is to use $\mathcal{D}_{\mathrm{cal}}$ to construct continuous-time prediction regions for a test trajectory $(X^{(m+1)}, \boldsymbol{T}^{(m+1)})\sim P_{X,\boldsymbol{T}}$.

Prior work, such as \cite{NEURIPS2021_312f1ba2, cleaveland2024conformalpredictionregionstime, sun2024copulaconformalpredictionmultistep,hashemi2024statistical,pmlr-v230-galvao-lopes24a}, constructs discrete-time prediction regions for fixed and deterministic sampling times $\boldsymbol{T}=\{t_1,\dots, t_N\}$. These works obtain prediction regions $\mathcal{C}_{t_k}$ such that, for any user-specified error rate $\delta\in (0,1)$,
\begin{align}\label{discrete_objective}
\mathbb{P}_{P_{X,\boldsymbol{T}}^{m+1}}\big(\boldsymbol{X}^{(m+1)}_{t_k}\in \mathcal{C}_{t_k} \text{ for all } k=1,\dots, N \big)\geq 1-\delta.
\end{align}
Our formulation pursues a stronger objective: building prediction regions over the  continuous-time domain $[0, T^*]$, and that even when the sampling times are random. Concretely, we seek to construct prediction regions $\mathcal{C}_t$  that satisfies
\begin{equation}\label{objective}
\mathbb{P}_{P_{X,\boldsymbol{T}}^{m+1}}\big(X^{(m+1)}_{t}\in \mathcal{C}_{t}\text{ for all } t\in [0,T^*] \big)\geq 1-\delta.
\end{equation}

\section{Prediction Regions under Deterministic Sampling and Known Lipschitz Constant}\label{problem_solution_section}
We begin with a simplified setting that we generalize later in Sections \ref{lipschitz_estimation_section} and \ref{random_sampling_times_section}. Firstly, we assume that the sampling times are deterministic and known, i.e., that $\boldsymbol{T}^{(i)}=\boldsymbol{T}$ for all $i\in[m+1]$ for a pre-specified $\boldsymbol{T}=\{t_1,\dots, t_N\}$.  Secondly, we assume to know a Lipschitz constant $L>0$ that describes the regularity properties of $X$ such that
\begin{equation}\label{assumption}
\mathbb{P}_{P_{X}}\Big(\sup_{0\leq t<s\leq T^*}\frac{|| X^{(m+1)}_{s}-  X^{(m+1)}_{t}||}{s-t} \leq L  \Big)=1. 
\end{equation}

\arxiv{Equation \eqref{assumption} characterizes the class of 
continuous-time trajectories for which our framework is suitable, thus excluding some standard continuous-time  
processes, such as Brownian motion. However, it is naturally 
satisfied by physical dynamical systems with bounded velocities, or dynamics as those governed by random or neural ODEs \cite{chen2018neural}.}


We proceed in two steps: (i) we obtain discrete-time prediction regions in \eqref{discrete_objective}, primarily following \cite{cleaveland2024conformalpredictionregionstime}, and (ii) we forward and backward propagate \eqref{discrete_objective} via $L$ to obtain \eqref{objective}.

\emph{Step 1. } Suppose that we are given a trajectory predictor $\hat{x}:[0,T^*]\to\mathcal{X}$, which could also be a function of the initial state $X_0$, and a normalizing function $\hat{\sigma}:[0,T^*]\to\mathbb{R}_{++}$, e.g., $\hat{x}_t$ could be a neural network predictor. Setting $\hat{\sigma}_t:=1$ provides no normalization, while optimal choice can be found in \cite{cleaveland2024conformalpredictionregionstime}. We also define the nonconformity score:
\begin{equation}\label{nonconformity}
S^{\mathrm{lf}}_{i}=s({X}^{(i)}, \boldsymbol{T}^{(i)})\coloneqq \max_{t\in \boldsymbol{T}^{(i)}}\Big\{ \frac{||{X}^{(i)}_{t}-\hat{x}_t||}{\hat{\sigma}_{t}}\Big\},
\end{equation} 
$\text{for each }i \in [m+1]$ where $[m+1]:=\{1,\dots, m+1\}$. The superscript $\mathrm{lf}$ stands for ``low frequency''. Given that $\boldsymbol{Z}^{(1)}, \hdots,\boldsymbol{Z}^{(m+1)}$ are i.i.d., it follows that $S^{\mathrm{lf}}_1,\dots, S^{\mathrm{lf}}_{m+1}$ are also i.i.d..   Following the conformal prediction paradigm, we compute  $\hat{q}_{1-\delta}\coloneqq\text{Quantile}(\{S^{\mathrm{lf}}_i\}_{i=1}^m;(1-\delta)(1+\frac{1}{m}))$, where $\text{Quantile}(\{S^{\mathrm{lf}}_i\}_{i\in [m]};(1-\delta)(1+\frac{1}{m}))$ denotes the $\lceil (m+1)(1-\delta) \rceil$-th smallest nonconformity score among $\{S^{\mathrm{lf}}_1,\dots, S^{\mathrm{lf}}_m\}$. By an immediate application of split CP (Theorem 3.2 from \cite{angelopoulos2025theoreticalfoundationsconformalprediction}) we have $\mathbb{P}_{P_{X,\boldsymbol{T}}^{m+1}}\Big(S^{\mathrm{lf}}_{m+1} \leq \hat{q}_{1-\delta}\Big)\geq1-\delta$. We can rewrite this statement as simultaneous prediction regions at the sampling times of the test trajectory:
\begin{equation}\label{epslions}
\mathbb{P}_{P_{X,\boldsymbol{T}}^{m+1}}\Big(||\mathbf{X}_{t}^{(m+1)}-\hat{f}_t||\leq \hat{r}_t \text{ for all }t\in \boldsymbol{T}^{(m+1)}\Big)\geq1-\delta,\end{equation}
where $\hat{r}_t\coloneqq\hat{\sigma}_t\hat{q}_{1-\delta}$ and  $\mathbb{P}_{P_{X,\boldsymbol{T}}^{m+1}}(\cdot)$ captures randomness over test and calibration data. Note that \eqref{epslions} provides  prediction regions $\mathcal{C}_t:=B_{\hat{r}_t}\big(\hat{f}_t\big)$ that satisfy \eqref{discrete_objective}, where $B_{\hat{r}_t}\big(\hat{f}_t\big)$ is a norm ball centered at the prediction $\hat{f}_t$ with radius $\hat{r}_t$.


\emph{Step 2. } We next propagate the prediction regions $\mathcal{C}_t$ forward and backward in time using the Lipschitz constant $L$. If the test trajectory is such that  $X^{(m+1)}_{t_k}\in\mathcal{C}_{t_k}$ at sampling time $t_k\in \boldsymbol{T}^{(m+1)}$ then, for any time $t\in (t_k,t_{k+1})$, we have
\begin{align}\label{lipschitz_propagation}
    ||X^{(m+1)}_{t}-\hat{f}_{t_k}|| 
    &= ||X^{(m+1)}_{t}-{X}^{(m+1)}_{{t}_{{k}}} +{X}^{(m+1)}_{{t}_{{k}}}-\hat{f}_{t_k}||\nonumber\\
    &\leq ||X^{(m+1)}_{t}-{X}^{(m+1)}_{{t}_{{k}}}||+||{X}^{(m+1)}_{{t}_{{k}}}-\hat{f}_{t_k}|| \nonumber\\
    &\leq L(t-{t}_{{k}})+\hat{r}_{t_{k}}=:R_{k}^{\mathrm{fwd}}(t).
\end{align}

\arxiv{
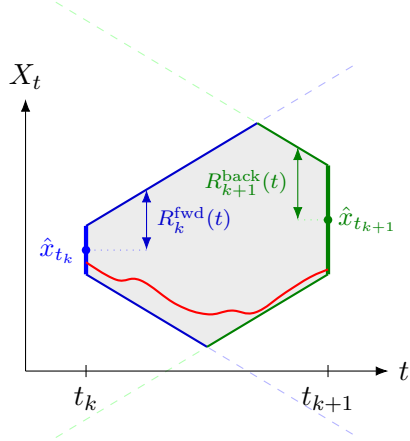
\begin{figure}[H]
\centering
\begin{tikzpicture}[scale=0.8, >=Latex]
\def\tk{1.5}        
\def\tkNext{5.5}    
\def\valK{2.0}      
\def\valNext{2.5}   

\def\rK{0.4}        
\def\rNext{0.9}     
\def\L{0.6}         


\def\tTop{4.33}
\def\yTop{4.1} 

\def\tBot{3.5}
\def\yBot{0.4} 

\fill[gray!15] 
    (\tk, \valK+\rK) -- (\tTop, \yTop) -- (\tkNext, \valNext+\rNext) -- 
    (\tkNext, \valNext-\rNext) -- (\tBot, \yBot) -- (\tk, \valK-\rK) -- cycle;

\draw[blue!30, dashed, thin] (\tk, \valK+\rK) -- (6.0, \valK+\rK + 2.7);
\draw[blue!30, dashed, thin] (\tk, \valK-\rK) -- (6.0, \valK-\rK - 2.7);

\draw[green!30, dashed, thin] (1.0, \valNext+\rNext + 2.7) -- (\tkNext, \valNext+\rNext);
\draw[green!30, dashed, thin] (1.0, \valNext-\rNext - 2.7) -- (\tkNext, \valNext-\rNext);

\draw[thick, blue!80!black] (\tk, \valK+\rK) -- (\tTop, \yTop);
\draw[thick, blue!80!black] (\tk, \valK-\rK) -- (\tBot, \yBot);

\draw[thick, green!50!black] (\tTop, \yTop) -- (\tkNext, \valNext+\rNext);
\draw[thick, green!50!black] (\tBot, \yBot) -- (\tkNext, \valNext-\rNext);

\def\tFwd{2.5}
\draw[<->, blue!80!black, thin] (\tFwd, \valK) -- (\tFwd, 3.0) 
    node[midway, right, font=\scriptsize] {$R^{\text{fwd}}_{k}(t)$};
\draw[dotted, blue!50] (\tFwd, \valK) -- (\tk, \valK);

\def\tBack{5.0}
\draw[<->, green!50!black, thin] (\tBack, \valNext) -- (\tBack, 3.7)
    node[midway, left, font=\scriptsize] {$R^{\text{back}}_{k+1}(t)$};
\draw[dotted, green!50] (\tBack, \valNext) -- (\tkNext, \valNext);
\draw[->] (0.5, 0) -- (6.5, 0) node[right] {$t$};
\draw[->] (0.5, 0) -- (0.5, 4.5) node[above] {$X_t$};

\draw (\tk, 0.1) -- (\tk, -0.1) node[below] {$t_k$};
\draw (\tkNext, 0.1) -- (\tkNext, -0.1) node[below] {$t_{k+1}$};
\draw[ultra thick, blue] (\tk, \valK-\rK) -- (\tk, \valK+\rK);
\fill[blue] (\tk, \valK) circle (2pt) node[left, font=\small] {$\hat{f}_{{t}_{k}}$};

\draw[ultra thick, green!50!black] (\tkNext, \valNext-\rNext) -- (\tkNext, \valNext+\rNext);
\fill[green!50!black] (\tkNext, \valNext) circle (2pt) node[right, font=\small] {$\hat{f}_{{t}_{k+1}}$};

\draw[red, thick] plot [smooth, tension=0.7] coordinates {
    (1.5, 1.8)      
    (2.0, 1.51)     
    (2.4, 1.5)      
    (3.0, 1.1)      
    (3.5, 0.94)     
    (3.9, 1.0)      
    (4.3, 0.97)     
    (4.8, 1.3)      
    (5.2, 1.55)      
    (5.5, 1.68)     
};
\end{tikzpicture}
\caption{Prediction region built with  the functions $R_{k}^{\mathrm{fwd}}(t)$, and  $R_{k+1}^{\mathrm{back}}(t)$.}
\label{tikz_fig}
\end{figure}
}
We can apply an analogous argument \emph{backwards} to obtain $R_{k+1}^{\mathrm{back}}(t)\coloneqq \hat{r}_{t_{k+1}}+L(t_{k+1}-t)$. Fig.\ref{tikz_fig} illustrates $R_{k}^{\mathrm{fwd}}(t)$ (blue cone) and $R_{k+1}^{\mathrm{back}}(t)$ (green cone). Importantly, their intersection will define our continuous-time prediction region in \eqref{objective}. Algorithm \ref{high_dimensional_predictions_alg} formalizes this intuition and takes care of some special cases. First, we compute the conformal radii $\hat{r}_{t_k}$ via the empirical quantile $\hat{q}$ (Lines 1–4). \cmmnt{Line 1 incorporates the exact initial state. if unavailable, we can simply initialize the construction of the region with $k=1$ and do a backward propagation: $\mathcal{C}_t\gets B_{R^{\text{back}}_1(t)}, t\in [0,t_1]$} Second, the algorithm performs a ``tightening pass'', replacing discrete-time prediction regions at time steps $t_{k}$ by the forward pass $R_{k+1}^{\mathrm{fwd}}(t_k)$ if it gives tighter prediction regions (Line 6). Finally, we construct the intersection between forward and backward cones $R_{k+1}^{\mathrm{fwd}}(t_k)$ and $R_{k+1}^{\mathrm{back}}(t_k)$ (Lines 10–12). We refer to the case studies for more illustrations. 

\begin{algorithm}
\small
\caption{Lipschitz-based prediction region}\label{high_dimensional_predictions_alg}
\hspace*{\algorithmicindent} \textbf{Inputs:} Calibration dataset $\mathcal{D}_{\text{cal}} \coloneqq \{(\boldsymbol{X}^{(i)},\boldsymbol{T}^{(i)})\}_{i=1}^m$, test times ${\boldsymbol{T}}^{(m+1)} := \{{t}_j\}_{j=1}^{{N}}$, error budget $\delta \in (0,1)$, $\hat{f}$, $\hat{\sigma}$, $L$, known initial value $X^{(m+1)}_0$.
\hspace*{\algorithmicindent}\\
\textbf{Output:} $\mathcal{C}$ prediction region for the test trajectory.
\begin{algorithmic}[1]
\State $t_0\gets0;\hat{f}_0\gets X^{(m+1)}_0; r_0\gets 0$
\State $S^{\mathrm{lf}}_i\gets \max_{t\in \boldsymbol{T}^{(i)}}\Big\{ \frac{||\boldsymbol{X}^{(i)}_{t}-\hat{x}_t||}{\hat{\sigma}_t}\Big\}$\Comment{for each $i\in [m]$}
\State$\hat{q}\gets \text{Quantile}(S^{\mathrm{lf}}_1,\dots, S^{\mathrm{lf}}_m;(1-\delta)(1+1/m))$
\State $ \hat{r}_{t_{k}}\gets\hat{\sigma}_{t_k}\hat{q} \text{ for } k=1,\dots, N$

\For{$k \text{ in } 0,\dots ,{N}-1$}\Comment{Check for tightening radius}
\If{$\hat{r}_{t_k}+L(t_{k+1}-t_k)\leq \hat{r}_{t_{k+1}}$}
\State $\hat{r}_{t_{k+1}}\gets \hat{r}_{t_k}+L(t_{k+1}-t_k)$
\State $\hat{f}_{t_{k+1}}\gets \hat{f}_{t_{k}}$
\EndIf
\EndFor
\State \text{Forward and backward radius for each $j\in\{0,\dots, {N}\}$}:
\Statex \hspace{\algorithmicindent} $R_j^{\text{fwd}}(t) \gets \hat{r}_{t_{j}} + L(t - {t}_j)$  
\Statex \hspace{\algorithmicindent}$R_j^{\text{back}}(t) \gets  \hat{r}_{t_{j}}+ L( {t}_{j}-t)$

\For{$k \text{ in } 0,\dots ,{N}-1$}
\Comment{Define interpolated prediction regions between sampling times}
\State  $\mathcal{C}_t\gets B_{R_k^{\text{fwd}}(t)}(\hat{x}_{{t}_k})\cap B_{R_{k+1}^{\text{back}}(t)}(\hat{x}_{{t}_k}), t\in (t_k,t_{k+1}]$
\EndFor
\State $\mathcal{C}_t\gets B_{R^{\text{fwd}}_{N}(t)}(\hat{x}_{t_N})$\quad \text{ for all }$t\in(t_{N},T^*]$
\State \Return $\mathcal{C}$
\end{algorithmic}
\end{algorithm}

To prove that Algorithm \ref{high_dimensional_predictions_alg}  achieves \eqref{objective}, define the events $\mathcal{E}_{\mathrm{disc}}:= \{ {X}_{t}^{(m+1)}\in B_{\hat{r}_t}(\hat{f}_t)\text{ for all }t\in {\boldsymbol{T}}^{(m+1)}\}$ (denoting discrete-time prediction regions), $\mathcal{E}_{\mathrm{lip}}:=\{||X^{(m+1)}_{t+s}-  X^{(m+1)}_{t}|| \leq s{L}  \text{ for all }s>0 \text{ and all }t \}$ (denoting Lipschitz bounds),  and  $\mathcal{E}_{\mathcal{C}}:= \{X_{t}^{(m+1)}\in \mathcal{C}_{t}\text{ for all }t\in [0,T^*]\}$ (denoting continuous-time prediction regions).

\begin{lemma}\label{lemma_algorithm_correctness}Let $\mathcal{C}$ denote the output of Algorithm \ref{high_dimensional_predictions_alg}. Then the previously defined events satisfy $\mathcal{E}_{\mathrm{disc}}\cap\mathcal{E}_{\mathrm{lip}}\subseteq\mathcal{E}_{\mathcal{C}}.$
\end{lemma}
\begin{proof}Let $t\in [0,T^*]$. Recall  $\boldsymbol{T}=\{t_1,t_2, \dots ,t_N\}$, inducing a partition of $[0,T^*]$ into disjoint sets: $[0,t_1)$, $[t_N,T^*]$, or $[t_k,t_{k+1})$ for $k\in[N-1]$. Without loss of generality, let $t\in [t_k,t_{k+1})$ for some $k\in [N-1]$ (other cases are analogous). On the event $\mathcal{E}_{\mathrm{disc}}\cap\mathcal{E}_{\mathrm{lip}}$, after applying the tightening pass, we have $X^{(m+1)}_{t_{k}}\in B_{\hat{r}_{t_k}}(\hat{x}_{t_k})$ and $X^{(m+1)}_{t_{k+1}}\in B_{\hat{r}_{t_{k+1}}}(\hat{x}_{t_{k+1}})$. Using $R_{k}^{\text{fwd}}(t)$ and $R_{k+1}^{\text{back}}(t)$ we conclude that $X^{(m+1)}_{t}\in B_{R_k^{\text{fwd}}(t)}(\hat{x}_{{t}_k})\cap B_{R_{k+1}^{\text{back}}(t)}(\hat{x}_{{t}_{k+1}})$, i.e., $X^{(m+1)}_{t}\in \mathcal{C}_t$.
\end{proof}

We are ready to show that Algorithm \ref{high_dimensional_predictions_alg} achieves \eqref{objective}.

\begin{theorem}\label{theorem_base}
    Suppose that $\mathbb{P}_{P_{X,\boldsymbol{T}}^{m+1}}\big(\mathcal{E}_{\mathrm{disc}}\big)\geq 1-\delta$ holds for $\delta>0$ and (\ref{assumption}) holds for a known value ${L}>0$. Then, the prediction regions $\mathcal{C}_t$ from Algorithm \ref{high_dimensional_predictions_alg} satisfy (\ref{objective}).
\end{theorem}

\begin{proof}
Recall the events $\mathcal{E}_{\mathrm{disc}}$ and $\mathcal{E}_{\mathrm{lip}}$. Then we have 
$\mathbb{P}_{P_{X,\boldsymbol{T}}^{m+1}}\big( X_{t}^{(m+1)} \in \mathcal{C}_{t}, \forall t \in [0,T^*] \big) 
\overset{(a)}{\geq} \mathbb{P}_{P_{X,\boldsymbol{T}}^{m+1}}(\mathcal{E}_{\mathrm{disc}} \cap \mathcal{E}_{\mathrm{lip}}) 
\overset{(b)}{\geq} \mathbb{P}_{P_{X,\boldsymbol{T}}^{m+1}}(\mathcal{E}_{\mathrm{disc}}) + \mathbb{P}_{P_{X,\boldsymbol{T}}^{m+1}}(\mathcal{E}_{\mathrm{lip}}) - 1 
\overset{(c)}{=} \mathbb{P}_{P_{\boldsymbol{Z}}^{m+1}}(\mathcal{E}_{\mathrm{disc}}) + \mathbb{P}_{P_X}(\mathcal{E}_{\mathrm{lip}}) - 1 
\overset{(d)}{\geq} 1-\delta$, 
where we used Lemma \ref{lemma_algorithm_correctness} in $(a)$, $(b)$ follows by a union bound, in $(c)$ we marginalize using that $P_{\boldsymbol{Z}}$ is induced by $P_{X, \boldsymbol{T}}$, and $(d)$ follows by coverage in the discrete points and our assumption on (\ref{assumption}) for $L>0$.
\end{proof}

\section{Prediction Regions with Estimated Trajectory Regularity Properties}\label{lipschitz_estimation_section}

We now relax the assumption of knowing a constant $L>0$ that satisfies \eqref{assumption}. Therefore, we assume that every trajectory $X^{(i)}\sim P_X$ is such that $L_{i}^{\mathrm{cont}}:=\sup\limits_{t,s\in[0,T^*], s\neq t} \frac{||{X}^{(i)}_t- {X}^{(i)}_s||}{|s-t|}$ exists. For instance, random ODEs \cite{hodgkinson2021stochastic}, as used in our experiments, provide almost surely continuously differentiable trajectories corresponding to solutions to initial value problems $\dot{X}^{(i)}=f(X^{(i)},t, \theta^{(i)})$ where $X^{(i)}_0\sim P_{X_0}$ and $\theta^{(i)}\sim P_{\theta}$ is a random parameter. In such a case, we can take $L_{i}^{\mathrm{cont}}=\max_{t\in[0,T^*]} ||\dot{X}^{(i)}_t ||$. To set the stage, we begin in an oracle setting for which, given an error rate $\alpha\in (0,1)$, we obtain an estimate $\hat{L}$ such that
\begin{equation}\label{lipschitz_estimation}
\mathbb{P}_{P_{X,\boldsymbol{T}}^{m+1}}\Big({L_{m+1}^{\mathrm{cont}} }\leq\hat{L}\Big)\geq 1-\alpha.
\end{equation}

\arxiv{\subsection{Oracle Setting}}
\conference{\textbf{Oracle Setting.} }Define the (generally unknown) maximum time derivative $L^{\mathrm{cont}}_i := \max_{t \in [0, T^*]} ||\dot{X}^{(i)}_t  ||$  for each $i\in [m+1]$. Using conformal prediction  $
\hat{L}_{\mathrm{oracle}} := \text{Quantile}\left(\{L^{\mathrm{cont}}_i\}_{i\in [m]}; (1-\alpha)(1 + 1/m)\right)$ immediately gives us $\mathbb{P}_{P^{m+1}_{X,\boldsymbol{T}}}(L^{\mathrm{cont}}_{m+1} \le \hat{L}_{\mathrm{oracle}} ) \ge 1-\alpha$. Hence, if we were to use $\hat{L}_{\mathrm{oracle}}$ instead of $L$ in Algorithm \ref{high_dimensional_predictions_alg}, we can get similar guarantees as in Theorem  \ref{theorem_base}.
\begin{theorem}\label{oracle_lipschitz_theorem} Suppose that $\mathbb{P}_{P_{X,\boldsymbol{T}}^{m+1}}\big(\mathcal{E}_{\mathrm{disc}}\big)\geq 1-\delta$ for $\delta>0$. Then, the prediction regions $\mathcal{C}_t$ from Algorithm \ref{high_dimensional_predictions_alg} using  $\hat{L}_{\mathrm{oracle}}$ satisfy
\begin{equation}
\mathbb{P}_{P^{m+1}_{X,\boldsymbol{T}}}\Big(X^{(m+1)}_t \in \mathcal{C}_t, \forall t \in [0, T^*] \Big)\geq 1-\alpha-\delta.
\end{equation}
\end{theorem}
\begin{proof}
Define the event $E_{\mathrm{lip}}: = \{ L_{m+1}^{\mathrm{cont}} \le \hat{L}_{\mathrm{oracle}}  \}$ and recall the event $\mathcal{E}_{\mathrm{disc}}$. Then 
\begin{align*}
&\mathbb{P}_{P^{m+1}_{X,\boldsymbol{T}}}\Big(X^{(m+1)}_t \in \mathcal{C}_t, \forall t \in [0, T^*] \Big) \overset{(a)}{\geq} \mathbb{P}_{P^{m+1}_{X,\boldsymbol{T}}}( E_{\mathrm{lip}} \cap \mathcal{E}_{\mathrm{disc}} ) \\
  &\overset{(b)}{\ge}  \mathbb{P}_{P^{m+1}_{X,\boldsymbol{T}}}( E_{\text{lip}} ) + \mathbb{P}_{P^{m+1}_{\boldsymbol{Z}}}( \mathcal{E}_{\mathrm{disc}} ) - 1 \overset{(c)}{\ge} 1-\alpha-\delta,
\end{align*}
where $(a)$ follows from Lemma \ref{lemma_algorithm_correctness} which holds similarly for $E_{\mathrm{lip}}$ instead of $\mathcal{E}_{\mathrm{lip}}$, $(b)$ is a union bound and marginalization on the second term, and in $(c)$ we use the coverage   guarantees from CP in each of the terms respectively. \end{proof}
In practice, however, we cannot compute  $L^{\mathrm{cont}}_i$. Nonetheless, this setting motivates our approach in which we assume to have high-frequency data that bridges the gap to $L^{\mathrm{cont}}_i$. 

\arxiv{\subsection{Gap Stochastic Dominance}}
\conference{\textbf{Gap Stochastic Dominance.}} Besides having access to trajectories sampled at low frequency, we  assume to have access to trajectories sampled at higher frequency. Concretely, our data is now generated via $(X^{(i)},\boldsymbol{T}^{(i)},\boldsymbol{T}_{\mathrm{full}}^{(i)})\overset{iid}{\sim}P_{X,\boldsymbol{T},\boldsymbol{T}_{\mathrm{\mathrm{full}}}}$ for $i\in [m+1]$. \arxiv{We note that the test trajectory does not need to be sampled at high frequency so that we could also write $(X^{(m+1)},\boldsymbol{T}^{(m+1)}){\sim}P_{X,\boldsymbol{T}}$, i.e., we only require high-frequency trajectories at calibration time but not  necessarily at test time, which is avoided here for simplicity.} This gives us the calibration dataset
\begin{equation}\label{calibration_dataset}
    \mathcal{D}_{\mathrm{cal}}\coloneqq \Big\{\Big(\boldsymbol{Z}^{(i)}, \boldsymbol{Z}^{(i)}_{\mathrm{full}}\Big)\Big\}_{i=1}^{m},
\end{equation}
where $\boldsymbol{Z}^{(i)}=(\boldsymbol{X}^{(i)}, \boldsymbol{T}^{(i)})$, $\boldsymbol{Z}^{(i)}_{\mathrm{full}}=( \boldsymbol{X}^{(i)}_{\mathrm{full}}, \boldsymbol{T}^{(i)}_{\mathrm{full}})$, and  $\boldsymbol{X}^{(i)}_{\mathrm{full}}$ denotes the $i-$th trajectory sampled at higher frequency with $\boldsymbol{T}^{(i)}_{\mathrm{full}}:=\{t^{(i)}_{1,\mathrm{full}}, \dots, t^{(i)}_{M_i,\mathrm{full}}\}$ for which we have $N_i<M_i$. For example, $\boldsymbol{Z}^{(i)}$ could be a downsampled version of $\boldsymbol{Z}^{(i)}_{\mathrm{full}}$, i.e. $\boldsymbol{T}^{(i)}\subseteq \boldsymbol{T}^{(i)}_{\mathrm{full}}$.

Next, we define  the maximum slope over the high-frequency and the low frequency trajectories $L^{\mathrm{full}}_i := \max\limits_{j\in \{1,\dots, M_i-1\}} ||X^{(i)}_{t_{j+1,\mathrm{full}}^{(i)}} - X^{(i)}_{t_{j,\mathrm{full}}^{(i)}}||/({t_{j+1,\mathrm{full}}^{(i)} -t_{j,\mathrm{full}}^{(i)}})$ and $L^{\mathrm{lf}}_i := \max_{j\in \{1,\dots, N_i-1\}} ||X^{(i)}_{t_{j+1}^{(i)}} - X^{(i)}_{t_j^{(i)}}||/({t_{j+1}^{(i)} -t_j^{(i)}})$, respectively. Note that $L^{\mathrm{full}}_i\leq L^{\mathrm{cont}}_i$  since $\boldsymbol{T}^{(i)}_{\mathrm{full}}\subset[0,T^*]$. Furthermore, define the hidden and the observable gaps $\Delta^{\mathrm{hidden}}_i := L^{\mathrm{cont}}_i - L^{\mathrm{full}}_i$ and $\Delta^{\mathrm{obs}}_i := L^{\mathrm{full}}_i - L^{\mathrm{lf}}_i$. The idea  is to use the observable gap $\Delta^{\mathrm{obs}}_i$ as a proxy of the hidden gap $\Delta^{\mathrm{hidden}}_i$, for which we make the next assumption.
\begin{assumption}\label{stochastic_dominance_assumption} For all $w\in \mathbb{R}$ and $i\in[m+1]$, it holds that $$\mathbb{P}_{\boldsymbol{Z},\boldsymbol{Z}_{\mathrm{full}}}(\Delta_i^{\mathrm{obs}} \leq w) \leq \mathbb{P}_{X, \boldsymbol{Z}_{\mathrm{full}}}(\Delta_i^{\mathrm{hidden}} \leq w).$$
\end{assumption}
 Assumption \ref{stochastic_dominance_assumption} is a stochastic dominance assumption which has previously appeared in the literature \cite{einbinder2024label}. Here, Assumption \ref{stochastic_dominance_assumption} postulates that the hidden gap is stochastically dominated by the observed gap, i.e.,  $\Delta^{\mathrm{hidden}}_i$ is more likely to take smaller values than $\Delta^{\mathrm{obs}}_i$. Intuitively, dense sampling accurately captures the maximum slope of the trajectory, yielding a small hidden gap $\Delta^{\mathrm{hidden}}_i$. Conversely, sparse sampling frequently misses sharp trajectory variations, causing a larger underestimation of the slope and a wider observable gap $\Delta^{\mathrm{obs}}_i$, see  Fig. \ref{dominance_figure_tikz}. Thus, it is reasonable to assume $\Delta^{\mathrm{hidden}}_i$ is stochastically dominated by $\Delta^{\mathrm{obs}}_i$.
\arxiv{In Appendix \ref{theoretical_results_dominance_section} we provide  justification supporting  \Cref{stochastic_dominance_assumption} by analyzing its satisfaction in terms of the sampling frequencies.}
 \arxiv{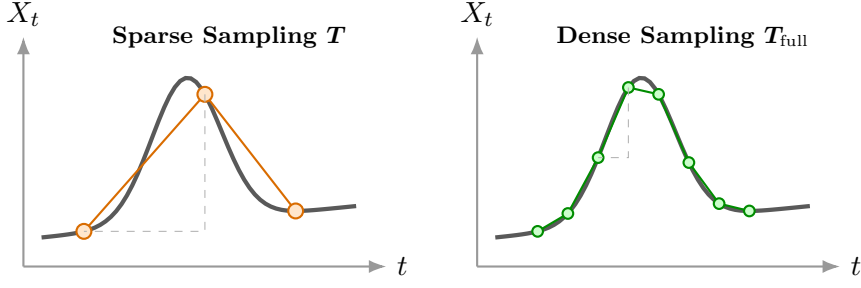
\begin{figure}[H]
\centering
\begin{tikzpicture}[scale=0.8, >=Latex]
\def\xStart{1.5}    
\def\xMid{3.5}      
\def\xEnd{5.0}      
\def\trajectory{0.4 + 2.4*exp(-1.5*(\x-3.2)^2) + 0.1*\x}

\begin{scope}
    \draw[->, gray!80, thick] (0.5, 0) -- (6.5, 0) node[right, text=black] {$t$};
    \draw[->, gray!80, thick] (0.5, 0) -- (0.5, 3.8) node[above, text=black] {$X_t$};
    \node[font=\footnotesize\bfseries] at (3.9, 3.8) {Sparse Sampling $\boldsymbol{T}$};

    \draw[gray!70!black, ultra thick, domain=0.8:6.0, samples=60] plot (\x, {\trajectory});

    \pgfmathsetmacro{\yStart}{0.4 + 2.4*exp(-1.5*(\xStart-3.2)^2) + 0.1*\xStart}
    \pgfmathsetmacro{\yMid}{0.4 + 2.4*exp(-1.5*(\xMid-3.2)^2) + 0.1*\xMid}
    \pgfmathsetmacro{\yEnd}{0.4 + 2.4*exp(-1.5*(\xEnd-3.2)^2) + 0.1*\xEnd}

    \draw[orange!85!black, thick] (\xStart, \yStart) -- (\xMid, \yMid);
    \draw[orange!85!black, thick] (\xMid, \yMid) -- (\xEnd, \yEnd);

    \filldraw[orange!85!black, fill=orange!20, thick] (\xStart, \yStart) circle (3.5pt);
    \filldraw[orange!85!black, fill=orange!20, thick] (\xMid, \yMid) circle (3.5pt);
    \filldraw[orange!85!black, fill=orange!20, thick] (\xEnd, \yEnd) circle (3.5pt);

    \draw[thin, color=gray!60, dashed] (\xMid, \yMid) -- (\xMid, \yStart) -- (\xStart, \yStart);
\end{scope}

\begin{scope}[shift={(7.5, 0)}]
    \draw[->, gray!80, thick] (0.5, 0) -- (6.5, 0) node[right, text=black] {$t$};
    \draw[->, gray!80, thick] (0.5, 0) -- (0.5, 3.8) node[above, text=black] {$X_t$};
    \node[font=\footnotesize\bfseries] at (3.9, 3.8) {Dense Sampling $\boldsymbol{T}_{\mathrm{full}}$};

    \draw[gray!70!black, ultra thick, domain=0.8:6.0, samples=60] plot (\x, {\trajectory});

    \def\prevX{\xStart}
    \pgfmathsetmacro{\prevY}{0.4 + 2.4*exp(-1.5*(\prevX-3.2)^2) + 0.1*\prevX}
    
    \foreach \step in {1,2,3,4,5,6,7} {
        \pgfmathsetmacro{\currX}{\xStart + \step*0.5}
        \pgfmathsetmacro{\currY}{0.4 + 2.4*exp(-1.5*(\currX-3.2)^2) + 0.1*\currX}
        
        \draw[green!55!black, thick] (\prevX, \prevY) -- (\currX, \currY);
        \filldraw[green!55!black, fill=green!20, thick] (\prevX, \prevY) circle (2.5pt);
        
        \xdef\prevX{\currX}
        \xdef\prevY{\currY}
    }
    
    \filldraw[green!55!black, fill=green!20, thick] (\xEnd, \prevY) circle (2.5pt);

    \pgfmathsetmacro{\xDenseA}{2.5}
    \pgfmathsetmacro{\yDenseA}{0.4 + 2.4*exp(-1.5*(\xDenseA-3.2)^2) + 0.1*\xDenseA}
    \pgfmathsetmacro{\xDenseB}{3.0}
    \pgfmathsetmacro{\yDenseB}{0.4 + 2.4*exp(-1.5*(\xDenseB-3.2)^2) + 0.1*\xDenseB}

    \draw[thin, color=gray!60, dashed] (\xDenseB, \yDenseB) -- (\xDenseB, \yDenseA) -- (\xDenseA, \yDenseA);
\end{scope}
\end{tikzpicture}
\caption{\textbf{Left:} Sparse sampling underestimating the Lipschitz constant. \textbf{Right:} Dense sampling giving accurate estimates.}
\label{dominance_figure_tikz}
\end{figure}
 \cmmnt{Analytically validating Assumption \ref{stochastic_dominance_assumption} is generally challenging. In Appendix \ref{theoretical_results_dominance_section}, we provide conditions on the subsampling frequency under which Assumption \ref{stochastic_dominance_assumption} holds.}}
 In our experiments, we subsample high-frequency trajectories to approximate the unobservable $\Delta_i^{\mathrm{hidden}}$ via a dense sampling-based surrogate $\hat{\Delta}_i^{\mathrm{hidden}}$,  and check if the empirical CDFs satisfy Assumption \ref{stochastic_dominance_assumption}. Generally, we observe that the sub-sampling frequencies are a tuning knob where higher ratios typically result in satisfaction of Assumption \ref{stochastic_dominance_assumption}.\arxiv{ Both empirical evidence and the theoretical results in Appendix \ref{theoretical_results_dominance_section} support this.}
 
Let  $\alpha_1,\alpha_2\in (0,1)$ be error rates. Next, apply conformal prediction to the observable quantities $L^{\mathrm{full}}_i$ and $\Delta^{\mathrm{obs}}_i$:
\begin{equation}\label{quantile_terms}
\begin{aligned}
q_{\mathrm{full}} &:= \text{Quantile}\left(\{L^{\mathrm{full}}_i\}_{i=1}^m; \left(1-\alpha_1\right)(1+\frac{1}{m})\right),\\
q_{\mathrm{obs}} &:= \text{Quantile}\left(\{\Delta^{\mathrm{obs}}_i\}_{i=1}^m; \left(1-\alpha_2\right)(1+\frac{1}{m})\right).
\end{aligned}
\end{equation}
Our estimated regularity parameter is now given by
\begin{equation}\label{predicted_L}
    \hat{L}_{\text{gap}} := q_{\mathrm{full}} + q_{\mathrm{obs}}.
\end{equation}

\begin{lemma}\label{lemma_gap_dominance}
    Let Assumption \ref{stochastic_dominance_assumption} hold and let $\alpha,\alpha_1,\alpha_2\in (0,1)$ be such that $\alpha_1+\alpha_2=\alpha$. Suppose $(X^{(i)},\boldsymbol{T}^{(i)},\boldsymbol{T}_{\mathrm{full}}^{(i)})\overset{iid}{\sim}P_{X,\boldsymbol{T},\boldsymbol{T}_{\mathrm{\mathrm{full}}}}$ for $i\in [m+1]$. Then,  $\hat{L}_{\text{gap}}$ from (\ref{predicted_L}) satisfies
\begin{equation}
\mathbb{P}_{P_{X,\boldsymbol{T}, \boldsymbol{T}_{\mathrm{full}}}^{m+1}}\big({L_{m+1}^{\mathrm{cont}} }\leq\hat{L}_{\text{gap}}\big)\geq 1-\alpha.
\end{equation}
\end{lemma}
\begin{proof}
    First, note that $
L^{\mathrm{cont}}_{m+1} = L^{\mathrm{full}}_{m+1} + \Delta^{\mathrm{hidden}}_{m+1}.$
Furthermore, we note that \conference{$\{L^{\mathrm{full}}_{m+1} \le q_{\mathrm{full}}\} \cap \{\Delta^{\mathrm{hidden}}_{m+1} \le q_{\mathrm{obs}}\}
\subseteq\{L^{\mathrm{full}}_{m+1} + \Delta^{\mathrm{hidden}}_{m+1} \leq q_{\mathrm{full}} + q_{\mathrm{obs}}\}.$}\arxiv{ $$\{L^{\mathrm{full}}_{m+1} \le q_{\mathrm{full}}\} \cap \{\Delta^{\mathrm{hidden}}_{m+1} \le q_{\mathrm{obs}}\}
\subseteq\{L^{\mathrm{full}}_{m+1} + \Delta^{\mathrm{hidden}}_{m+1} \leq q_{\mathrm{full}} + q_{\mathrm{obs}}\}.$$}
Therefore, we have  
\conference{\begin{align*}
&\mathbb{P}_{P_{X,\boldsymbol{T}, \boldsymbol{T}_{\mathrm{full}}}^{m+1}}\big(L^{\mathrm{cont}}_{m+1} \le \hat{L}_{\text{gap}}\big) \\
&\overset{(a)}{\geq} \mathbb{P}_{P_{X, \boldsymbol{Z},\boldsymbol{Z}_{\mathrm{full}}}^{m+1}}\big(\{L^{\mathrm{full}}_{m+1} \le q_{\mathrm{full}}\}\cap \{\Delta^{\mathrm{hidden}}_{m+1} \le q_{\mathrm{obs}}\} \big) \\
&\overset{(b)}{\ge} 1 - \Big[ \mathbb{P}_{P_{\boldsymbol{Z}_{\mathrm{full}}}^{m+1}}(L^{\mathrm{full}}_{m+1} > q_{\mathrm{full}}) \\
&\qquad \qquad + \mathbb{P}_{P_{X,\boldsymbol{Z},\boldsymbol{Z}_{\mathrm{full}}}^{m+1}}(\Delta^{\mathrm{hidden}}_{m+1} > q_{\mathrm{obs}}) \Big] \\
&\overset{(c)}{\geq} 1 - \Big[ \mathbb{P}_{P_{\boldsymbol{Z}_{\mathrm{full}}}^{m+1}}(L^{\mathrm{full}}_{m+1} > q_{\mathrm{full}}) \\
&\qquad \qquad + \mathbb{P}_{P_{\boldsymbol{Z},\boldsymbol{Z}_{\mathrm{full}}}^{m+1}}(\Delta^{\mathrm{obs}}_{m+1} > q_{\mathrm{obs}}) \Big] \overset{(d)}{\geq}  1 - \alpha,
\end{align*}}
\arxiv{
\begin{align*}
&\mathbb{P}_{P_{X,\boldsymbol{T}, \boldsymbol{T}_{\mathrm{full}}}^{m+1}}\big(L^{\mathrm{cont}}_{m+1} \le \hat{L}_{\text{gap}}\big) \\
&\overset{(a)}{\geq} \mathbb{P}_{P_{X, \boldsymbol{Z},\boldsymbol{Z}_{\mathrm{full}}}^{m+1}}\big(\{L^{\mathrm{full}}_{m+1} \le q_{\mathrm{full}}\}\cap \{\Delta^{\mathrm{hidden}}_{m+1} \le q_{\mathrm{obs}}\} \big) \\
&\overset{(b)}{\ge} 1 - \Big[ \mathbb{P}_{P_{\boldsymbol{Z}_{\mathrm{full}}}^{m+1}}(L^{\mathrm{full}}_{m+1} > q_{\mathrm{full}}) + \mathbb{P}_{P_{X,\boldsymbol{Z},\boldsymbol{Z}_{\mathrm{full}}}^{m+1}}(\Delta^{\mathrm{hidden}}_{m+1} > q_{\mathrm{obs}}) \Big] \\
&\overset{(c)}{\geq} 1 - \Big[ \mathbb{P}_{P_{\boldsymbol{Z}_{\mathrm{full}}}^{m+1}}(L^{\mathrm{full}}_{m+1} > q_{\mathrm{full}})+ \mathbb{P}_{P_{\boldsymbol{Z},\boldsymbol{Z}_{\mathrm{full}}}^{m+1}}(\Delta^{\mathrm{obs}}_{m+1} > q_{\mathrm{obs}}) \Big]\\
&\overset{(d)}{\geq}  1 - \alpha,
\end{align*}}
where we use monotonicity of the probability measure in $(a)$, in $(b)$ we use a union bound, in $(c)$ we use Assumption \ref{stochastic_dominance_assumption} and preservation of dominance under independent random thresholds\footnote{If $X\text{ and } Y$ are random variables whose CDF satisfy $F_X(z) \le F_Y(z)$ for all $z \in \mathbb{R}$ and $Q$ is independent of both, then by the tower rule: $\mathbb{P}_{X,Y,Q}(X\leq Q)= \mathbb{E}_{Q}[\mathbb{P}(X\leq Q|Q)]= \mathbb{E}_{Q}[F_{X}(Q)]\leq \mathbb{E}_{Q}[F_{Y}(Q)]=\mathbb{E}_{Q}[\mathbb{P}(Y\leq Q|Q)]=\mathbb{P}_{X,Y,Q}(Y\leq Q)$.}, and $(d)$ follows from conformal prediction.
\end{proof}
\conference{We can now use $\hat{L}_{\text{gap}}$ instead of $L$ in Algorithm \ref{high_dimensional_predictions_alg}. Due to  Lemma \ref{lemma_gap_dominance}, we get the same guarantees as in Theorem \ref{oracle_lipschitz_theorem}. \textcolor{red}{We refer the details to  \cite{alvarez2026continuous_conformal} due to space constraints.}}
\arxiv{
Once we have a valid regularity parameter for the test trajectory, we can use it as an input to Algorithm \ref{high_dimensional_predictions_alg} to get a continuous-time prediction region based on discrete-time prediction regions.}
\arxiv{
\begin{theorem}\label{Lipschitz_estimation_theorem} Let the assumptions of Lemma \ref{lemma_gap_dominance} hold, with fixed sampling times. Then, using $\hat{L}_{\text{gap}}$ from (\ref{predicted_L}) as an input to obtain the  prediction regions $\mathcal{C}_t$ from Algorithm \ref{high_dimensional_predictions_alg} we get
\begin{equation}
\mathbb{P}_{P^{m+1}_{X,\boldsymbol{T}, \boldsymbol{T}_{\mathrm{full}}}}\Big(X^{(m+1)}_t \in \mathcal{C}_t, \forall t \in [0, T^*] \Big)\geq 1-\alpha-\delta.
\end{equation}
\end{theorem}
\begin{proof}
We rely on the same ideas of Lemma \ref{lemma_algorithm_correctness} and Theorem \ref{theorem_base}. Taking $\tilde{\mathcal{E}}_{\mathrm{lip}}=\{L^{\mathrm{cont}}_{m+1} \le \hat{L}_{\text{gap}}\}$ in Lemma \ref{lemma_algorithm_correctness}, we also have that $\tilde{\mathcal{E}}_{\mathrm{lip}}\cap\mathcal{E}_{\mathrm{disc}}\subseteq \mathcal{E}_{\mathcal{C}}$, where here $\mathcal{C}$ is the output prediction region of Algorithm \ref{high_dimensional_predictions_alg}, and $\mathcal{E}_{\mathrm{disc}}$ denotes the discrete time prediction regions for the test trajectory in the sparse sampling times $\boldsymbol{T}^{(m+1)}$ (or alternatively, this naturally  also accommodates the setting where the discrete-time prediction regions are in the high-frequency times $\boldsymbol{T}_{\mathrm{full}}^{(m+1)}$). Hence, 
\begin{align*}
&\mathbb{P}_{P^{m+1}_{X,\boldsymbol{T},\boldsymbol{T}_{\mathrm{full}}}}\Big(X^{(m+1)}_t \in \mathcal{C}_t, \forall t \in [0, T^*] \Big)\\
&\overset{(e)}{\geq} \mathbb{P}_{P^{m+1}_{X,\boldsymbol{T}, \boldsymbol{T}_{\mathrm{full}}}}( \tilde{\mathcal{E}}_{\mathrm{lip}} \cap \mathcal{E}_{\mathrm{disc}} ) \\
  &\overset{(f)}{\ge}  \mathbb{P}_{P^{m+1}_{X,\boldsymbol{T}, \boldsymbol{T}_{\mathrm{full}}}}( \tilde{\mathcal{E}}_{\text{lip}} ) + \mathbb{P}_{P^{m+1}_{\boldsymbol{Z}}}( \mathcal{E}_{\mathrm{disc}} ) - 1\\
&\overset{(g)}{\ge} 1-\alpha-\delta,
\end{align*}
where $(e)$ follows from Lemma \ref{lemma_algorithm_correctness} which holds similarly for $\tilde{\mathcal{E}}_{\mathrm{lip}}$ instead of $\mathcal{E}_{\mathrm{lip}}$, $(f)$ is a union bound and marginalization on the second term, and in $(g)$ we use the coverage guarantees from CP, using the discrete time prediction regions and Lemma \ref{lemma_gap_dominance}. \end{proof}
}
\arxiv{
Next, we present an alternative stochastic dominance assumption to obtain a valid estimated regularity parameter.}
\arxiv{\subsection{Sum Stochastic Dominance}}

\conference{\textbf{Sum Stochastic Dominance.}} Instead of making a stochastic dominance assumption on the gaps, we now make a stochastic dominance assumption on $L_i^{\mathrm{full}}$ and the gaps. Recognizing that $L_i^{\mathrm{full}} \leq L_i^{\mathrm{cont}}$, we use $\Delta_i^{\mathrm{obs}}$ so that $L_i^{\mathrm{full}} + \Delta_i^{\mathrm{obs}}$ stochastically upper-bounds  $L_i^{\mathrm{cont}} = L_i^{\mathrm{full}} + \Delta_i^{\mathrm{hidden}}$.
\begin{assumption}\label{sum_stochastic_dominance_assumption}For all $w\in \mathbb{R}$ and $i\in[m+1]$, it holds that
\begin{equation}\label{alternative_sum}
\mathbb{P}_{\boldsymbol{Z},\boldsymbol{Z}_{\mathrm{full}}}( L_i^{\mathrm{full}}+\Delta_i^{\mathrm{obs}}\leq w) \le \mathbb{P}_{X, \boldsymbol{Z}_{\mathrm{full}}}(L_i^{\mathrm{full}}+\Delta_i^{\mathrm{hidden}}\leq  w).  
\end{equation}
\end{assumption}

Our estimated regularity parameter is in this case given by \begin{equation}\label{predicted_L_2}
    \hat{L}_{\mathrm{sum}}\coloneqq \text{Quantile}\big(\{L_{i}^{\mathrm{full}}+\Delta_{i}^{\mathrm{obs}}\}_{i=1}^{m};(1-\alpha)(1+1/m)\big)
\end{equation}
\begin{lemma} Let  Assumption \ref{sum_stochastic_dominance_assumption} hold and let $\alpha,\in (0,1)$. Then,  $\hat{L}_{\mathrm{sum}}$ from (\ref{predicted_L_2}) satisfies
\begin{equation}
\mathbb{P}_{P_{X,\boldsymbol{T}, \boldsymbol{T}_{\mathrm{full}}}^{m+1}}\Big({L_{m+1}^{\mathrm{cont}} }\leq\hat{L}_{\mathrm{sum}}\Big)\geq 1-\alpha.
\end{equation}
\end{lemma}
\begin{proof}
Similar to  Lemma \ref{lemma_gap_dominance}, hence omitted. 
\end{proof}
\arxiv{
From this Lemma it is possible to unify the valid predicted Lispchitz constant with Algorithm  \ref{high_dimensional_predictions_alg}, as we showed in Theorem \ref{Lipschitz_estimation_theorem}, we omit the details to avoid repetitive arguments.

In our experiments and in our discussion in Section \ref{discussion_section} we provide a comparison between using Assumption \ref{stochastic_dominance_assumption} and Assumption \ref{sum_stochastic_dominance_assumption}, and their underlying tradeoffs in terms of reliability and efficiency.
}

\section{Prediction Regions under Random Sampling }\label{random_sampling_times_section}
In practice, the test sampling times $\boldsymbol{T}^{(m+1)}$ may be unknown \emph{a priori}, which does not allow us to compute prediction regions as in (\ref{epslions}).  In this section, we build prediction regions without  knowledge of $\boldsymbol{T}^{(m+1)}$. We provide four methods that operate under different assumptions. 
\arxiv{
As a summary and  roadmap of this section, Table \ref{table_methods}  offers a guide to determine the best suited framework for a user's application, with a high-level schematic comparison of these methods in terms of their reliance on a Lipschitz constant, usage of synthetic sampling times, and their underlying sampling distribution requirements.
\begin{table}[H]
\centering
\caption{Methods for random sampling times.}
\label{table_methods}
\resizebox{\textwidth}{!}{%
\begin{tabular}{c|c|c|c}
\hline
Method & Requires Lipschitz constant & Reliance on imputed sampling times & Sampling distribution \\ \hline
Exact Test Sample Imputation    & \ding{51} & \ding{51} & $P_{X}P_{\boldsymbol{T}}$           \\
Calibrated Test Sample Imputation    & \ding{51} & \ding{51} & $P_{X,\boldsymbol{T},\tilde{\boldsymbol{T}} }$ \\
Stochastic Dominance  & \ding{55} & \ding{55} & $P_{X,\boldsymbol{T},\boldsymbol{T}_{\mathrm{full}}}$ \\
Prediction Envelope    & \ding{51} & \ding{55} & $P_{X,\boldsymbol{T}}$           \\ \hline
\end{tabular}%
}
\end{table}

From Table \ref{table_methods} we highlight that each of the methods has nuances and different underlying assumptions that we will present in this section. For example, the Exact Test Sample Imputation requires access to samples from $P_{\boldsymbol{T}}$ to obtain synthetic sampling times and assumes that the trajectories are independent of the sampling times, which allows us to write the joint  sampling distribution $P_{X,\boldsymbol{T}}$ as $P_{X}P_{\boldsymbol{T}}$. As a complement to this table, 
in Fig. \ref{decision_flowchart} we present a decision chart to guide the reader with the most suitable method according to the use case and the properties of the sampling regime and the trajectories. 
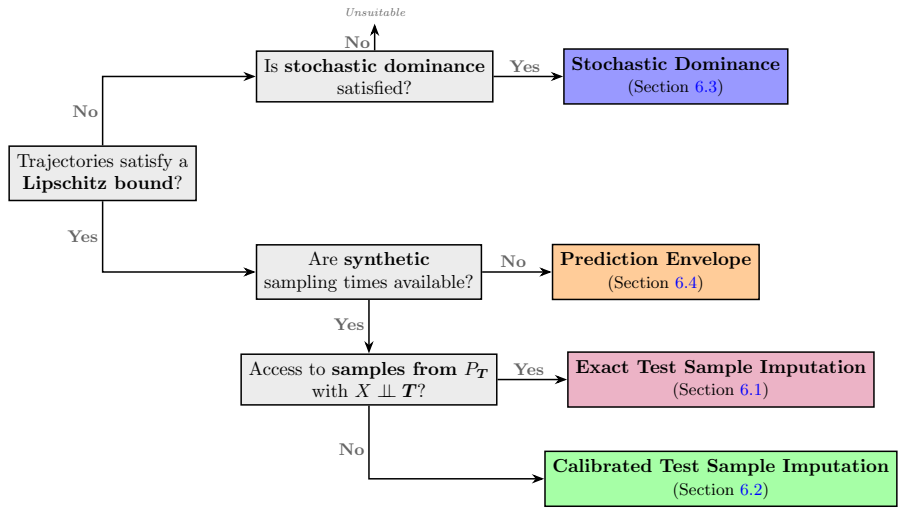
\begin{figure}[H]
\centering
\resizebox{0.75\textwidth}{!}{
\begin{tikzpicture}[
    >=Stealth,
    node distance=0.8cm and 1.2cm,
    qbox/.style={
        draw, 
        rectangle, 
        thick, 
        fill=gray!15, 
        align=center, 
        font=\small, 
        minimum height=0.9cm, 
        minimum width=3.0cm,
        inner sep=4pt
    },
    methodbox/.style={
        draw, 
        rectangle, 
        thick, 
        align=center, 
        font=\small\bfseries, 
        minimum height=0.9cm, 
        minimum width=3.2cm,
        inner sep=4pt
    },
    branch/.style={
        font=\footnotesize\bfseries, 
        text=gray!80!black,
        inner sep=2pt
    }
]
    \node[qbox] (q1) {Trajectories satisfy a\\ \textbf{Lipschitz bound}?};

    \node[qbox, above right=0.8cm and 1.1cm of q1] (q_stoch) {Is \textbf{stochastic dominance}\\ satisfied?};

    \node[qbox, below right=0.8cm and 1.1cm of q1] (q_extra) {Are \textbf{synthetic}\\ sampling times available?};

    \node[methodbox, fill=blue!40, right=1.3cm of q_stoch] (m_stoch) {
        Stochastic Dominance\\
        \mdseries\footnotesize(Section \ref{stochastic_dominance_section})
    };

    \node[methodbox, fill=orange!40, right=1.3cm of q_extra] (m_envelope) {
        Prediction Envelope\\
        \mdseries\footnotesize(Section \ref{prediction_envelope_section})
    };

    \node[qbox, below=1.0cm of q_extra] (q_impute_type) {Access to \textbf{samples from $P_{\boldsymbol{T}}$}\\ with $X \perp \!\!\! \perp \boldsymbol{T}$?};

    \node[methodbox, fill=purple!30, right=1.3cm of q_impute_type] (m_test_impute) {Exact 
        Test Sample Imputation\\
        \mdseries\footnotesize(Section \ref{test_sample_imputation_section})
    };

    \node[methodbox, fill=green!35, below=0.8cm of m_test_impute] (m_calib_impute) {
        Calibrated Test Sample Imputation\\
        \mdseries\footnotesize(Section \ref{calibrated_imputations_section})
    };


    \draw[->, thick] (q1.north) |- (q_stoch.west) 
        node[branch, pos=0.25, left] {No};

    \draw[->, thick] (q1.south) |- (q_extra.west) 
        node[branch, pos=0.25, left] {Yes};

    \draw[->, thick] (q_stoch.east) -- (m_stoch.west) 
        node[branch, pos=0.45, above] {Yes};

    \draw[->, thick] (q_stoch.north) -- ++(0,0.5) 
        node[branch, pos=0.3, left] {No}
        node[above, font=\tiny\itshape, text=gray!80!black, align=center] {Unsuitable};

    \draw[->, thick] (q_extra.east) -- (m_envelope.west) 
        node[branch, pos=0.45, above] {No};

    \draw[->, thick] (q_extra.south) -- (q_impute_type.north) 
        node[branch, pos=0.45, left] {Yes};

    \draw[->, thick] (q_impute_type.east) -- (m_test_impute.west) 
        node[branch, pos=0.45, above] {Yes};

    \draw[->, thick] (q_impute_type.south) |- (m_calib_impute.west) 
        node[branch, pos=0.3, left] {No};
\end{tikzpicture}
}
\caption{Decision flowchart for choosing a continuous-time prediction region method based on assumptions and sampling availability.}
\label{decision_flowchart}
\end{figure}}

\arxiv{\subsection{Exact Test Sample Imputation}\label{test_sample_imputation_section}}
\conference{\textbf{Exact Test Sample Imputation.}} While the specific realization of $\boldsymbol{T}^{(m+1)}$ is unknown at calibration, suppose we can sample from the generating distribution $P_{\boldsymbol{T}}$. For example, this fits settings where accurate models of the sampling distribution are available, such as Bernoulli processes \cite{packet_loss, packet_loss_bernoulli}. Under an independence assumption, we can impute these times to build valid conformal regions. 
\begin{proposition}\label{theorem_independence}
Let $\tilde{\boldsymbol{T}}^{(m+1)}\sim {P}_{\boldsymbol{T}}$ be a synthetic set of sampling times. Suppose that the trajectory and the sampling times are independent, that is,  $P_{X,\boldsymbol{T}}=P_{X}P_{\boldsymbol{T}}$. Then
\begin{equation*}
\mathbb{P}_{P_{X,\boldsymbol{T}}^{m+1}}\Big(||X^{(m+1)}_{\tilde{t}}-\hat{x}_{\tilde{t}}||\leq \hat{q}_{1-\delta}\hat{\sigma}_{\tilde{t}}, \forall \tilde{t}\in\tilde{\boldsymbol{T}}^{(m+1)} \Big)\geq 1-\delta,
\end{equation*}
where $\hat{q}_{1-\delta}\coloneqq\text{Quantile}(S^{\mathrm{lf}}_1,\dots, S^{\mathrm{lf}}_m;(1-\delta)(1+\frac{1}{m}))$.
\end{proposition}
\begin{proof}
Let $\tilde{S}^{\mathrm{lf}}_{m+1}:=s({X}^{(m+1)}, \tilde{\boldsymbol{T}}^{(m+1)})$.
Given that independence holds, $({{X}}^{(m+1)}, \tilde{\boldsymbol{T}}^{(m+1)})$ is i.i.d. with the calibration data. Hence, $\tilde{S}^{\mathrm{lf}}_{m+1}$ is i.i.d. with the scores of the calibration data $\{S^{\mathrm{lf}}_{i}\}_{i\in [m]}$. Therefore, we can apply split CP with the imputation times: $\mathbb{P}_{P_{X,\boldsymbol{T}}^{m+1}}\Big( ||X^{(m+1)}_{\tilde{t}}-\hat{x}_{\tilde{t}}||\leq \hat{q}_{1-\delta}\hat{\sigma}_{\tilde{t}} \forall {\tilde{t}}\in\tilde{\boldsymbol{T}}^{(m+1)} \Big)=\mathbb{P}_{P_{X,\boldsymbol{T}}^{m+1}}\Big( \tilde{S}^{\mathrm{lf}}_{m+1}\leq \hat{q}_{1-\delta} \Big)\geq 1-\delta$.
\end{proof}
Proposition \ref{theorem_independence} generalizes the fixed sampling times setting, where the imputation trivially reduces to the deterministic times. Under these conditions, we can apply an argument analogous to Theorem \ref{theorem_base} to construct valid continuous-time prediction regions. In practice, this simply requires using the synthetic times $\tilde{\boldsymbol{T}}^{(m+1)}$ as input in Algorithm \ref{high_dimensional_predictions_alg} in place of the unobserved $\boldsymbol{T}^{(m+1)}$.
\arxiv{
Note that this method is non-conservative, since it leverages nonconformity scores that do not rely on the Lipschitz constant for obtaining discrete-time prediction regions. However, we highlight that the assumptions of this method are strong:  an independence between sampling times and the trajectory, and access to samples of $P_{\boldsymbol{T}}$. For this reason, we refer to this method as an oracle in our experiments. Next, we propose a method that copes with these two limitations, while still relying on imputations.}

\arxiv{\subsection{Calibrated Test Sample Imputation}\label{calibrated_imputations_section}}

\conference{\textbf{Calibrated Test Sample Imputation.}} Consider $(X^{(i)}, \boldsymbol{T}^{(i)},\tilde{\boldsymbol{T}}^{(i)})\overset{iid}{\sim}P_{X,\boldsymbol{T},\tilde{\boldsymbol{T}}}$ for all $ i \in  [m+1]$.
Here, $P_{\tilde{\boldsymbol{T}}}$ denotes a known distribution from which we can impute test sampling times. We can think of  $\tilde{\boldsymbol{T}}^{(i)}$ as a model of ${\boldsymbol{T}}^{(i)}$. To correct for a possible mismatch between $\tilde{\boldsymbol{T}}^{(i)}$  and ${\boldsymbol{T}}^{(i)}$, define the nonconformity scores:
$$S_i^* = \max_{\tilde{t} \in \tilde{\boldsymbol{T}}^{(i)}} \Bigg\{ \frac{\min\limits_{t \in \boldsymbol{T}^{(i)}} \big( || X^{(i)}_t - \hat{x}_{\tilde{t}}|| + L|\tilde{t} - t| \big)}{\hat{\sigma}_{\tilde{t}}} \Bigg\}.$$

\begin{theorem}Suppose that
(\ref{assumption}) holds for a known value ${L}>0$. Let $\hat{q}:=\text{Quantile}(\{S_j^*\}_{j=1}^{m};(1-\delta)(1+1/m))$. Then
    $$\mathbb{P}_{P_{X,\boldsymbol{T}, \tilde{\boldsymbol{T}}}^{m+1}}\Big(X^{(m+1)}_{\tilde{t}} \in B_{\hat{q}\hat{\sigma}_{\tilde{t}}}(\hat{f}_{\tilde{t}}) ,  \quad \forall \tilde{t} \in \tilde{\boldsymbol{T}}^{(m+1)}\Big) \geq 1 - \delta.$$
\end{theorem}
\begin{proof}
    By split CP, we have 
$\mathbb{P}_{P_{X,\boldsymbol{T}, \tilde{\boldsymbol{T}}}^{m+1}}(S_{m+1}^* \leq \hat{q}) \geq 1 - \delta.$
We can rewrite the argument inside the probability so that for all $ \tilde{t} \in \tilde{\boldsymbol{T}}^{(m+1)}$, we have
$$\frac{\min\limits_{t \in \boldsymbol{T}^{(m+1)}} \left( \left|\left| X^{(m+1)}_t - \hat{x}_{\tilde{t}} \right|\right| + L|\tilde{t} - t| \right)}{\hat{\sigma}_{\tilde{t}}} \leq \hat{q}.$$
Using the triangle inequality and the constant $L$, we infer that, for any imputation  time $\tilde{t}\in \tilde{\boldsymbol{T}}^{(m+1)}$ and any time $t\in \boldsymbol{T}^{(m+1)}$:
\begin{equation*}
    \begin{split}
        ||X^{(m+1)}_{\tilde{t}} - \hat{x}_{\tilde{t}}||& \leq ||X^{(m+1)}_{\tilde{t}} - X^{(m+1)}_{t}|| + ||X^{(m+1)}_{t} - \hat{x}_{\tilde{t}}||\\
        &\leq L|\tilde{t} - t| + ||X^{(m+1)}_{t} - \hat{x}_{\tilde{t}}||.
    \end{split}
\end{equation*}
Because this holds for all $t\in \boldsymbol{T}^{(m+1)}$, the bound holds for the minimum over $\boldsymbol{T}^{(m+1)}$.
Therefore, with probability $1-\delta$:
\begin{equation*}
    \begin{split}
    ||X^{(m+1)}_{\tilde{t}}- \hat{f}_{\tilde{t}}|| &\leq \min_{t \in \boldsymbol{T}^{(m+1)}} \left( || X^{(m+1)}_{t} - \hat{f}_{\tilde{t}}|| + L|\tilde{t} - t| \right)\\
&\leq \hat{q}\hat{\sigma}_{\tilde{t}}.
    \end{split}
\end{equation*}\end{proof}
\arxiv{
We remark that in an oracle setting where the imputations $\tilde{\boldsymbol{T}}^{(i)}$ perfectly predict the true sampling  times ${\boldsymbol{T}}^{(i)}$ we recover the nonconformity scores from Section \ref{problem_solution_section}. Formally, if $\tilde{\boldsymbol{T}}={\boldsymbol{T}}$ almost surely, then $S_i^*=S^{\mathrm{lf}}_i$ for all $i\in [m+1]$ almost surely. In this case, the discrete-time nonconformity scores would not rely on the the Lipschitz constant, thus producing a tight conformal quantile.}

Once we have valid prediction regions at the imputed times, we can simply propagate the Lipschitz constant using Algorithm \ref{high_dimensional_predictions_alg} to obtain a valid prediction region over the entire continuous domain $[0,T^*]$. This requires knowledge of $L>0$  satisfying ($\ref{assumption}$). \arxiv{We can again estimate $L$ following Section \ref{lipschitz_estimation_section}, but have to be more careful now as $L$ appears also in the nonconformity score $S_i^*$. This will also be the case for prediction envelopes which we present in Section \ref{prediction_envelope_section} and for which we show in  Appendix \ref{exchangeability_correction_section} how to jointly estimate $L$ while dealing with $L$ appearing in the nonconformity score. We can apply the ideas from Appendix \ref{exchangeability_correction_section} similarly for calibrated imputations and omit them here for brevity.}
\conference{\textcolor{red}{Replicating the procedure from the Appendix of \cite{alvarez2026continuous_conformal} we can unify the prediction region of this section with the methods from Section \ref{lipschitz_estimation_section}.}}
\cmmnt{We highlight that the advantage is that we do not make any assumption about independence between the trajectories and the sampling times. However, we pay with the worst-case Lipchitz constant twice: once, absorbed by the quantile, via the nonconformity scores, and then for propagating the Lipschitz constant over the entire continuous domain. Nonetheless, when the sampling times are very dense, we can get non-conservative coverage as we show in our experiments using samples of the empirical distribution of sampling times as imputations.} 

\arxiv{\subsection{Stochastic Dominance Method}\label{stochastic_dominance_section}}

\conference{\textbf{Stochastic Dominance Method.}} Recall that we have assumed access to densely sampled trajectories in the setting of Section \ref{lipschitz_estimation_section}. We can leverage these high-frequency trajectories by using the stochastic dominance assumptions directly on the nonconformity scores instead of the  trajectories.

We define the normalized residual $R_i(t) \coloneqq ||X^{(i)}_t - \hat{f}_t||/\hat{\sigma}_t$. We distinguish between the continuous, dense and sparse versions of the nonconformity score by $S^{\mathrm{cont}}_i \coloneqq \sup\limits_{t \in [0, T^*]} R_i(t)$, $S^{\mathrm{full}}_i \coloneqq \max\limits_{t \in \boldsymbol{T}_{\mathrm{full}}^{(i)}} R_i(t)$, and $S^{\mathrm{lf}}_i \coloneqq \max\limits_{t \in \boldsymbol{T}^{(i)}} R_i(t)$, respectively, where the latter two are observable. 

In an oracle setting where $S^{\mathrm{cont}}$ is observable, the threshold $\hat{Q}^{\mathrm{oracle}}_{1-\delta} \coloneqq \text{Quantile}(S^{\mathrm{cont}}_{1},\dots, S^{\mathrm{cont}}_{m}; (1-\delta)(1+1/m))$ guarantees validity by standard split CP. Thus, the region $\mathcal{C}_t = B_{\hat{Q}^{\mathrm{oracle}}_{1-\delta}\hat{\sigma}_t}(\hat{f}_t)$ satisfies 
\begin{equation}
  \mathbb{P}_{P_{X,\boldsymbol{T}, \boldsymbol{T}_{\mathrm{full}}}^{m+1}}\big(\ ||X^{(m+1)}_t-\hat{f}_t||\leq \hat{Q}^{\mathrm{oracle}}_{1-\delta}\hat{\sigma}_t,   \forall t\in [0,T^*]\big)\geq 1-\delta.  
\end{equation}
This is an oracle, and we can replicate the same ideas and assumptions that we had for the Lipschitz constants, but directly on the nonconformity scores.

Intuitively, if we have
$S^{\mathrm{lf}}_{i}\approx {S}^{\mathrm{cont}}_{i}$ for all $i\in [m+1]$, we can expect that using split CP taking $\hat{Q}^{\mathrm{lf}}_{1-\delta}\coloneqq\text{Quantile}\big(\{S^{\mathrm{lf}}_i\}_{i=1}^{m};(1-\delta)(1+1/m)\big)$, we can obtain a prediction region like that of the oracle. However, notice that $S_{i}^{\mathrm{lf}}\leq S_{i}^{\mathrm{cont}}$ almost surely for every $i\in [m]$, hence $\hat{Q}^{\mathrm{lf}}_{1-\delta}< \hat{Q}^{\mathrm{oracle}}_{1-\delta}$, which means that using $\hat{Q}^{\mathrm{lf}}_{1-\delta}$ is not reliable to obtain a prediction region in $[0,T^*]$. 

Adapting Assumption \ref{stochastic_dominance_assumption} to the context of nonconformity scores, we take  $\boldsymbol{\Delta}^{\textbf{hidden}}_i := S^{\mathrm{cont}}_i - S^{\mathrm{full}}_i$. This is the unseen high-frequency error, whereas we can observe $\boldsymbol{\Delta}^{\textbf{obs}}_i := S^{\mathrm{full}}_i - S^{\mathrm{lf}}_i$. This is an observable lower frequency error. The interpretation and the validation of the nonconformity score gap stochastic dominance assumption is completely analogous to its Lipschitz counterpart. The same is true for the sum stochastic dominance assumption.   We give the explicit argument for the gap stochastic dominance assumption for clarity. We do not expand further for the sum stochastic dominance case, to avoid repetitive ideas.

For a given $\delta\in (0,1)$, let $\delta_1,\delta_2\in (0,1)$ be such that $\delta_1+\delta_2=\delta$. We define the following two empirical quantiles:
\begin{align*}
    \hat{Q}_{\mathrm{full}} &:= \text{Quantile}\big(\{S^{\mathrm{full}}_i\}_{i=1}^m; \left(1-\delta_1\right)(1+\frac{1}{m})\big),\\
    \text{and } \hat{Q}_{\mathrm{obs}} &:= \text{Quantile}\left(\{\boldsymbol{\Delta}^{\textbf{obs}}_i\}_{i=1}^m; \big(1-\delta_2\right)(1+\frac{1}{m})\big).
\end{align*}
Define the estimated gap dominance-based quantile
\begin{equation}\label{eq:q_gap}
    \hat{Q}_{\text{gap}} :=  \hat{Q}_{\mathrm{full}}  +  \hat{Q}_{\mathrm{obs}}.
\end{equation}

\begin{theorem} Let Assumption \ref{stochastic_dominance_assumption} hold with $\boldsymbol{\Delta}^{\textbf{hidden}}_i$ and $\boldsymbol{\Delta}^{\textbf{obs}}_i$ and let $\delta,\delta_1,\delta_2\in (0,1)$ be such that $\delta_1+\delta_2=\delta$. Then, $\hat{Q}_{\text{gap}}$ from (\ref{eq:q_gap}) satisfies
\begin{equation}\label{new_inflation_theorem}
 \mathbb{P}_{P_{X,\boldsymbol{T}, \boldsymbol{T}_{\mathrm{full}}}^{m+1}}\Big(||X^{(m+1)}_t-\hat{x}_t||\leq \hat\sigma_t \hat{Q}_{\text{gap}}, \forall t\in [0,T^*]\Big)\geq 1-\delta.
\end{equation}
\end{theorem}
\begin{proof}
Note that (\ref{new_inflation_theorem}) is equivalent to $\mathbb{P}_{P_{X,\boldsymbol{T}, \boldsymbol{T}_{\mathrm{full}}}^{m+1}}\big({S_{m+1}^{\mathrm{cont}} }\leq\hat{Q}_{\text{gap}}\big)\geq 1-\delta$, which follows by an
 argument similar to that of Lemma \ref{lemma_gap_dominance}, hence we omit the remaining details.
\end{proof}
\arxiv{The advantage of this prediction region is that there is no need to know the distribution of the sampling times. Furthermore, the argument works even for trajectory-dependent sampling regimes, and does not require a Lipschitz constant.} \arxiv{Note that the method can also be suitable for fixed sampling times, since nothing inherently of this method requires the sampling times to be random.}

We can use completely analogous ideas to obtain a prediction region based on a sum stochastic dominance assumption directly on the nonconformity scores. Given that the ideas are replication of the previous subsections, as mentioned, we omit the details.

\arxiv{\subsection{Prediction Envelope}\label{prediction_envelope_section}}
\conference{\textbf{Prediction Envelope}.} \arxiv{When imputations of sampling times are not available and a stochastic dominance assumption directly on the nonconformity scores is not suitable, we can still get valid continuous-time prediction regions, as we will show next.}
At any time $t$, let $D^{(i)}_t := \min\limits_{t^{(i)}_k\in \boldsymbol{T}^{(i)}} \left( \|X^{(i)}_{t^{(i)}_k} - \hat{f}_t\| + L|t - t^{(i)}_k| \right)$ denote an error bound for the $i-$th trajectory. Using the Lipschitz property (\ref{assumption}), and triangle inequality, it holds that for any time $t\in [0, T^*]$, $||X^{(i)}_{t} - \hat{x}_t||\leq D^{(i)}_t$. 
This motivates the nonconformity scores $\bar{S}_i := \max\limits_{t \in [0, T^*]} \{\frac{D^{(i)}_t}{\hat{\sigma}_t}\}$ from which we  obtain valid prediction regions.

\begin{theorem}\label{envelope_theorem}
Suppose that
(\ref{assumption}) holds for a known value ${L}>0$. Then, the prediction regions $\mathcal{C}_t = \left\{ x \in \mathbb{R}^d : \|x - \hat{f}_t\| \leq \hat{q} \cdot \hat{\sigma}_t \right\}$ with $\hat{q} := \text{Quantile}(\{\bar{S}_i\}_{i=1}^m;( 1-\delta)(1+\frac{1}{m}))$ satisfy (\ref{objective}).\end{theorem}
\begin{proof}
    By definition of $D^{(i)}_t$ and since the constant $L$ is known, we infer that $\|X^{(i)}_t - \hat{f}_t\| \leq D^{(i)}_t$ for any trajectory $i$ and any time $t\in [0,T^*]$. 
    Therefore, we know that $$\frac{\|X^{(i)}_t - \hat{x}_t\|}{\hat{\sigma}_t} \leq \frac{D^{(i)}_t}{\hat{\sigma}_t} \leq \sup_{s} \frac{D^{(i)}_s}{\hat{\sigma}_s} = \bar{S}_i.$$
    If $\bar{S}_{m+1} \leq \hat{q}$, which happens with probability at least $1-\delta$ by split CP, then it holds that, for any $t\in [0,T^*]$, $\frac{\|X^{(m+1)}_t - \hat{f}_t\|}{\hat{\sigma}_t} \leq \hat{q}$. With monotonicity of the probability measure we conclude the proof.
\end{proof}
\arxiv{In Theorem \ref{envelope_theorem} we assumed knowledge of $L>0$. In Appendix \ref{exchangeability_correction_section} 
we extend this method to unify the prediction region of this section with the methods from Section \ref{lipschitz_estimation_section}.}
\conference{In the Appendix of our extended version of the paper \cite{alvarez2026continuous_conformal}
we unify the prediction region of this section with the methods from Section \ref{lipschitz_estimation_section}.}

\section{Experiments}\label{experiments_section}
\arxiv{In this section we present experiments in which we use Euler's method for discretizing ordinary differential equations.   We illustrate the prediction regions obtained with our main algorithm, study the satisfaction of Assumption \ref{stochastic_dominance_assumption} and Assumption \ref{sum_stochastic_dominance_assumption}, and compare the prediction regions under random sampling. Additional experiments are included in Appendix \ref{additional_experiments}, which is structured as follows. In Appendix \ref{dmps_experiments} we provide another comparison of the random sampling times methods using dynamic movement primitives. Finally, we narrow down the experiments to study the stochastic dominance assumptions on the nonconformity scores as well as ablations for the estimation of Lipschitz constants of a system in higher dimensions in Appendix \ref{attractor_experiment} and Appendix \ref{lorenz_experiment}, respectively.}
\conference{We present experiments using Euler's method. Additional experiments where we solve random ODEs with BDF discretization methods, as well as experiments  using dynamic movement primitives are presented in the Appendix in \cite{alvarez2026continuous_conformal}.}

In the reminder, we consider the following dynamical system which represents the movement of an object in one dimension:
\begin{align*}
    \dot{x}&=v\\
    \dot{v}&=a+d,
\end{align*}
where $x$, $v$,  $a$ and $d$ denote position, velocity,  acceleration, and a piecewise continuous stochastic disturbance, respectively. 

Applying Euler's method, we obtain a first order approximation of the continuous-time  system as
\begin{equation*}
\begin{split}
x_{t_{k+1}}&=x_{t_{k}}+v_{t_{k}}(t_{k+1}-t_k)\\
v_{t_{k+1}}&=v_{t_{k}}+a_{t_{k}}(t_{k+1}-t_k) +d_{t_k}{(t_{k+1}-t_k)}.
    \end{split}
\end{equation*}

This approximation provides a discrete-time signal of the underlying continuous-time system.  We set the following parameters:
$a_{t_k}=sin(t_k)$ for acceleration and $d_{t_k}\overset{i.i.d.}{\sim}U(-3,3)$ as a disturbance for times $t_k<5$ and $d_{t_k}\overset{i.i.d.}{\sim}U(-5,5)$ for times $t_k\geq5$. 
In order to guarantee the existence of a constant $L>0$ satisfying (\ref{assumption}) we clip the generated velocities to $L=1/100$, so that $v(t)\in [-L, L]$ for all $t\geq0$, leading to a known value satisfying (\ref{assumption}). We learn point predictors from training data for $\hat{x}$ and $\hat{\sigma}$ by fitting cubic splines to $1,000$ training trajectories.
\subsection{Deterministic Sampling Times}\label{experiments1_deterministic_sampling}
We begin with a fixed-sampling times setup, taking $\boldsymbol{T}=\{1,2,\dots, 9\}$. In  Fig. \ref{running_exmaples} we show in gray two examples of the  prediction regions from Section \ref{problem_solution_section} outputted by Algorithm \ref{high_dimensional_predictions_alg} for this experiments. 

\begin{figure}[H]
\centering\includegraphics[width=0.9\linewidth]{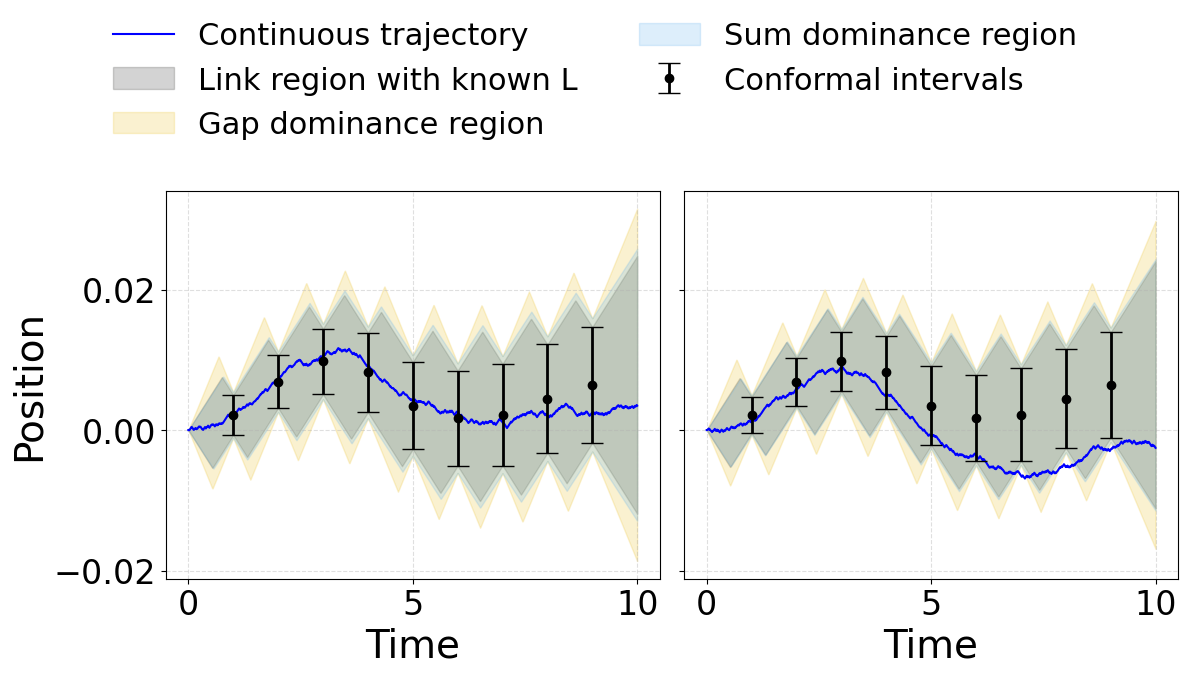}
    \caption{Prediction regions in two instances with fixed sampling times with a target coverage of $1-\delta=0.9$. }
\label{running_exmaples}
\end{figure}
We now validate the two different methods from Section \ref{lipschitz_estimation_section}. Their prediction regions are also shown in Fig. \ref{running_exmaples}. In Fig. \ref{comparison_frequencies_empirical_cdf}, we first present empirical CDFs of the gaps to verify the gap dominance assumption from Assumption \ref{stochastic_dominance_assumption}. Therefore, we use $\boldsymbol{T}_{\mathrm{full}}$ at different frequencies. We  let $h$ denote the distance between consecutive sampling times that partition the interval $[0,10]$ used to construct $\boldsymbol{T}_{\mathrm{full}}=\{t_{0,\mathrm{full}},t_{1,\mathrm{full}}, \dots,t_{K,\mathrm{full}} \}$ with $t_{j,\mathrm{full}} =jh, j=0,1\dots, K$, and $h=T^*/K$. In Fig. \ref{comparison_frequencies_cdf_sum} we present the corresponding empirical CDFs for validation of the sum dominance assumption from Assumption \ref{sum_stochastic_dominance_assumption}.  Shaded regions represent $2$ standard deviation across 10 independent calibration samples with $m=200$ each.
\arxiv{We refer the reader to  \cref{appendix_empirical_validation}  for more details on our empirical validation of the stochastic dominance assumptions.}
\begin{figure}[H]
\includegraphics[width=0.99\linewidth]{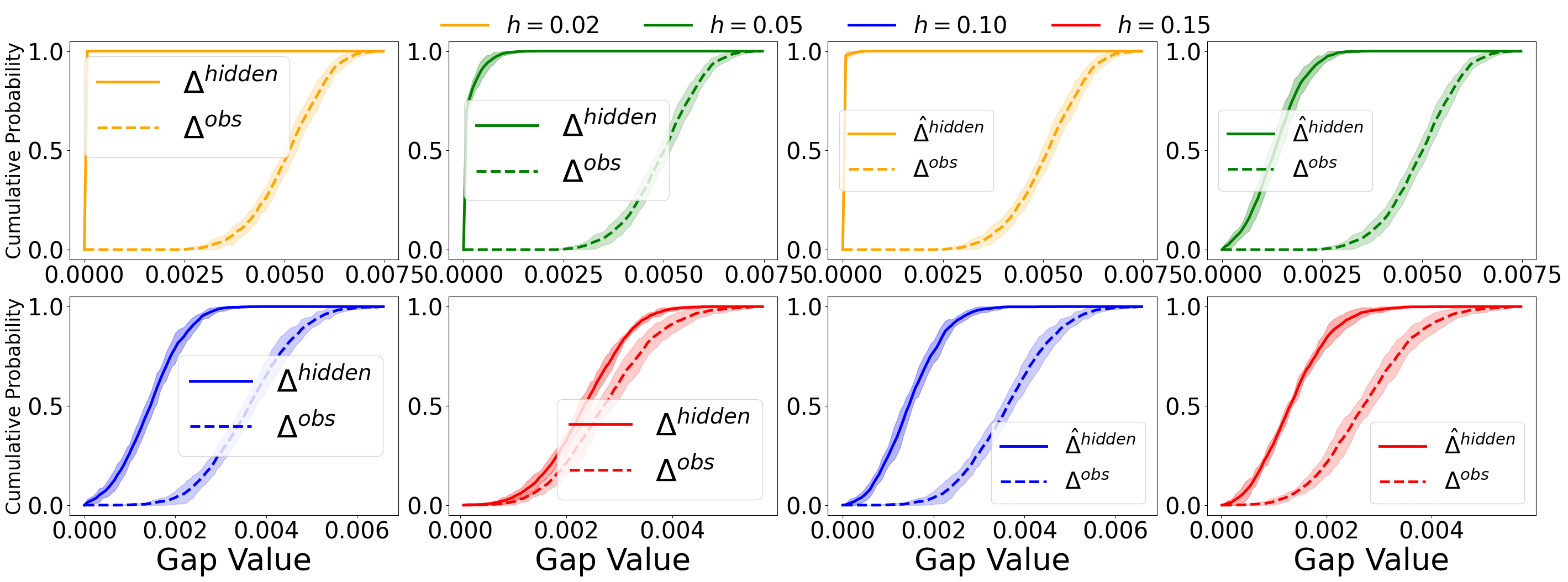}
    \caption{Comparison of gap empirical CDFs for different downsampling frequencies with oracle access to continuous trajectories (left) and validation of gap dominance assumption via downsampling to obtain surrogates $\hat{\Delta}^{\mathrm{hidden}}$ (right).}
\label{comparison_frequencies_empirical_cdf}
\end{figure}
 Comparing the left and right plots in Fig. \ref{comparison_frequencies_empirical_cdf}, we note that our subsampling approach to diagnose Assumption \ref{stochastic_dominance_assumption} is a faithful representation of the gaps that we see on an oracle setting with access to continuous trajectories. We observe a tradeoff: larger values of $h$, say $h=0.15$ provide tighter Lipschitz estimates (as can be seen in Fig.\ref{comparison_gap_sum_dominance}), but can compromise the satisfaction of Assumption \ref{stochastic_dominance_assumption}. \arxiv{In Fig. \ref{comparison_frequencies_cdf_sum} we highlight that the CDF of the continuous Lipschitz constant has a discontinuous jump at $L=1/100$; this happens because all the trajectories achieve the clipped maximum velocity, hence a point mass is formed at $L=1/100$.}

\begin{figure}[H]
\centering
\includegraphics[width=0.99\linewidth]{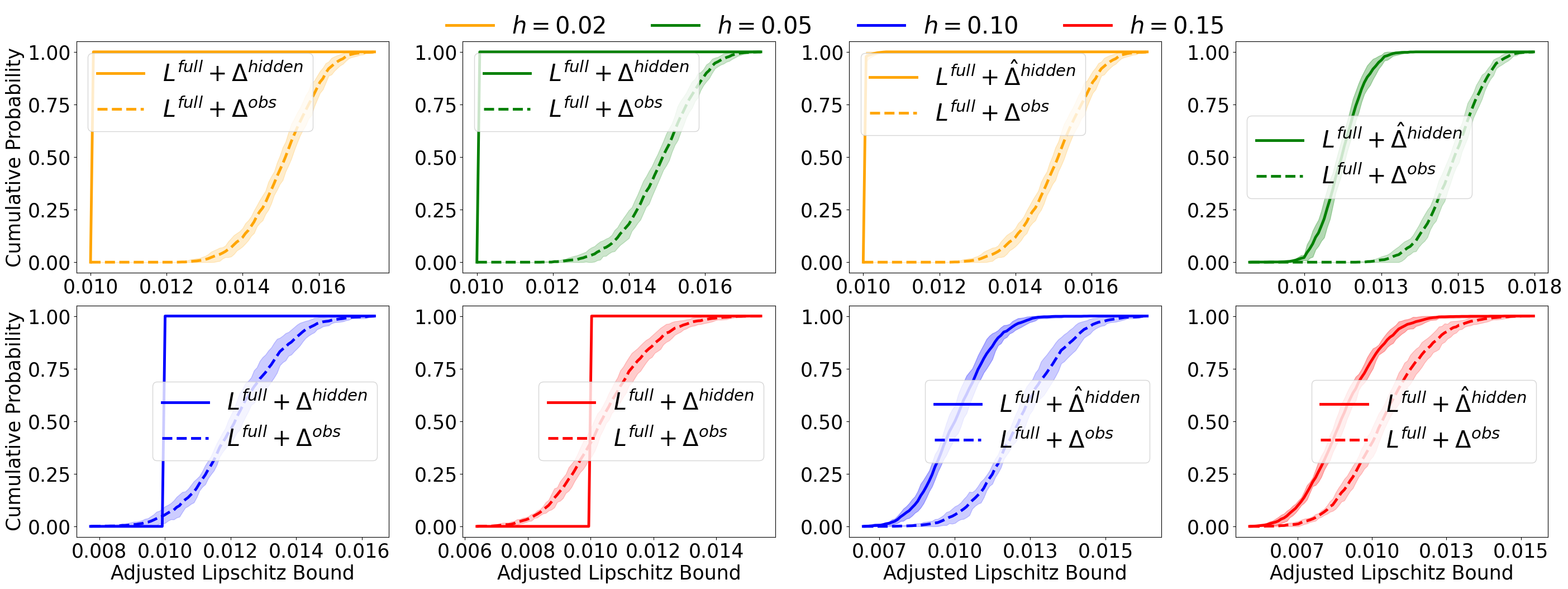}
\caption{CDFs with oracle access to continuous trajectories (left) and validation of sum dominance assumption via subsampling to obtain surrogates $\hat{\Delta}^{\mathrm{hidden}}$ (right).}
\label{comparison_frequencies_cdf_sum}
\end{figure}
To obtain Fig. \ref{comparison_gap_sum_dominance}, we generated 100 independent predictions of the Lipschitz constant of a test trajectory  based on 200 calibration trajectories for each estimate $\hat{L}_{\text{gap}}$ and $\hat{L}_{\text{sum}}$, corresponding to the conformal quantiles (\ref{predicted_L}) and (\ref{predicted_L_2}), respectively. We typically observe $\hat{L}_{\text{sum}} \leq \hat{L}_{\text{gap}}$ because splitting the error budget $\alpha$ makes the gap approach more sensitive to tail behavior.  
\begin{figure}[H]
\centering\includegraphics[width=0.5\linewidth]{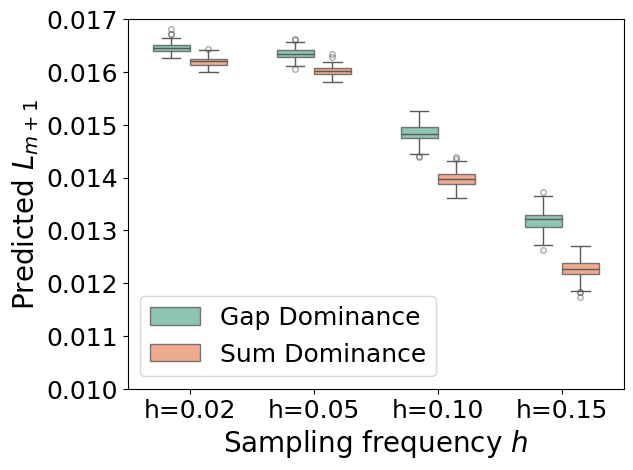}
    \caption{Comparison of the predicted Lipschitz constants}
\label{comparison_gap_sum_dominance}
\end{figure}
With high frequency sampling ($h=0.02$) we prioritize robustness of satisfying the stochastic dominance assumption (recall Fig. \ref{comparison_frequencies_empirical_cdf} and Fig. \ref{comparison_frequencies_cdf_sum}), at the expense of conservatism  in the estimation (see Fig. \ref{comparison_gap_sum_dominance}). Ideally, we want a sampling frequency with the tightest possible estimation such that the stochastic dominance assumption is satisfied.

A key concept in conformal prediction which we use to evaluate our framework is that of coverage, which in our case refers to the proportion of instances in which the prediction region includes the test trajectory. This is evaluated in reference to $1-\delta$, which plays the role as the target coverage. Intuitively, if the realizations of coverage are on average below this target, we consider the prediction region to be unreliable, whereas if the attained average coverage is overly larger than the target, the method is considered conservative. To generate Fig. \ref{coverages_lipschitz_running_example}, we sample 100 calibration datasets of size $m=200$ and evaluate 200 continuous test trajectories for each. We set the target coverage rate to $1-\delta=0.90$. We allocate an error budget of $\delta/2=0.05$ to the discrete-time conformal intervals and $0.05$ to the Lipschitz constant prediction (which is further split in half to compute the gap dominance bounds). Overall, in this case coverage is not sensitive to the  estimated Lipschitz constant, due to the fact that all our methods safely overestimate $L=1/100$, so the final coverage is attributed to the conformal discrete-time prediction regions. The average coverage under estimated Lipchitz constants is $0.95$.
\begin{figure}[H]
\centering\includegraphics[width=0.9\linewidth]{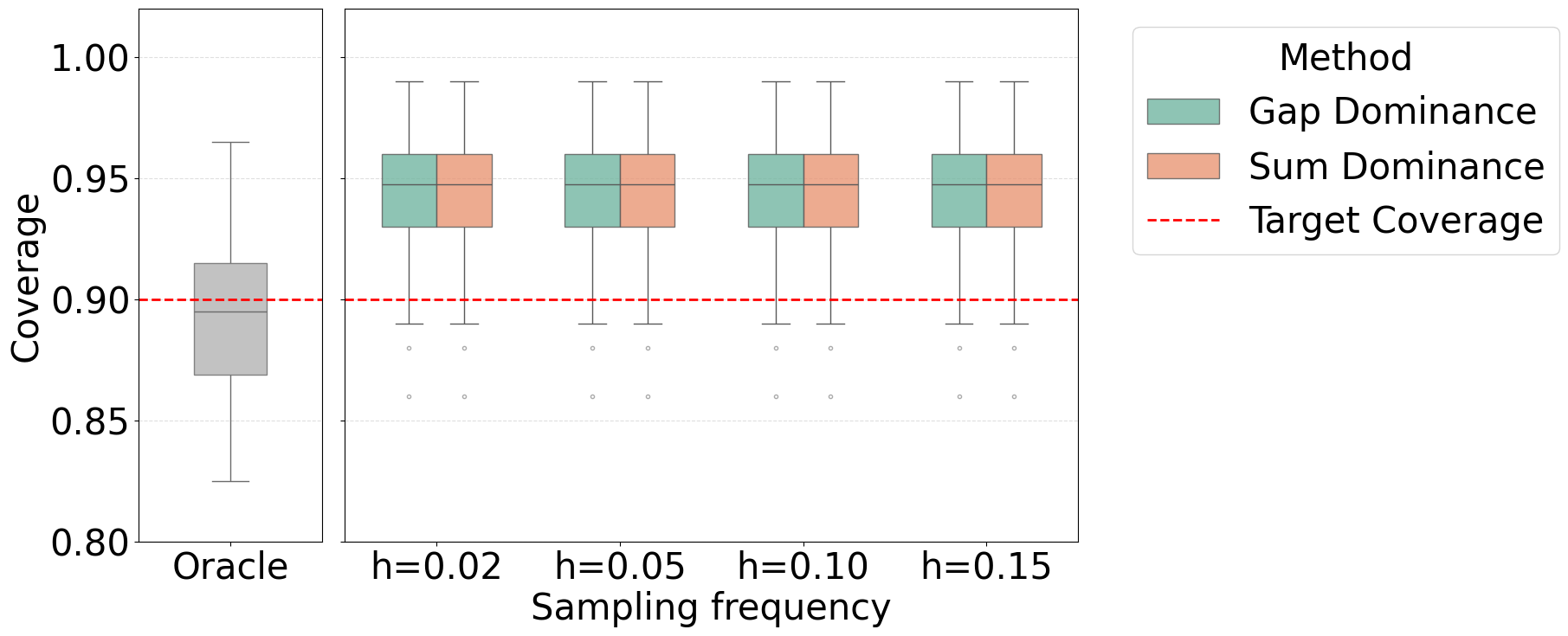}
    \caption{Coverage using fixed sampling times, at different sampling frequencies for Lipschitz estimation. The oracle corresponds to Algorithm \ref{high_dimensional_predictions_alg} with knowledge of $L=0.01$.}
\label{coverages_lipschitz_running_example}
\end{figure}
\subsection{Random Sampling Times}
To obtain random sampling times, a sample from $P_{\boldsymbol{T}}$ is generated as follows: we take $N$
 i.i.d. samples from a uniform distribution in $[0,10]$, independently from the trajectories. By picking different values of $N\in \{50,80,100, 200\}$ our experiments consider setups with different sparsity levels. 

In Fig. \ref{examples_random_sampling_runnign_example} we present instances of the prediction regions we obtain for all methods from Section \ref{random_sampling_times_section} at different sampling sparsities. We see that the prediction regions get tighter for denser sampling times (larger $N$), which is also reflected in the coverage in Fig. \ref{compare_running_example}. For the coverage evaluation in Fig. \ref{compare_running_example},  we generate 100 calibration datasets with $m=200$ and evaluate each with 200 continuous test trajectories. 

\arxiv{Here the oracle corresponds to the Exact Test Sample Imputation due to its strong assumptions.} We take as imputations $\tilde{T}$ draws from a fixed empirical distribution of sampling times $P_{\tilde{T}}$ across each of the sampling sparsity levels\arxiv{, so we have $(X^{(i)}, \boldsymbol{T}^{(i)},\tilde{\boldsymbol{T}}^{(i)})\sim P_{X}P_{\boldsymbol{T}}P_{\tilde{\boldsymbol{T}}}$}. The gap and sum dominance regions are from applying the stochastic dominance on the nonconformity scores. For the methods requiring a Lipschitz constant we use $L=0.01$.
\begin{figure}[H]
    \centering
\includegraphics[width=0.8\linewidth]{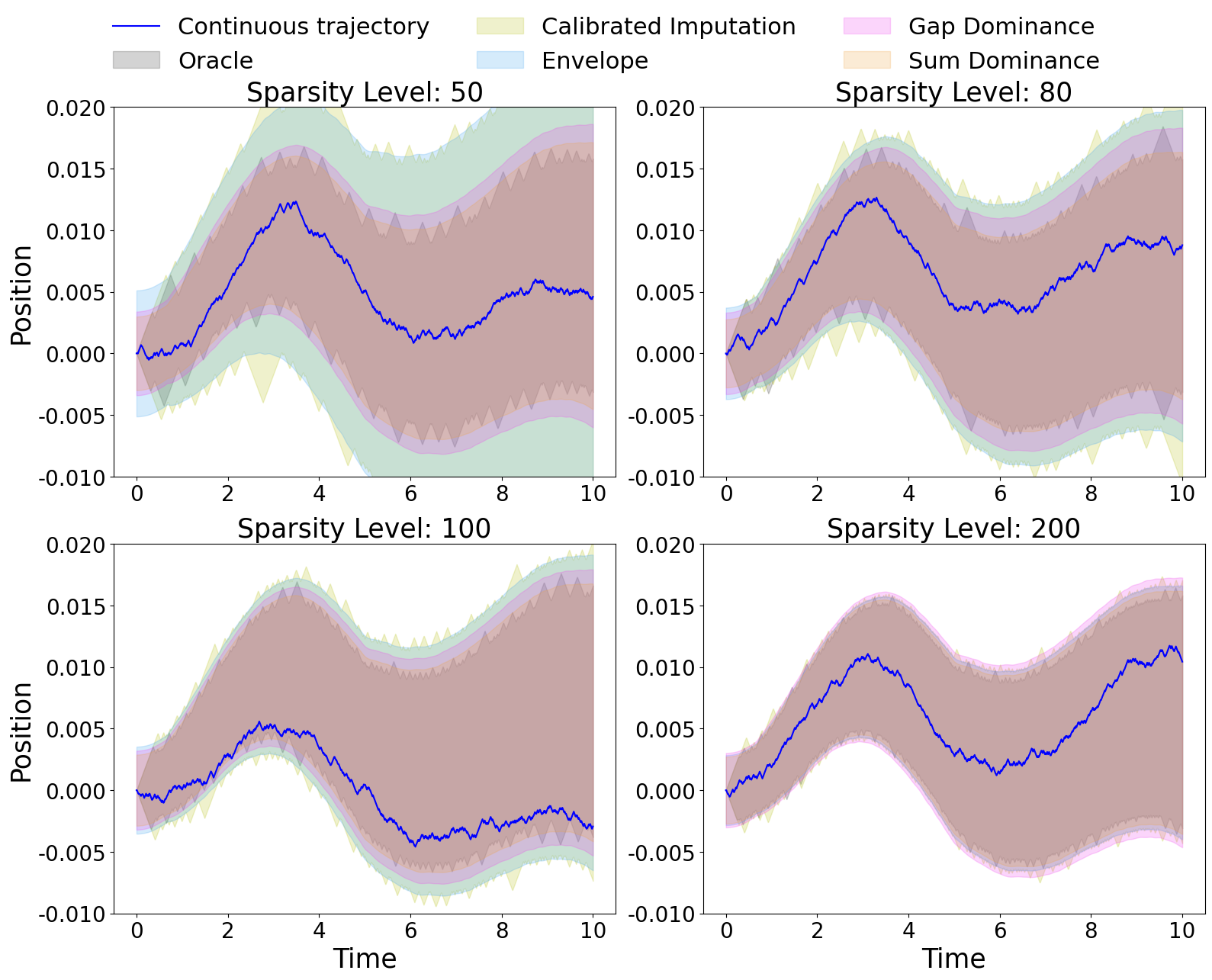}
\caption{Prediction regions by number of uniform random samples.}
\label{examples_random_sampling_runnign_example}
\end{figure}

\begin{figure}[H]
    \centering
\includegraphics[width=0.99\linewidth]{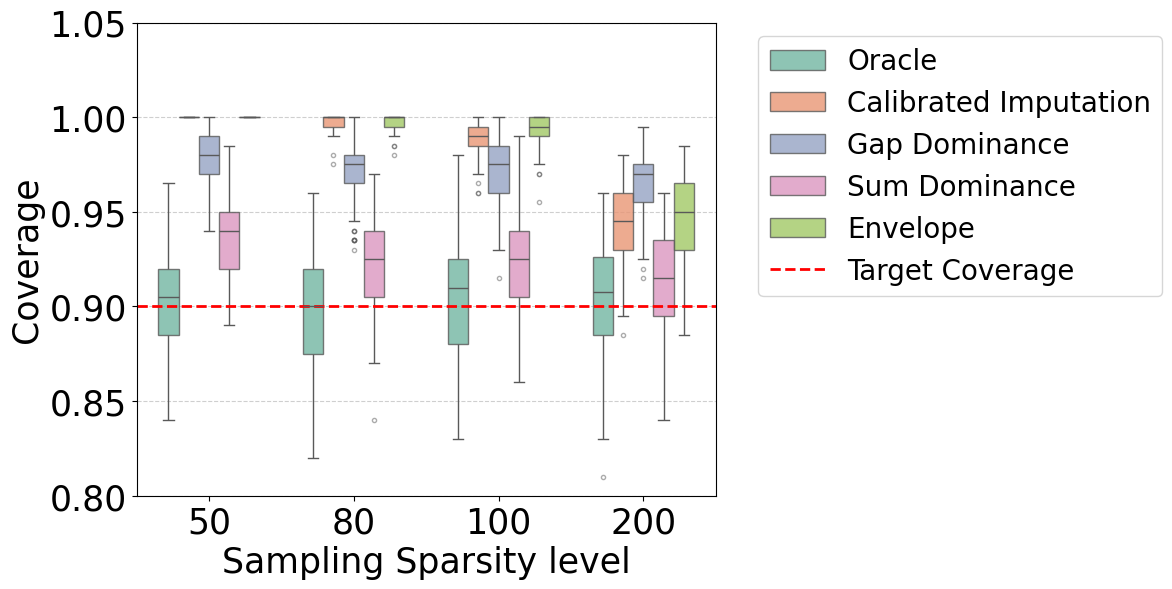}
    \caption{Coverage of methods at different random sampling sparsity levels.}
\label{compare_running_example}
\end{figure}

\conference{\textbf{Acknowledgments. } This work was supported by the National Science Foundation under the Grant IIS-SLES-2417075, by Lockheed Martin, and by a USC-Capital One CREDIF award.}




\arxiv{
\section{Discussion}\label{discussion_section}
At the core of most of our methods is the assumption of a Lipschitz constant for the test trajectory. Our proposed Lipschitz estimation method from Section \ref{lipschitz_estimation_section} requires an additional dataset as part of the calibration, however in many settings this comes with no additional cost. For example, in moderately dense sampling regimes, downsampling the dense trajectories gives the sparse version that we need to leverage either of the stochastic dominance assumption approaches that we introduced. While the geometric interpretation of Assumption \ref{stochastic_dominance_assumption}
is similar to that of Assumption \ref{sum_stochastic_dominance_assumption}, our experiments demonstrate that they differ in a fundamental tradeoff: we find that Assumption \ref{stochastic_dominance_assumption} holds with stronger evidence more often than Assumption \ref{sum_stochastic_dominance_assumption}, but this comes with  conservatism in the obtained prediction regions.

\cmmnt{We introduced a variety of methods for building valid prediction regions under random sampling times. While each method works under different assumptions and settings, having different methods offers more flexibility to obtain continuous-time prediction regions under diverse regimes. The only method that manages to work without a Lipschitz assumption is the stochastic dominance applied directly to the nonconformity scores. This method has the advantage of generally being less conservative, since using a Lipschitz constant relies on worst-case behavior of the trajectories. However the drawback is the reliance on a stochastic dominance assumption, which may be hard to verify in practice. We highlight that the calibrated imputations and the prediction envelope rely on a Lipschitz constant, but calibrated imputations requires the additional predicted sampling times, which in some settings may not be available. However, under accurate imputations, this method approximately yields prediction regions with less reliance on a Lipschitz constant, approximating the oracle nonconformity scores provided by the test sample imputation approach. Furthermore, the calibrated imputations method can be seen as an approach that relaxes the strong assumptions of the test sample imputation method, at the expense of access to synthetic sampling times for calibration.}

For future work, we consider determining sampling conditions under which either of the stochastic dominance assumptions can be satisfied, and one step further in this direction would be to have a principled sampling approach that guarantees a sweet spot in the tradeoff of conservativeness and the satisfaction of a stochastic dominance assumption. The main challenge is the fact that by construction the hidden gap, $\Delta^{\mathrm{hidden}}$, remains unobservable.
}
\arxiv{\section*{{Code Availability}}
Code to reproduce our experiments is available at
\href{https://github.com/jalvarezga/Continuous-Time-Conformal-Prediction}
{https://github.com/jalvarezga/Continuous-Time-Conformal-Prediction}.}

\bibliographystyle{ieeetr}
\bibliography{refs}

\appendix
\section*{{Appendix}}

\section{Results for the Gap Stochastic Dominance Assumption}\label{theoretical_results_dominance_section}
\cmmnt{While the  theoretical properties that we present here do not explicitly give a closed-form, finite-sample subsampling ratio under which the  stochastic dominance assumption is guaranteed to hold, they establish its foundational validity in the limiting behavior of the high frequency sampling. Our results are asymptotic yet consist with classical results from calculus.}
Consider a parametrized version of the dense sampling times in the setting of the gap dominance assumption, so that all the calibration trajectories have as sampling times $\mathbf{T}_{\mathrm{full}}(h)=\{t^{\mathrm{full}}_1,t^{\mathrm{full}}_2,\dots, t^{\mathrm{full}}_{K(h)}\}\subset[0,T^*]$. These can be  unstructured sampling times that partition $[0, T^*]$, with   $0=t^{\mathrm{full}}_1<\dots<t^{\mathrm{full}}_{K(h)}=T^*$  such that $\text{mesh}(\mathbf{T}_{\mathrm{full}}(h)):=\max(t^{\mathrm{full}}_2-t^{\mathrm{full}}_1,\dots,t^{\mathrm{full}}_{K}-t^{\mathrm{full}}_{K-1})\leq h$. The notation $K(h)$ emphasizes that $K\in \mathbb{N}$ is a function of $h$. One particular example is taking equally spaced points of length $h$ that partition $[0,T^*]$, i.e.,  $\mathbf{T}_{\mathrm{full}}(h)=\{t^{\mathrm{full}}_j:j=0,1,\dots, K\}$, where  $t^{\mathrm{full}}_j=jh$ for $j=0,1,\dots, K$, and $h=\frac{T^*}{K}$ like those used to obtain Fig. \ref{dominance_frequencies}. Define the discrete Lipschitz constant over this grid for a generic trajectory as $L_h = \max\limits_{s, t \in \mathbf{T}_{\mathrm{full}}(h), s>t} \frac{\|X_s - X_t\|}{s - t}$ over consecutive $s,t\in \mathbf{T}_{\mathrm{full}}(h)$.  This resembles the definition of $L^{\mathrm{full}}$  from Section \ref{lipschitz_estimation_section} for the specific sampling times $\mathbf{T}_{\mathrm{full}}(h)$. Throughout, we will also  denote by $\Delta_h^{\mathrm{hidden}}=L^{\mathrm{cont}}-L_h$, $\Delta^{\mathrm{obs}}_h = L_h - L^{\mathrm{lf}}$, and  $L^{\mathrm{cont}}$ to be the Lipschitz constant of a generic continuously differentiable trajectory $X\sim P_{X}$, i.e.,  $L^{\mathrm{cont}}=\max\limits_{t\in [0,T^*]}||\dot{X}_t||$. Consistent with the definition of Section \ref{lipschitz_estimation_section}, throughout we take $L^{\mathrm{lf}}= \max\limits_{s, t \in \boldsymbol{T}, s>t} \frac{\|X_s - X_t\|}{s - t}$ over consecutive $s,t\in \boldsymbol{T}$.

\begin{lemma}\label{lemma_frequency_limit}
Under the above definitions,  
$$\lim\limits_{h \to 0^+} \Delta_h^{\mathrm{hidden}} = 0\text{ almost surely}.$$
\end{lemma}
\begin{proof}
Let $X\sim P_X$ be a realization of the continuously differentiable process. By Weierstrass Extreme Value Theorem  there exists $t^*$ such that $L^{\mathrm{cont}} = \|\dot{X}_{t^*}\|$. If $L^{\mathrm{cont}} = 0$, the trajectory is constant in time, trivially satisfying the result. So assume $L^{\mathrm{cont}} > 0$. Let $u := \frac{\dot{X}_{t^*}}{\|\dot{X}_{t^*}\|}$. Define a function $g$ given by $g(t) := \langle u, X_t \rangle$. Given that $X$ is continuously differentiable as a function of time, then $g(t)$ is continuously differentiable. Hence, 
${g}'(t^*) = \left\langle \frac{\dot{X}_{t^*}}{\|\dot{X}_{t^*}\|}, \dot{X}_{t^*} \right\rangle  = \|\dot{X}_{t^*}\| = L^{\mathrm{cont}}$. Since $g'(t)$ is continuous on a closed interval, it must be uniformly continuous. Hence for any arbitrarily small $\epsilon > 0$, there exists $\delta > 0$ such that: $\text{ if } |t - \tau| < \delta, \text{ then } |{g}'(t) - {g}'(\tau)| < \epsilon.$

Choose a grid spacing $h < \delta$. The grid $\mathbf{T}_{\mathrm{full}}(h)$ partitions the interval, so $t^*$ must fall in some sub-interval $[t^{\mathrm{full}}_j, t^{\mathrm{full}}_{j+1}]$ of width  smaller than or equal to $h$. By the Mean Value Theorem, there exists some time $c \in (t^{\mathrm{full}}_{j}, t^{\mathrm{full}}_{j+1})$ such that  $\frac{g(t^{\mathrm{full}}_{j+1}) - g(t^{\mathrm{full}}_{j})}{t^{\mathrm{full}}_{j+1}-t^{\mathrm{full}}_{j}} = {g}'(c)$. Since $c$ and $t^*$ are both in this interval (or on the boundary) of width at most $h < \delta$, we know $|c - t^*| < \delta$.
Therefore, by uniform continuity: $${g}'(c) > {g}'(t^*) - \epsilon = L^{\mathrm{cont}} - \epsilon.$$

On the other hand, we have $g'(c)=\frac{g(t^{\mathrm{full}}_{j+1}) - g(t^{\mathrm{full}}_{j})}{t^{\mathrm{full}}_{j+1}-t^{\mathrm{full}}_{j}} = \left\langle u, \frac{X_{ t^{\mathrm{full}}_{j+1}} - X_{ t^{\mathrm{full}}_{j}}}{t^{\mathrm{full}}_{j+1}-t^{\mathrm{full}}_{j}} \right\rangle$. Thus, by Cauchy-Schwarz Inequality,
\begin{equation*}\begin{aligned}
    \left\langle u, \frac{X_{t^{\mathrm{full}}_{j+1}} - X_{t^{\mathrm{full}}_j}}{ t^{\mathrm{full}}_{j+1}- t^{\mathrm{full}}_{j}} \right\rangle &\leq \|u\| \left\| \frac{X_{ t^{\mathrm{full}}_{j+1}} - X_{ t^{\mathrm{full}}_{j}}}{t^{\mathrm{full}}_{j+1}- t^{\mathrm{full}}_{j}} \right\|\\
    &= \left\| \frac{X_{ t^{\mathrm{full}}_{j+1}} - X_{ t^{\mathrm{full}}_{j}}}{t^{\mathrm{full}}_{j+1}- t^{\mathrm{full}}_{j}} \right\|.
    \end{aligned}
\end{equation*}

Hence, we have that $\left\| \frac{X_{t^{\mathrm{full}}_{j+1}} - X_{t^{\mathrm{full}}_{j}}}{t^{\mathrm{full}}_{j+1}- t^{\mathrm{full}}_{j}} \right\| \geq {g}'(c) > L^{\mathrm{cont}} - \epsilon$.

Given that the discrete Lipschitz estimate $L_h$ is the maximum over all grid consecutive pairs, $$L_h \geq \left\| \frac{X_{t^{\mathrm{full}}_{j+1}} - X_{t^{\mathrm{full}}_{j}}}{t^{\mathrm{full}}_{j+1}- t^{\mathrm{full}}_{j}} \right\| > L^{\mathrm{cont}} - \epsilon.$$
Thus, $\Delta^{\mathrm{hidden}}_h<\epsilon$.\end{proof}
Connecting Lemma \ref{lemma_frequency_limit} with an explicit comparison between the gaps, we obtain the following result.

\begin{proposition}\label{assumption_1_verification}
Assume the process $X \sim P_X$ is continuously differentiable almost surely and non-degenerate in the sense that $\mathbb{P}_{P_{X, \boldsymbol{T}}}\!\left(L^{\mathrm{cont}} - L^{\mathrm{lf}} > 0\right) = 1$. Let $\mathbf{T}_{\mathrm{full}}(\cdot)$ be any deterministic family of sampling time grids parameterized by $h > 0$ satisfying $\mathrm{mesh}(\mathbf{T}_{\mathrm{full}}(h)) \le h$ for all $h>0$. Then,
$$\mathbb{P}_{P_{X, \boldsymbol{T}}}\!\left( \exists h^* > 0 \text{ such that } \forall h \in [0, h^*), \; \Delta^{\mathrm{hidden}}_h \le \Delta^{\mathrm{obs}}_h \right) = 1.$$
\end{proposition}

\begin{proof} We will show that the following set relationship holds $\{ L^{\mathrm{cont}} - L^{\mathrm{lf}}>0\}\cap\{\lim\limits_{h \to 0^+} \Delta_h^{\mathrm{hidden}} = 0\}\subseteq \{\exists h^* > 0 \text{ such that } \forall h < h^*, \; \Delta^{\mathrm{hidden}}_h \leq \Delta^{\mathrm{obs}}_h\}$. From there, we will leverage that intersection of almost sure events is another almost sure event. 

Fix a realization $(X, \boldsymbol{T}) \sim P_{X, \boldsymbol{T}}$. Let $a: = L^{\mathrm{cont}} - L^{\mathrm{lf}}> 0$. Note that $\lim_{h \to 0^+} L_h = L^{\mathrm{cont}}$. Therefore, the limits of the gaps are: $\lim\limits_{h \to 0^+} \Delta^{\mathrm{hidden}}_h= 0$, and  $\quad \lim\limits_{h \to 0^+} \Delta^{\mathrm{obs}}_h = \lim\limits_{h \to 0^+} (L_h - L^{\mathrm{lf}}) = L^{\mathrm{cont}} - L^{\mathrm{lf}} = a$.
Following the proof of Lemma \ref{lemma_frequency_limit} taking $\epsilon = \frac{a}{2} > 0$, there exists an $h^* > 0$ such that for all $h \in [0, h^*)$: $\Delta^{\mathrm{hidden}}_h < \epsilon=a/2.$
Simultaneously, for $h < h^*$, the observed gap evaluates to:$$\Delta^{\mathrm{obs}}_h = L_h - L^{\mathrm{lf}} = (L^{\mathrm{cont}} - \Delta^{\mathrm{hidden}}_h) - L^{\mathrm{lf}} = a - \Delta^{\mathrm{hidden}}_h.$$ Since $\Delta^{\mathrm{hidden}}_h < \frac{a}{2}$, we have: $\Delta^{\mathrm{obs}}_h > a - \frac{a}{2} = \frac{a}{2}$. Therefore, for all $h < h^*$, the gaps are strictly ordered: $\Delta^{\mathrm{hidden}}_h < \frac{a}{2} < \Delta^{\mathrm{obs}}_h$. 

We have shown that

$\{ L^{\mathrm{cont}} - L^{\mathrm{lf}}>0\}\cap\{\lim\limits_{h \to 0^+} \Delta_h^{\mathrm{hidden}} = 0\}\subseteq \{\exists h^* > 0 \text{ such that } \forall h < h^*, \; \Delta^{\mathrm{hidden}}_h \leq \Delta^{\mathrm{obs}}_h\}$, with $\mathbb{P}_{P_{X, \boldsymbol{T}}}\left( L^{\mathrm{cont}} - L^{\mathrm{lf}}>0\right)=1$, and 
$\mathbb{P}_{P_{X, \boldsymbol{T}}}\left( \lim\limits_{h \to 0^+} \Delta_h^{\mathrm{hidden}} = 0\right)=1$. 
Therefore, 

$\mathbb{P}_{P_{X, \boldsymbol{T}}}\left( \exists h^* > 0 \text{ such that } \forall h < h^*, \; \Delta^{\mathrm{hidden}}_h \leq \Delta^{\mathrm{obs}}_h \right) = 1$.
\end{proof}

Our non-degeneracy assumption $\mathbb{P}(L^{\mathrm{cont}} - L^{\mathrm{lf}} > 0) = 1$ is satisfied whenever the low-frequency grid $\boldsymbol{T}$ fails to perfectly capture the exact moment when the maximum derivative magnitude is attained. For general continuous-time processes, this is not a restrictive assumption, since the event $\{L^{\mathrm{cont}}=L^{\mathrm{lf}}\}$ will have measure zero, unless we have an adversarial process, such as strictly constant time derivative, where sparse sampling achieves the peak velocity over the entire continuous time domain.\footnote{The particular adversarial case of almost sure constant velocity would not be a problem, since any sampling frequency would recover a Lipschitz constant for the trajectory.}

This result gives us a principled understanding for why under high-frequency sampling regimes, Assumption \ref{stochastic_dominance_assumption} is reasonable.
In particular, since almost sure dominance implies stochastic dominance, we have:

$$\Delta^{\mathrm{hidden}}\leq\Delta^{\mathrm{obs}}\text{ almost surely }\implies \text{Assumption } \ref{stochastic_dominance_assumption}\text{ holds}.$$

To support the intuition of Proposition \ref{assumption_1_verification}, consider Fig. \ref{dominance_frequencies}, where we present the simulation of 200 realizations of the difference of the gaps $\Delta_h^{\mathrm{obs}}-\Delta_h^{\mathrm{hidden}}$ for different resolutions of the high-frequency sampling times $\mathbf{T}_{\mathrm{full}}(h)$.  When the support of $\Delta_h^{\mathrm{obs}}-\Delta_h^{\mathrm{hidden}}$ is non-negative, we get almost sure dominance. The setup to obtain this figure is the same as that of Fig. \ref{comparison_frequencies_empirical_cdf}. The difference is that in the context of Fig. \ref{comparison_frequencies_empirical_cdf} we are evaluating stochastic dominance, whereas in Fig. \ref{dominance_frequencies} we are evaluating almost sure dominance.
\begin{figure}[H]
\centering
\includegraphics[width=0.5\linewidth]{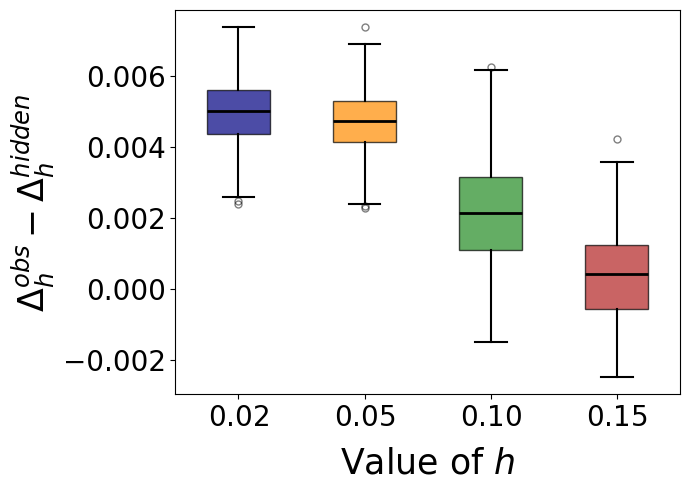}
\caption{Gap differences.}\label{dominance_frequencies}
\end{figure}
For higher frequencies ($h\in\{0.02,0.05\}$) all the observed gap differences are always non-negative, whereas sparser resolutions break this almost sure dominance relationship between the gaps.

\section{Calibrated Envelopes with an Estimated Lipschitz Constant}\label{exchangeability_correction_section}
The enveloped prediction region from Section \ref{prediction_envelope_section} relies on nonconformity scores computed using ${L}>0$. However, if we recycle the exact same calibration data to both estimate $\hat{L}$ and then compute the nonconformity scores, we violate exchangeability between the calibration and test nonconformity scores, since the test nonconformity score would depend on a threshold $\hat{L}$ computed  solely with the calibration dataset. 

Assume $\{(X^{(i)}, \boldsymbol{T}^{(i)}, \boldsymbol{T}^{(i)}_{{\mathrm{full}}})\}_{i=1}^{m+1}\overset{iid}{\sim}P_{X, \boldsymbol{T}, \boldsymbol{T}_{{\mathrm{full}}}}$. To preserve the i.i.d. assumption, we partition the calibration dataset into two disjoint subsets: $\mathcal{D}_{cal} = \mathcal{D}_{L} \cup \mathcal{D}_{S}$. Index-wise without loss of generality the first $m-k$  observations correspond to $\mathcal{D}_{L}$ and the remaining $k$  correspond to  $\mathcal{D}_{S}$. Optimal data splitting is outside the scope of this work, the reader may benefit from ideas from \cite{das2026optimal}.
From here we proceed in two steps. Firstly, we use $\mathcal{D}_{L}$ exclusively to obtain a valid upper bound for $L_{m+1}^{\mathrm{cont}}$ (with either (\ref{predicted_L}) or (\ref{predicted_L_2})), obtaining $\hat{L}$ such that:
\begin{equation}\label{lipschitz_step1}
\mathbb{P}_{P_{X,\boldsymbol{T}, \boldsymbol{T}_{\mathrm{full}}}^{m-k+1}}\Big(L_{m+1}^{\mathrm{cont}} \leq\hat{L}\Big)\geq 1-\alpha.\end{equation}
We then use $\mathcal{D}_{S}$ solely to compute the envelope nonconformity scores based on this fixed $\hat{L}$. For each trajectory $i \in\{m-k+1, \dots, m, m+1\}$, let $\hat{{D}}^{(i)}_t$ be the envelope at time $t$:

\begin{equation}
  \hat{{D}}^{(i)}_t \coloneqq \min_{t^{(i)}_k \in \boldsymbol{T}^{(i)}} \left( \|X^{(i)}_{t^{(i)}_k} - \hat{x}_t\| + \hat{L}|t - t^{(i)}_k| \right),  
\end{equation}
where we use an additional upper hat notation to distinguish with the main text definition based on a known Lipschitz constant. 
We define the corresponding nonconformity score $\hat{{S}}_i$ as: $ \hat{{S}}_i:= \max\limits_{t \in [0, T^*]} \frac{\hat{{D}}^{(i)}_t}{\hat{\sigma}_t}.$ With this setup we are ready to state our result hybridizing Section \ref{lipschitz_estimation_section}
 with Section \ref{random_sampling_times_section}.
 \begin{proposition}
Let $\hat{L} = \hat{L}(\mathcal{D}_L)$ be such that (\ref{lipschitz_step1}) holds, and use it to define $\hat{q}\coloneqq\text{Quantile}(\{\hat{{S}}_i\}_{i=m-k+1}^m;(1+\frac{1}{k})(1-\delta))$. If we construct the prediction region as
$ \mathcal{C}_t = \left\{ x \in \mathbb{R}^d : \|x - \hat{f}_t\| \leq \hat{q} \cdot \hat{\sigma}_t \right\}, t\in [0,T^*],$
then:
$$ \mathbb{P}_{P_{X,\boldsymbol{T}, \boldsymbol{T}_{\mathrm{full}}}^{m+1}}\Big( X^{(m+1)}_t \in \mathcal{C}_t \text{ for all } t \in [0, T^*] \Big) \geq 1 - (\alpha + \delta). $$
\end{proposition}

\begin{proof}
Let $E_{\text{split}}:=\{L_{m+1}^{\mathrm{cont}} \leq\hat{L}\}$. By triangle inequality, this event implies $\|X^{(m+1)}_t - \hat{f}_t\| \leq \hat{{D}}^{(m+1)}_t$ for all $t$. By assumption, $\mathbb{P}_{P_{X,\boldsymbol{T}, \boldsymbol{T}_{\mathrm{full}}}^{m-k+1}}(E_{\text{split}}) \geq 1-\alpha$. Let $E_{\text{score}}:=\{\hat{{S}}_{m+1} \leq \hat{q}\}$. 
Conditional on $\mathcal{D}_L$, the threshold $\hat{L}$ is a fixed constant, and the variables $\mathcal{D}_S \cup \{(X^{(m+1)}, \boldsymbol{T}^{(m+1)}, \boldsymbol{T}^{(m+1)}_{\mathrm{full}})\}$ remain exchangeable (Fact 3.7  from \cite{angelopoulos2025theoreticalfoundationsconformalprediction}). Consequently, the nonconformity scores $\{\hat{{S}}_i\}_{i=m-k+1}^m$ and the test score $\hat{{S}}_{m+1}$ are conditionally exchangeable. By split CP, it follows that $\mathbb{P}_{P^{k+1}_{(X,\boldsymbol{T}, \boldsymbol{T}_{\mathrm{full}})|\mathcal{D}_{L}}}(E_{\text{score}} \mid \mathcal{D}_L) \geq 1-\delta$. Taking the expectation over $\mathcal{D}_L$ and applying the tower property, we obtain the marginal validity: $\mathbb{P}_{P^{k+1}_{X,\boldsymbol{T}, \boldsymbol{T}_{\mathrm{full}}}}(E_{\text{score}}) = \mathbb{E}_{\mathcal{D}_L}[\mathbb{P}_{P^{k+1}_{(X,\boldsymbol{T}, \boldsymbol{T}_{\mathrm{full}})|\mathcal{D}_{L}}}(E_{\text{score}} \mid \mathcal{D}_L)] \geq \mathbb{E}_{\mathcal{D}_L}[1-\delta] = 1-\delta.$ 
Thus,  if the event $E_{\text{score}}\cap E_{\text{split}}$ holds, then we have for all $t \in [0, T^*]$:
\begin{equation*}
    \frac{\|X^{(m+1)}_t - \hat{f}_t\|}{\hat{\sigma}_t} \leq \frac{\hat{{D}}^{(m+1)}_t}{\hat{\sigma}_t} \leq \sup_{s \in [0,T^*]} \frac{\hat{{D}}^{(m+1)}_s}{\hat{\sigma}_s} = \hat{{S}}_{m+1} \leq \hat{q}.
\end{equation*}
This sequence of inequalities implies $X^{(m+1)}_t \in \mathcal{C}_t$ across the entire domain. Using the monotonicity of the probability measure along with the union bound, and marginalization, we conclude: 
\begin{equation*}
\begin{split}
    \mathbb{P}_{P_{X,\boldsymbol{T}, \boldsymbol{T}_{full}}^{m+1}}\Big(X^{(m+1)}_t \in \mathcal{C}_t, \forall t \in [0, T^*]\Big)& \geq \mathbb{P}_{P_{X,\boldsymbol{T}, \boldsymbol{T}_{full}}^{m+1}}(E_{\text{split}} \cap E_{\text{score}})\\
    &\geq \mathbb{P}_{P_{X,\boldsymbol{T}, \boldsymbol{T}_{full}}^{k+1}}(E_{\text{score}})+\mathbb{P}_{P_{X,\boldsymbol{T},\boldsymbol{T}_{full}}^{m-k+1}}(E_{\text{split}})-1\\
    &\geq 1 - (\alpha + \delta).
\end{split}
\end{equation*}
\end{proof}
The steps to obtain valid coverage under estimated Lipschitz constant for the calibrated imputations method follow the same procedure based on a split of the dataset to preserve exchangeability.

\section{Empirical Validation of the Stochastic Dominance Assumption}\label{appendix_empirical_validation}
The fundamental idea of Assumption \ref{stochastic_dominance_assumption} and Assumption \ref{sum_stochastic_dominance_assumption} is motivated by a downsampling induced error. The key bottleneck that prevents a rigorous validation of these style of  assumptions is that, by definition,  the hidden gap, $\Delta^{\mathrm{hidden}}$, is unobservable, which is precisely what motivates the approach of using two sampling regimes. 
However, we can use the available data to check if the principle behind the downsampling behavior of the trajectories holds at different observable downsampling scales. 

Here we present an empirical diagnostic procedure to check stochastic dominance at observable subsampling ratios without relying on ground truth continuous data. We use this approach to obtain Fig. \ref{comparison_frequencies_empirical_cdf}  and Fig. \ref{comparison_frequencies_cdf_sum} where we compare this approach against oracle ground truth continuous data.

We construct sampling times $\boldsymbol{T}^{(i)} , \boldsymbol{T}_{\mathrm{mid}}^{(i)},\boldsymbol{T}^{(i)}_{\mathrm{full}}$ such that $|\boldsymbol{T}^{(i)}|<|\boldsymbol{T}_{\mathrm{mid}}^{(i)}|<|\boldsymbol{T}^{(i)}_{\mathrm{full}}|$, where $\boldsymbol{T}^{(i)}$ 
is the target sparse sampling regime, and $\boldsymbol{T}_{\mathrm{mid}}^{(i)}$ is a chosen mid-frequency regime for the $i$-th trajectory. For example, 
$\boldsymbol{T}_{\mathrm{mid}}^{(i)} \subset \boldsymbol{T}^{(i)}_{\mathrm{full}}, \text{ and } \boldsymbol{T}^{(i)} \subset \boldsymbol{T}^{(i)}_{\mathrm{full}}$. To avoid  a trivial  comparison\footnote{If the grids are nested, then $L^{\mathrm{mid}} \geq L^{\mathrm{lf}}$ is true pointwise for every single trajectory. Thus, testing pointwise dominance would be a trivial identity. This follows since $\boldsymbol{T} \subset \boldsymbol{T}_{\mathrm{mid}}$ implies we can upper bound $ L^{\mathrm{lf}}$ by a weighted average of slopes of all the consecutive points in the interval where the maximum that defines $ L^{\mathrm{lf}}$ is attained.} between subsampling ratios we require that $\boldsymbol{T}^{(i)}\not\subset\boldsymbol{T}_{\mathrm{mid}}^{(i)}$. Another example in the fixed sampling times setting is to take evenly spaced samples with spacings of $0<h<h_{\mathrm{mid}}<H$, for $\boldsymbol{T}^{(i)}_{\mathrm{full}}, \boldsymbol{T}_{\mathrm{mid}}^{(i)},\text{ and } \boldsymbol{T}^{(i)}$, respectively.

We can then check if the dominance relation holds at these observable scales. Usually, if the dominance relation holds consistently across the observed scales, it provides empirical evidence that the chosen grid resolution is faithful to preserve the stochastic dominance assumption.

Concretely, we define a generic surrogate hidden gap to act as a proxy of the oracle hidden gap: $\hat{\Delta}^{\mathrm{hidden}} = L^{\mathrm{full}} - L^{\mathrm{mid}}$
and we compare it to 
$\Delta^{\mathrm{obs}}$. Here $L^{\mathrm{mid}}:=\max\limits_{s,t\in \boldsymbol{T}_{\mathrm{mid}}, s<t}\frac{||{X}_t- {X}_s||}{t-s}$ over consecutive $s,t\in \boldsymbol{T}_{\mathrm{mid}}$.  The intuition of this surrogate is to mimic a “fine-to-coarse” drop using only observable scales. Analogously, we define the corresponding quantities for the case of checking  stochastic dominance directly on the nonconformity scores as in the setting of Section \ref{stochastic_dominance_section}, namely, we take $\boldsymbol{\hat\Delta}^{\textbf{hidden}}= S^{\mathrm{full}} - S^{\mathrm{mid}}$, with $S^{\mathrm{mid}}:=\max\limits_{t\in \boldsymbol{T}_{\mathrm{mid}}}R(t)$.

We can then plot the empirical CDFs of these quantities as in Fig. \ref{comparison_frequencies_empirical_cdf} and Fig. \ref{comparison_frequencies_cdf_sum}, where we can see that the surrogate comparison serves as a proxy for the true oracle stochastic dominance at different sampling  frequencies. This empirical pipeline recovers an intuitive phenomenon that holds in the  oracle plots: when the subsampling ratio $H/h$ is larger, the stochastic dominance tends to hold more strongly.\cmmnt{  More generally, for a larger drops in the sampling times $|\boldsymbol{T}^{(i)}_{\mathrm{full}}|-|\boldsymbol{T}^{(i)}|$ the gap between the CDF of the hidden and the observable gap gets larger.}

\section{Additional Experiments}\label{additional_experiments}
As a reminder, this section is structured as follows. In Appendix \ref{dmps_experiments} we reproduce the methodology of Section \ref{experiments_section} for comparing our conformal methods under random sampling times. Then, we narrow down the experiments to study our methods under stochastic dominance assumptions: in Appendix \ref{attractor_experiment} we study stochastic dominance directly on the nonconformity scores, whereas in Appendix \ref{lorenz_experiment} we present experiments of Lipschitz constant estimation in higher dimensions.
\subsection{Simulation of Dynamic Movement Primitives}\label{dmps_experiments}
Dynamic Movement Primitives (DMPs) is a versatile robot learning paradigm to learn from demonstrations \cite{pastor2009learning,ijspeert2013dynamical, hielscher2025interactive}. We use the  Python library  movement-primitives \cite{fabisch2024movement_primitives} for this experiment. Following \cite{pastor2009learning}, we consider the ODE formulation of a one-dimensional DMP:

\begin{equation}\label{eq1_dmps}
\begin{split}
\tau \dot{v}&=K(g-x)-Dv+(g-x_0)f\\
\tau \dot{x}&=v,
\end{split}
\end{equation}
where $x$ and $v$ are position and velocity of the system, $x_0$ and $g$ are start and goal position respectively; $\tau$ is a temporal scaling factor, $K$ is a spring constant, $D$ is the damping term (constant), and $f$ is a nonlinear function, called the learned forcing term, learned to allow for generating complex movements. Explicitly, it is given by $f(s)=\frac{\sum^{M}_{i=1}w_i\phi_i(s)s}{\sum_{i=1}^{M}\phi_{i}(s)},$ where $\phi_{i}(s)=\text{exp}(-h_i (s-c_i)^2)$ are Gaussian basis functions, $M$ is the number of basis functions, with center $c_i$ and width $h_i$. The variable $s$ corresponds to a phase variable which monotonically decreases from $1$ to $0$, and such that $\tau\dot{s}=-\alpha_ss$, where $\alpha_s$ is a predefined constant. 
The basis functions divide the motion into phases, the weights $w_i$ are learned to encode how to match the movement to the demonstrated trajectory, as we will explain next for our particular demonstration.

We take a total of $M=20$ weights and set $\tau=T^*=1$. For the demonstration, we take $x_{\mathrm{demo}}(t)=10 t^3-15t^4+6t^5$, which corresponds to the minimum jerk trajectory \cite{flash1985coordination}, which is standard in robotics. Following the DMP approach, we collect samples of the expert demonstration at times $t_{1},\dots, t_N$, along with its first and second derivative, obtaining $x_{\mathrm{demo}},v_{\mathrm{demo}}, \dot{v}_{\mathrm{demo}}$. Using the first equation from (\ref{eq1_dmps}), which is called the transformation system, we obtain values for the target function:
\begin{equation}\label{target_f}
    f_{\mathrm{demo}}(s(t_k))=\frac{-K(g-x_{\mathrm{demo}}(t_k))+Dv_{\mathrm{demo}}(t_k)+\tau \dot{v}_{\mathrm{demo}}(t_k)}{g-x_0},\text{ for all }k=1,\dots,N.
\end{equation}
Here we set $x_0$ and $g$ the start and goal to be given by $x_{\mathrm{demo}}(0)$ and $x_{\mathrm{demo}}(t_N)$, respectively.

After this step, we learn the optimal weights by minimizing the following function with respect to $w\in \mathbb{R}^M$:

\begin{equation}\label{regression_DMPs}
J(w):=\sum_{k=1}^N\Big(f_{\mathrm{demo}}(s(t_k))-f(s(t_k))\Big)^2.
\end{equation}

Note that this corresponds to solving a least squares problem in linear regression, since the function $f$ is linear on the weights. For generating trajectories, we draw inspiration from \cite{hielscher2025interactive} by introducing noise on the learned weights from the expert demonstration, denoted by $w^*\in \mathbb{R}^M$. From here, we generate i.i.d. trajectories by solving a random ODE from (\ref{eq1_dmps}) using the weights $w^{(i)}:=w^*+w'_i$, with i.i.d. Gaussian weights $w'_i\sim \mathcal{N}(0,0.15\mathcal{I}_{M})$, where $\mathcal{I}_{M}$ denotes the identity matrix in $\mathbb{R}^{M\times M}$. A sample of these trajectories are in gray in Fig. \ref{dmps_data}, representing i.i.d. draws over DMPs that share a common start and goal but exhibit variance in the path shape.
\begin{figure}[H]
\centering
\includegraphics[width=0.5\linewidth]{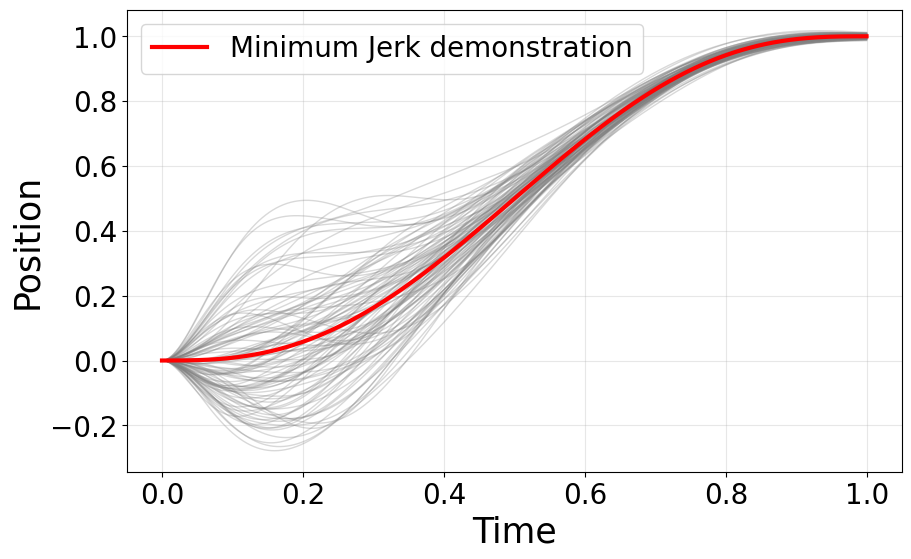}
\caption{Simulated trajectories and demonstration.}
\label{dmps_data}
\end{figure}
To obtain our point predictors, $\hat{x}, \hat{\sigma}$ we fit splines to training data at fixed sampling times.  In Fig. \ref{compare2} we include instances at different sampling sparsity levels, using the same setup of random sampling times as in Section \ref{experiments_section}. 

For evaluation of coverage taking $\delta=0.10$ we generate 100 calibration datasets of size $m=500$ and evaluate each on $200$ test trajectories. We highlight that in this setup we do not know a Lipschitz constant for the trajectories. Hence, for the methods that require it, we follow Appendix \ref{exchangeability_correction_section} by partitioning the calibration dataset of $m=500$ and using $250$ trajectories for Lipschitz estimation and the remaining $250$ for computing the nonconformity scores. For the calibrated imputations method, we use a static empirical distribution of sampling times $P_{\tilde{\boldsymbol{T}}}$ for each sampling sparsity level, as in the experiments from Section \ref{experiments_section}. Although the oracle is obtained using an estimated Lipschitz constant, it assumes access to samples from $P_{\boldsymbol{T}}$ which corresponds to uniform random sampling times from $[0,1]$.

\begin{figure}[H]
    \centering
\includegraphics[width=0.9\linewidth]{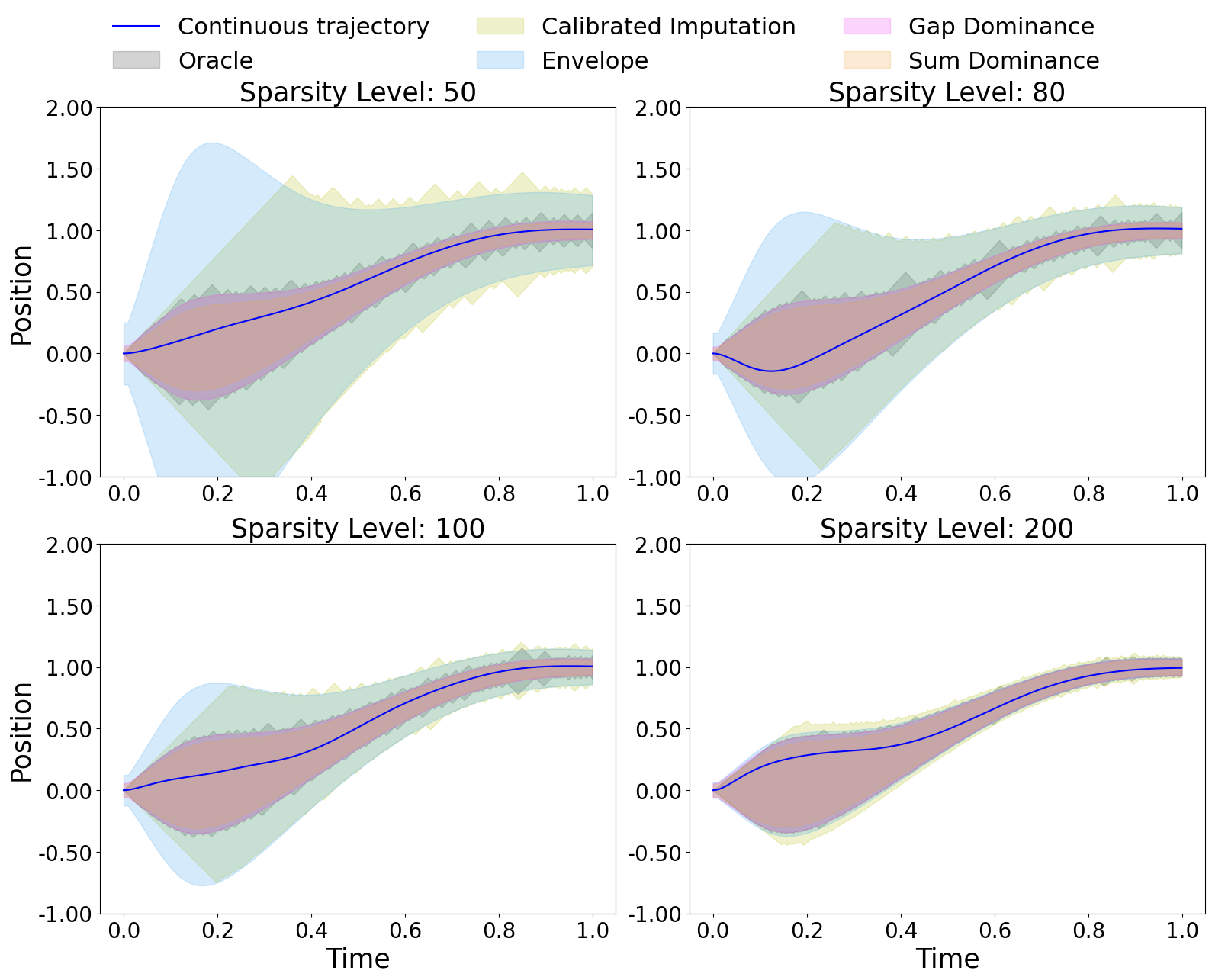}
    \caption{Prediction regions by number of uniform random sampling times.}
\label{compare2}
\end{figure}
Figure \ref{compare2} is consistent with our results from Section \ref{experiments_section}: the oracle region is non-conservative compared to the rest of the approaches, and the prediction regions get tighter as the sampling becomes more dense. In turn, we corroborate how this affects coverage in Figure \ref{dmps_compare_methods}.

\begin{figure}[H]
\centering\includegraphics[width=0.9\linewidth]{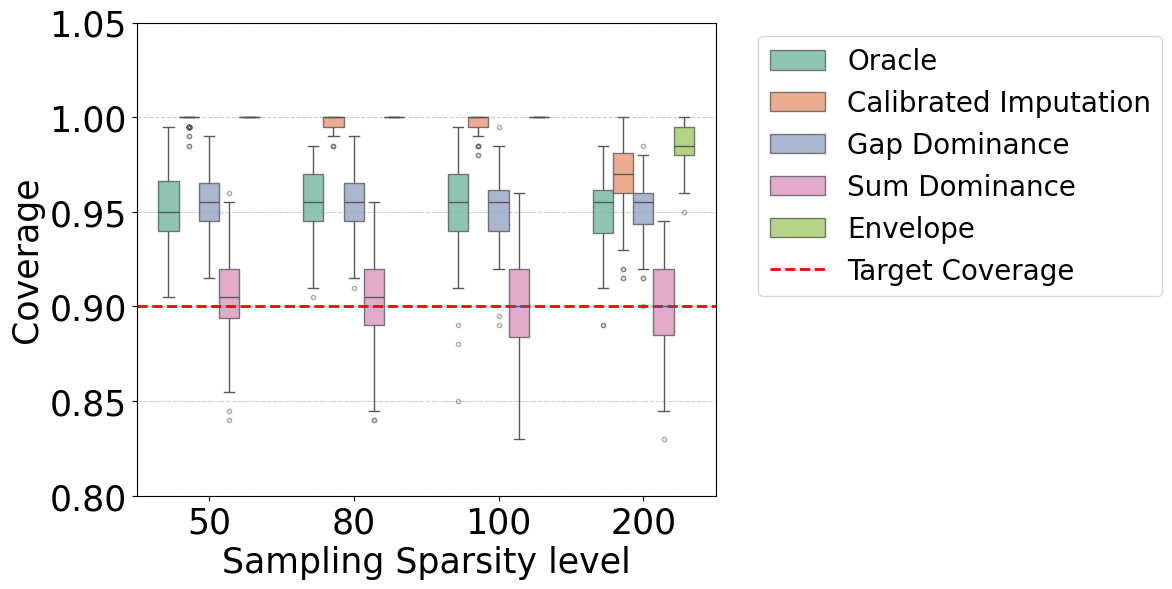}
    \caption{Coverage of methods under varying sampling sparsity levels.}
\label{dmps_compare_methods}
\end{figure}


\subsection{Rössler Attractor}\label{attractor_experiment}
Focusing on adaptive ODE solvers in a RDE setup, we consider trajectories generated by the solution to Rössler's attractor:
\begin{equation}\label{attractor_equations}
\begin{split}
    \dot{x}&=-y - z\\
    \dot{y}&=x + ay\\
    \dot{z}&=b + z(x - c),\\
\end{split}
\end{equation}
which is a chaotic system useful for modeling equilibrium in chemical reactions \cite{auer2012vericomp}. We take $a\sim U(0.2-\epsilon,0.2+\epsilon)$, $b\sim U(0.2-\epsilon, 0.2+\epsilon)$, $c\sim U(5.7-\epsilon, 5.7+\epsilon)$, with $\epsilon=0.05$. 
In this setting we have random sampling times due to the randomness in the parameters which induce changes in spike timing, curvature of the trajectory, and stiffness of the system. Hence, the solver internally adjusts step sizes differently for each trajectory. We use $\text{solve}\_\text{ivp}$ from SciPy with the BDF (backward differentiation formula) method in the time interval $[0,40]$. To obtain a point predictor, we use a weighted kernel regression using $M=500$ trajectories as training data. Concretely, let $w^{(i)}_j(t)$ denote the normalized kernel weights for a query time $t$: $w^{(i)}_j(t) = \frac{K_h\!\big(t - t^{(i)}_j\big)}{\sum_{k=1}^{M}\sum_{l=1}^{n_k} K_h\!\big(t - t^{(k)}_l\big)},$ where $K_h(u)=\exp(-\frac{u^2}{2h^2})$ with bandwidth $h=0.3$. 
We take  as our point predictor
$$\hat{x}_t = \sum_{i=1}^{M}\sum_{j=1}^{n_i} w^{(i)}_j(t)\, x^{(i)}_{t^{(i)}_j},$$
where, $x^{(i)}_{t^{(i)}_j}$ denotes the $i$-th trajectory at sampling time $t^{(i)}_j$, and $n_i$ denotes the number of sampling times for the $i$-th trajectory in the training data. We use a similar approach to estimate the heuristic uncertainty $\hat{\sigma}_t$ by computing the weighted root mean squared error, bounded below by a small constant for numerical stability:$$\hat{\sigma}_t = \max\Big\{ \Big( \sum_{i=1}^{M}\sum_{j=1}^{n_i} w^{(i)}_j(t) \big(x^{(i)}_{t^{(i)}_j} - \hat{x}_t\big)^2 \Big)^{\!1/2}, 10^{-3} \Big\}.$$

In Fig. \ref{regions_vericomp} we present prediction regions via gap and via sum dominance on the nonconformity scores for a test trajectory at different downsampling scales, corresponding to the $y$ solution from Rössler's attractor (\ref{attractor_equations}). For a trajectory with full grid times $\boldsymbol{T}_{\mathrm{full}} = \{t_1, t_2, \dots, t_K\}$, a downsampled grid at scale $k \in \mathbb{N}$ is formed by retaining every $k$-th observation: $\boldsymbol{T}(k):= \left\{ t_{1 + j \cdot k} \;\middle\vert{}\; j = 0, 1, \dots, \left\lfloor \frac{K-1}{k} \right\rfloor \right\}.$  We consider four distinct downsampling regimes characterized by the downsampling step $k \in \{15, 10, 5, 4\}$:
\begin{itemize}
    \item \texttt{Sparse+} ($k = 15$): Highly degraded observations retaining approximately $6.7\%$ of the original solver grid points.
    \item \texttt{Sparse} ($k = 10$): Coarse observation grids retaining $10\%$ of the full sampling density.
    \item \texttt{Dense} ($k = 5$): Moderately fine observation grids retaining $20\%$ of the solver steps.
    \item \texttt{Dense+} ($k = 4$): High-density subsampling retaining $25\%$ of the full adaptive solver grid.
\end{itemize}
We can see that the most sparse, \texttt{Sparse+},  produces more conservative prediction regions than the most dense downsampling, $\texttt{Dense+}$. Consequently, this is reflected in the coverage as shown in  Fig. \ref{coverages_vericomp_dominance}. Intuitively, this occurs because $\texttt{Dense+}$ subsampling keeps more sampling times than $\texttt{Sparse+}$, making the underlying observable gap closer to the hidden gap. This is consistent with the empirical CDFs in Fig. \ref{gaps_cont}.

\begin{figure}[H]
    \centering
    \includegraphics[width=0.9\linewidth]{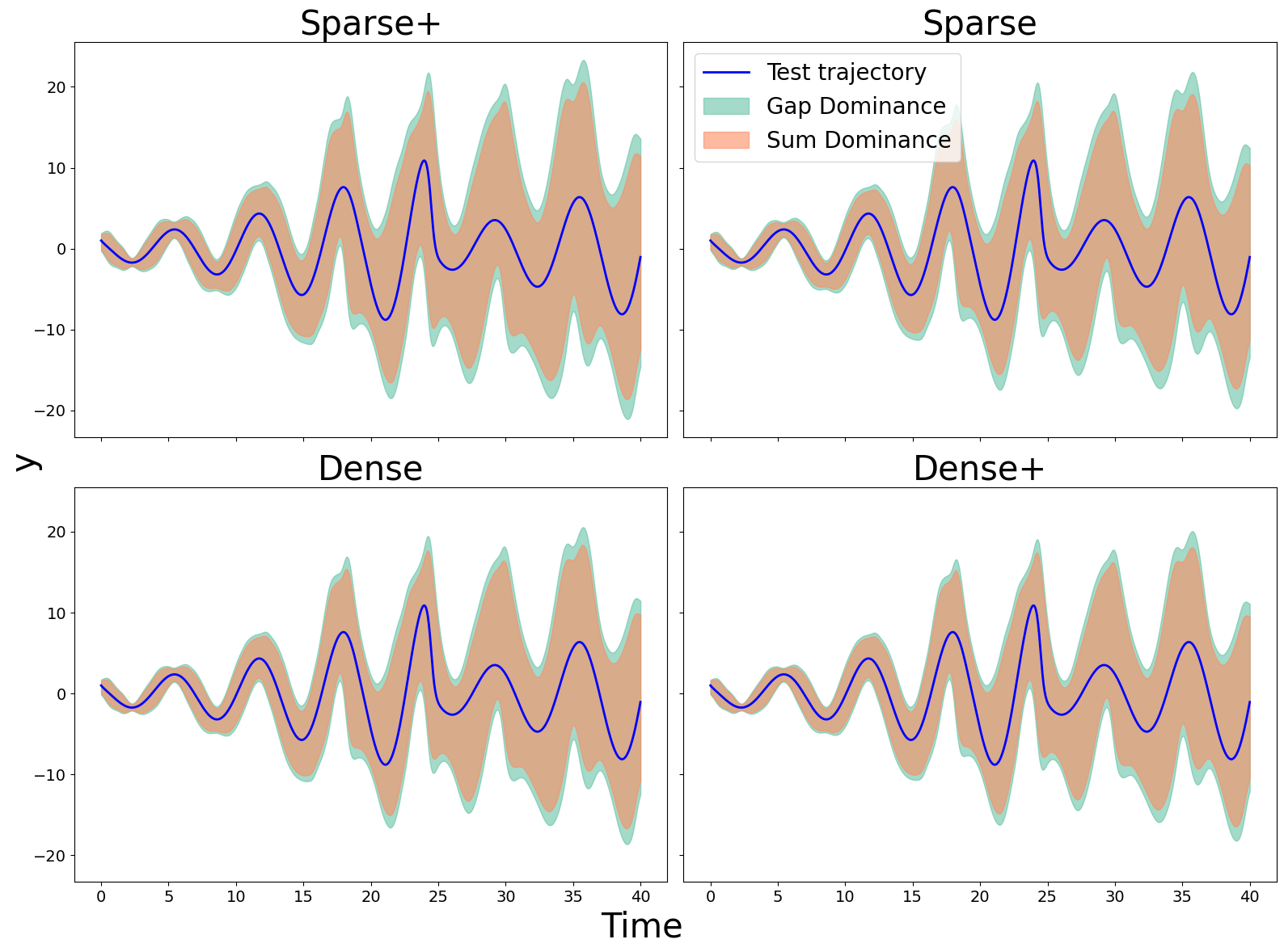}
    \caption{Prediction regions for a test trajectory at different downsampling scales.}
    \label{regions_vericomp}
\end{figure}

We generate 100 observations of coverage based on calibration datasets consisting of $200$ trajectories, each coverage is  obtained using 200 test trajectories. With this setup we obtain Fig. \ref{coverages_vericomp_dominance}.

\begin{figure}[H]
\centering \includegraphics[width=0.6\linewidth]{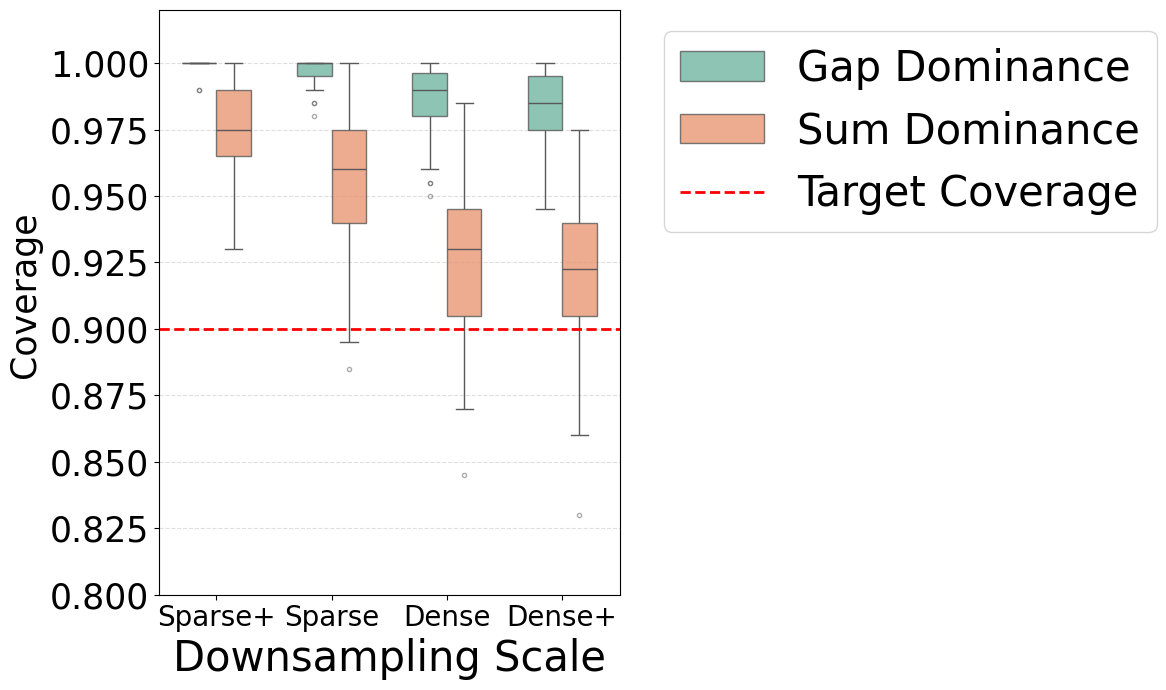}
\caption{Coverage of dominance frameworks at different downsampling regimes.}\label{coverages_vericomp_dominance}
\end{figure}
We conclude with an analysis of the validation of the gap and sum dominance assumptions on the nonconformity scores. We generate 10 independent datasets of 200 trajectories to build the CDFs following the procedure of Appendix \ref{appendix_empirical_validation}; shaded regions represent 2 standard deviations from the mean empirical CDF.

Overall our empirical verification method based on surrogates provides faithful guidance of the oracle setting with access to samples of $S^{\mathrm{cont}}$, when we compare Fig. \ref{gaps_cont} to Fig. \ref{gaps_surrogate}, and Fig. \ref{sum_cont} relative to Fig. \ref{sum_surrogate}. Furthermore, here we also validate the tradeoff where we get less conservative prediction regions (reflected in the coverage and in the prediction regions), compromising the satisfaction of the dominance assumption.

\begin{figure}[H]
\centering \includegraphics[width=0.6\linewidth]{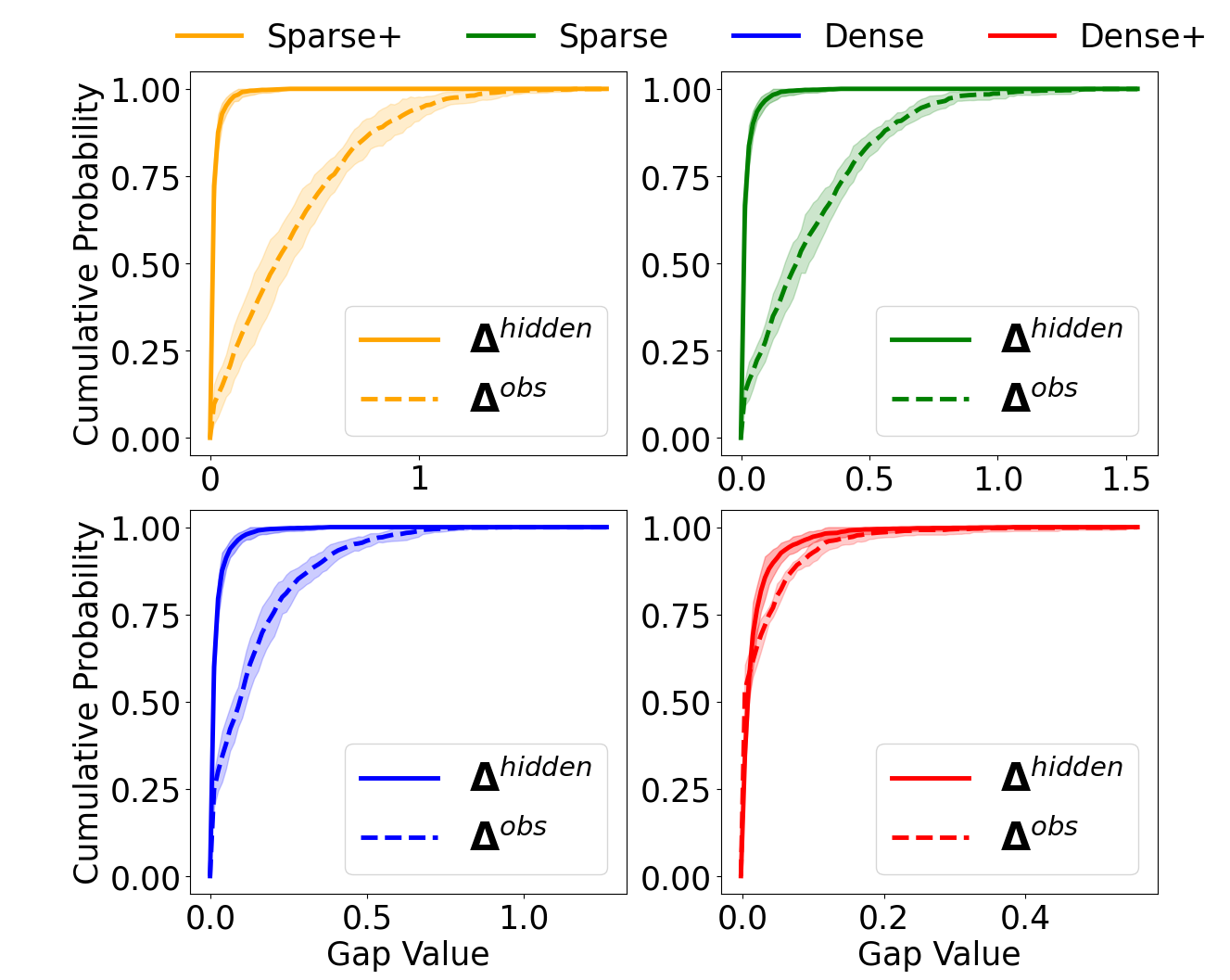}
\caption{Gap dominance with  continuous trajectories.}\label{gaps_cont}
\end{figure}

\begin{figure}[H]
\centering \includegraphics[width=0.6\linewidth]{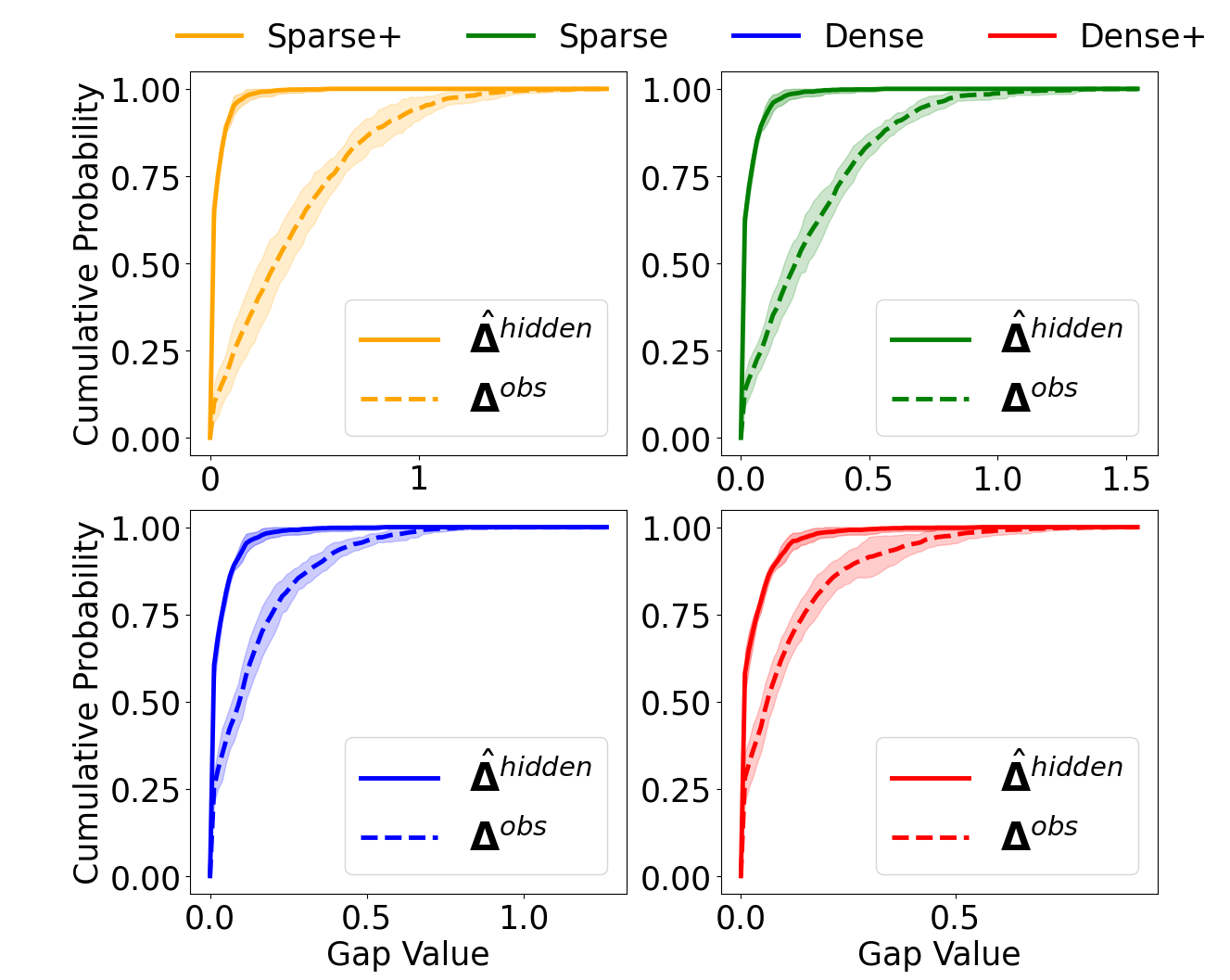}
\caption{Gap dominance using $\boldsymbol{\hat{\Delta}^{\mathrm{hidden}}}$  surrogates.}\label{gaps_surrogate}
\end{figure}

\begin{figure}[H]
\centering \includegraphics[width=0.6\linewidth]{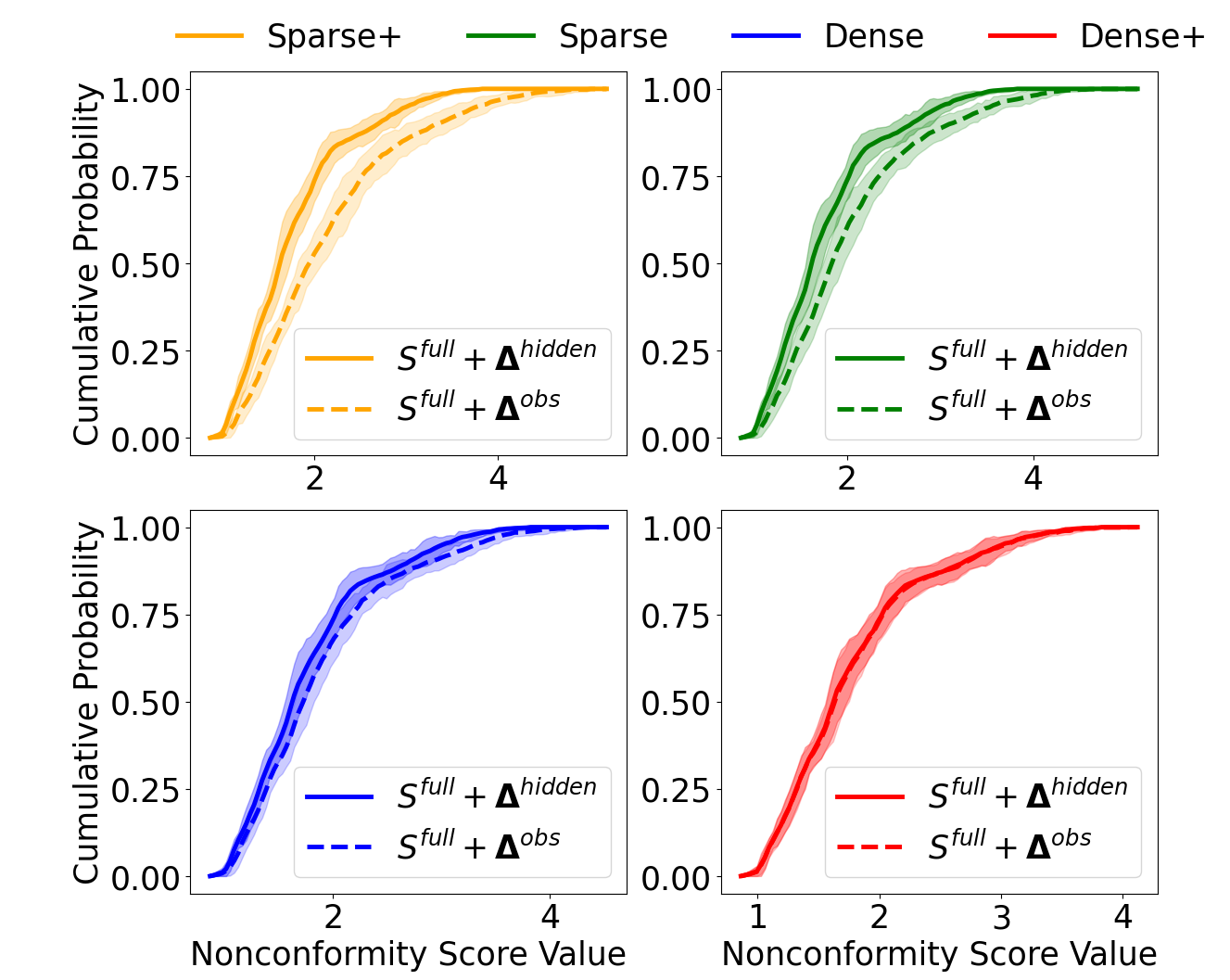}
\caption{Sum dominance with  continuous trajectories}\label{sum_cont}
\end{figure}

\begin{figure}[H]
\centering \includegraphics[width=0.6\linewidth]{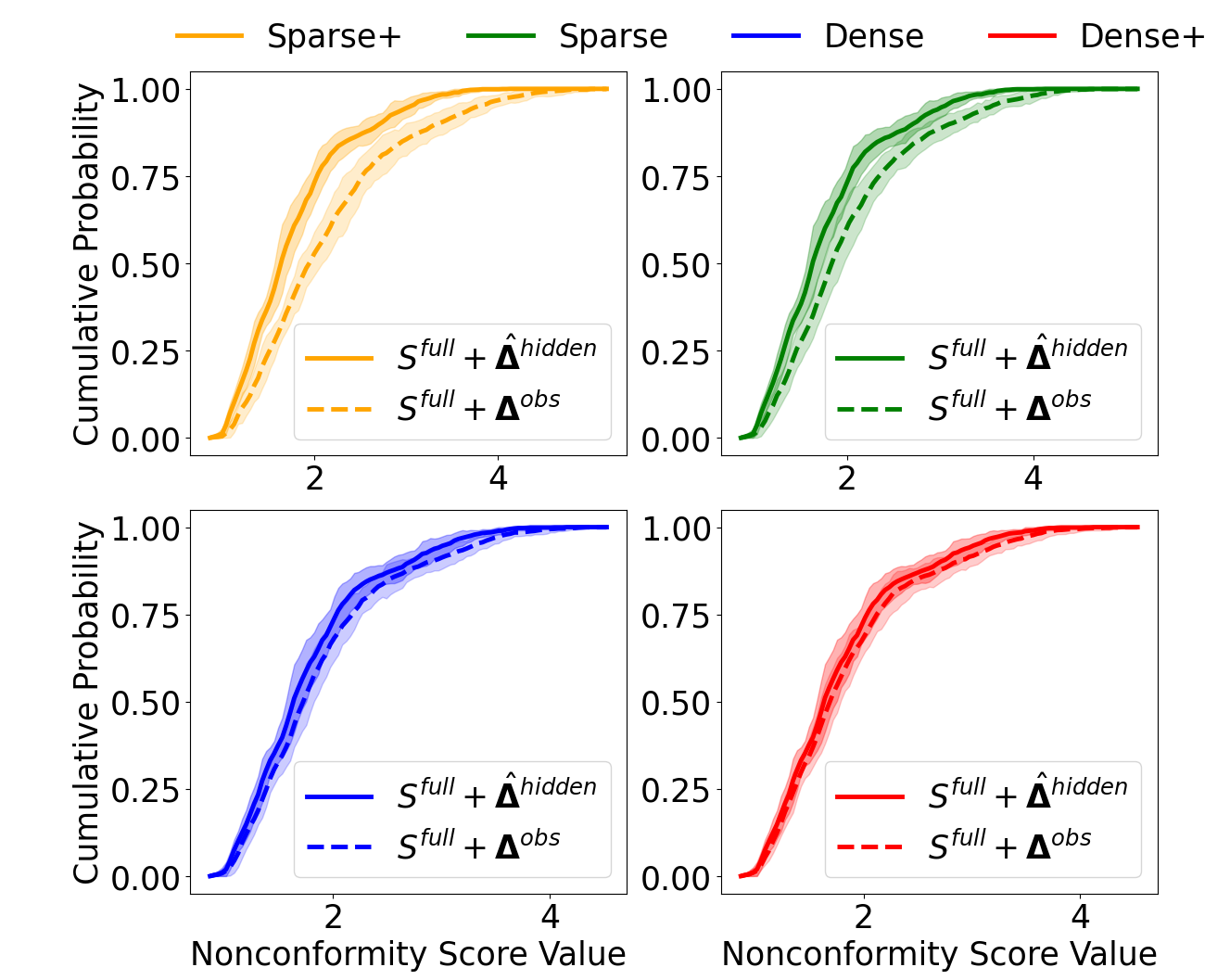}
\caption{Sum dominance using $\boldsymbol{\hat{\Delta}^{\mathrm{hidden}}}$ surrogates.}\label{sum_surrogate}
\end{figure}

\subsection{Lorenz 96}\label{lorenz_experiment}

We employ the 5-dimensional Lorenz-96 chaotic system, which serves as a canonical experiment case study in the literature \cite{jiang2022embed, sun2022bayesian}. The governing state equations for $x = [x_1, \dots, x_d]^T \in \mathbb{R}^d$ (we take $d=5$) are defined as:

$$\begin{aligned} \frac{\mathrm{d}x_1}{\mathrm{d}t} &= (x_2 - x_4)x_5 - x_1 + F \\ \frac{\mathrm{d}x_2}{\mathrm{d}t} &= (x_3 - x_5)x_1 - x_2 + F \\ \frac{\mathrm{d}x_3}{\mathrm{d}t} &= (x_4 - x_1)x_2 - x_3 + F \\ \frac{\mathrm{d}x_4}{\mathrm{d}t} &= (x_5 - x_2)x_3 - x_4 + F \\ \frac{\mathrm{d}x_5}{\mathrm{d}t} &= (x_1 - x_3)x_4 - x_5 + F \end{aligned}$$

To induce randomness across realizations, we make the initial condition and the forcing term random: $F= F_{\text{base}} + \xi$, where $F_{\text{base}} = 8.0$ and $\xi \sim U(-0.5, 0.5)$. For the initial conditions we independently generate $\boldsymbol{x_0} = F_{\text{base}} \mathbf{1}_d + \boldsymbol{\xi_0} + 0.1 \mathbf{e_1}$, where $\boldsymbol{\xi}_0 \sim U([-0.5, 0.5]^d)$. We use $\texttt{solve}\_\texttt{ivp}$ for $t \in [0, T^*] = [0, 20]$, and we use the $\texttt{RK45}$ method.
In this experiment, we study the effect of different subsampling regimes on the stochastic dominance assumptions in higher dimensions under 2 different sampling approaches.

\subsubsection{Fixed High-Frequency Sampling Times and Random Sparse Sampling Times}
In our first ablations we take $\boldsymbol{T}$ to be given by adaptive sampling and  $\boldsymbol{T}_{\mathrm{full}}$ by fixed prespecified sampling times. We take 
$\boldsymbol{T}_{\mathrm{full}}$ to be a \emph{linspace} of $[0,20]$ under varying number of sampling times $N\in \{500,600,800\}$, whereas we let the sparse sampling times $\boldsymbol{T}$ be given by the adaptive sampling times determined by the ODE solver. In Fig. \ref{adaptive_sampling_times1} we show a histogram of the number of sampling times obtained with the adaptive solver, showing that the number of sampling times is random and has less sampling times than our convention of a high-frequency sampling $\boldsymbol{T}_{\mathrm{full}}$.

\begin{figure}[H]
    \centering
\includegraphics[width=0.4\linewidth]{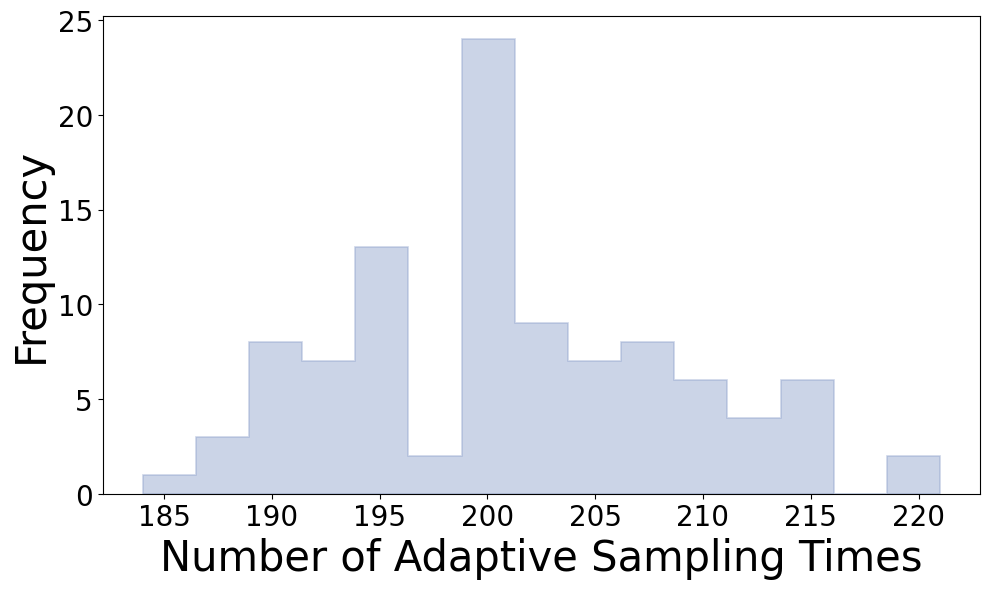}
    \caption{Histogram of 100 independent sampling times obtained with $\texttt{solve}\_\texttt{ivp}$  using the $\texttt{RK45}$ method.}
    \label{adaptive_sampling_times1}
\end{figure}

In Fig. \ref{cdfs_validation1} we present the validation of both Assumption \ref{stochastic_dominance_assumption} and Assumption \ref{sum_stochastic_dominance_assumption} for different full sampling regimes, where we generated 10 independent realizations of datasets with $200$ i.i.d. trajectories, plotting the average and 2 standard deviation as uncertainty band. Here we verify that when the number of sampling times is not enough our assumptions are violated, but this is remediated with higher sampling frequencies. 
\begin{figure}[H]
    \centering
    \begin{subfigure}[b]{0.7\textwidth}
\includegraphics[width=\textwidth]{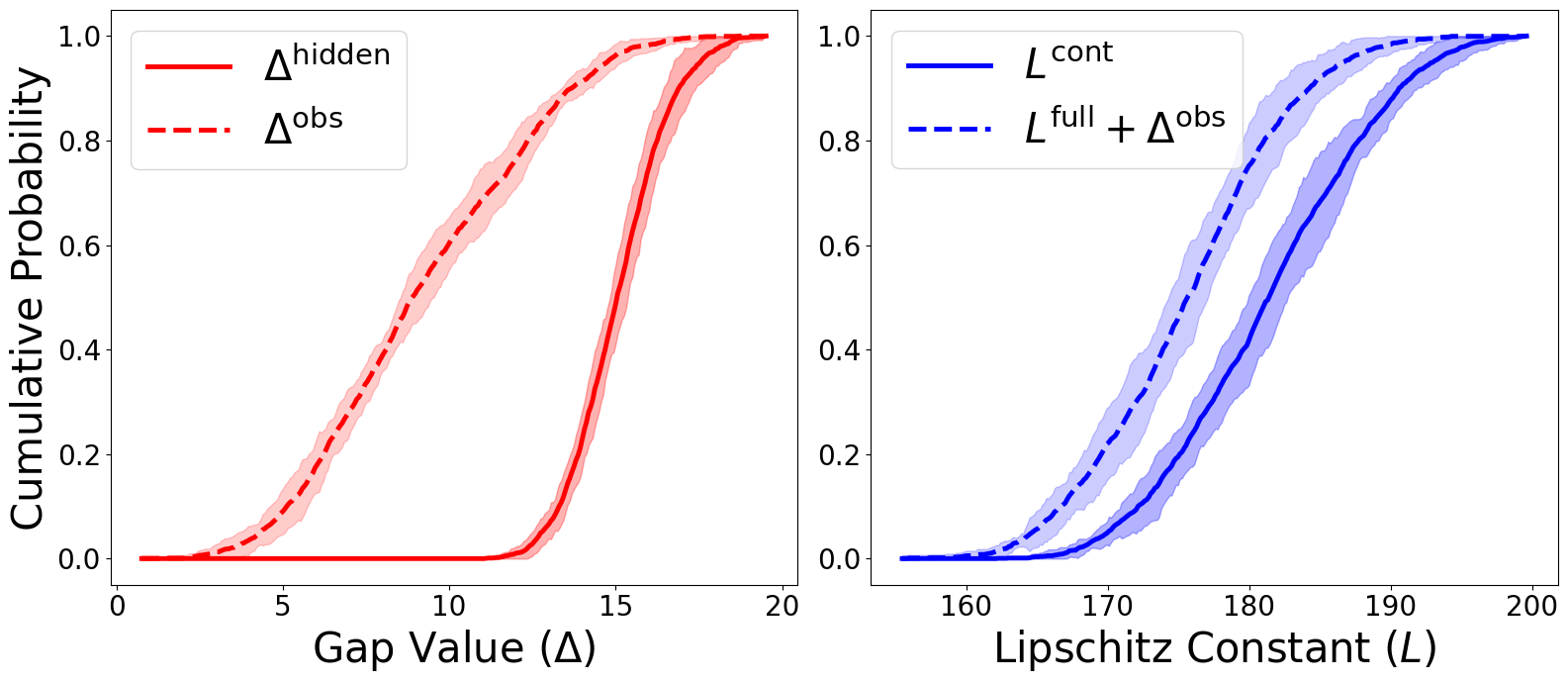}
        \caption{$N=500$.}
    \end{subfigure}
    \vfill
    \begin{subfigure}[b]{0.7\textwidth}
\includegraphics[width=\textwidth]{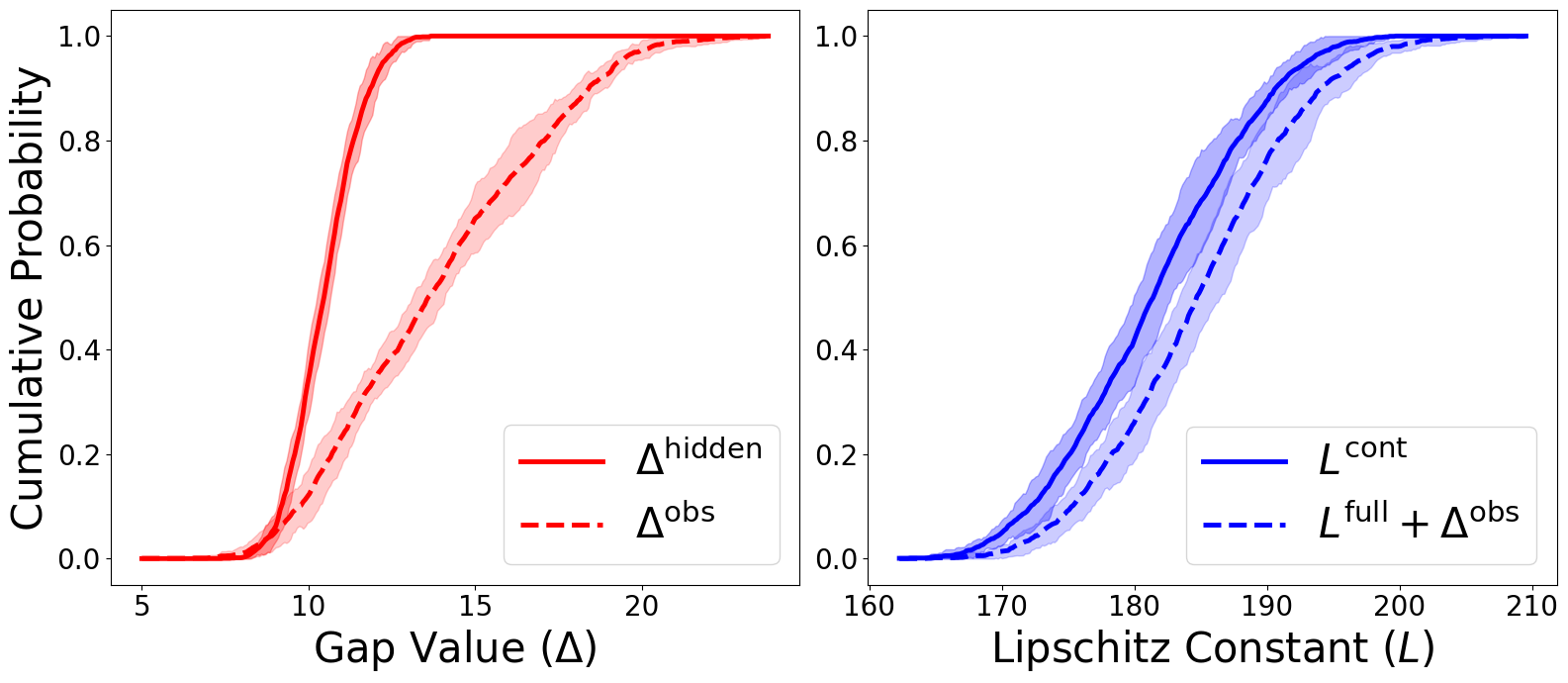}
        \caption{$N=600$.}
    \end{subfigure}
    \vfill
    \begin{subfigure}[b]{0.7\textwidth}
\includegraphics[width=\textwidth]{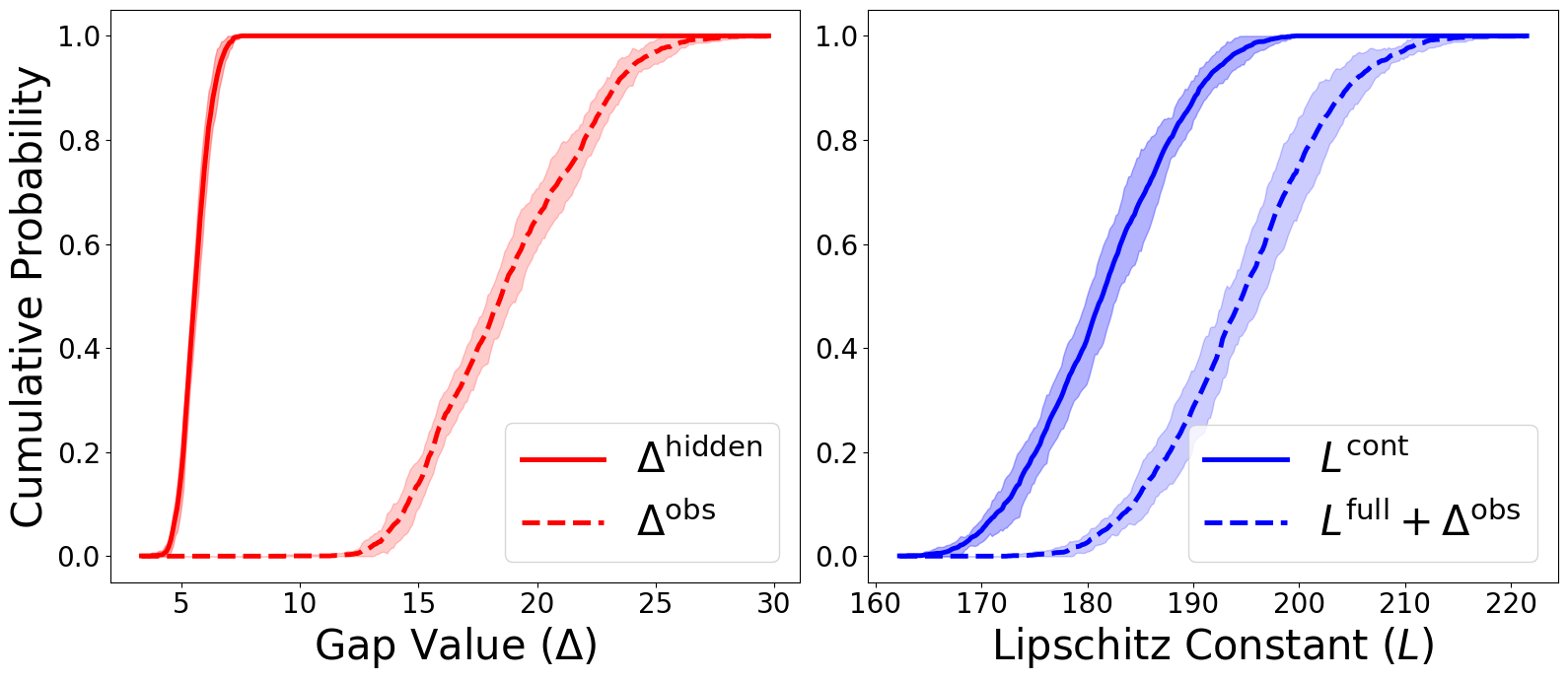}
        \caption{$N=800$.}
    \end{subfigure}
    \caption{Empirical CDFs to validate Assumption \ref{stochastic_dominance_assumption} (red), and Assumption \ref{sum_stochastic_dominance_assumption} (blue).}
    \label{cdfs_validation1}
\end{figure}
In Fig. \ref{coverages_quantiles_regime1} we evaluate the impact of each regime on the calibration quantiles and the resulting coverage. Although the violation of the sum stochastic dominance assumption directly translates into below average target coverage, this is not necessarily true for the gap stochastic dominance assumption. When using $N=500$ samples, despite the fact that Assumption \ref{stochastic_dominance_assumption} is not satisfied, we can see that our conformal approach achieves the desired target coverage. This is due to the fact that the gap approach is more conservative: its coverage guarantee is achieved by a union bound  and a split in the error budget. Hence, even if the hidden gap is stochastically larger than the observable gap, the difference can be compensated with the margin by which $q_{\mathrm{full}}$ exceeds $L^{\mathrm{full}}_{m+1}$, which is typically large due to the fact that $q_{\mathrm{full}}$  is a large quantile by the split in the error budget, for which we took $\alpha=0.1$ and computed $q_{\mathrm{full}}$ taking a $1-\alpha/2$ conformal quantile. Intuitively, we can think of two margins that cancel each other out, in such a way that we approximately recover our target coverage despite violating Assumption \ref{stochastic_dominance_assumption}.

\begin{figure}[H]
    \centering
    \begin{subfigure}[b]{0.49\textwidth}
\includegraphics[width=\textwidth]{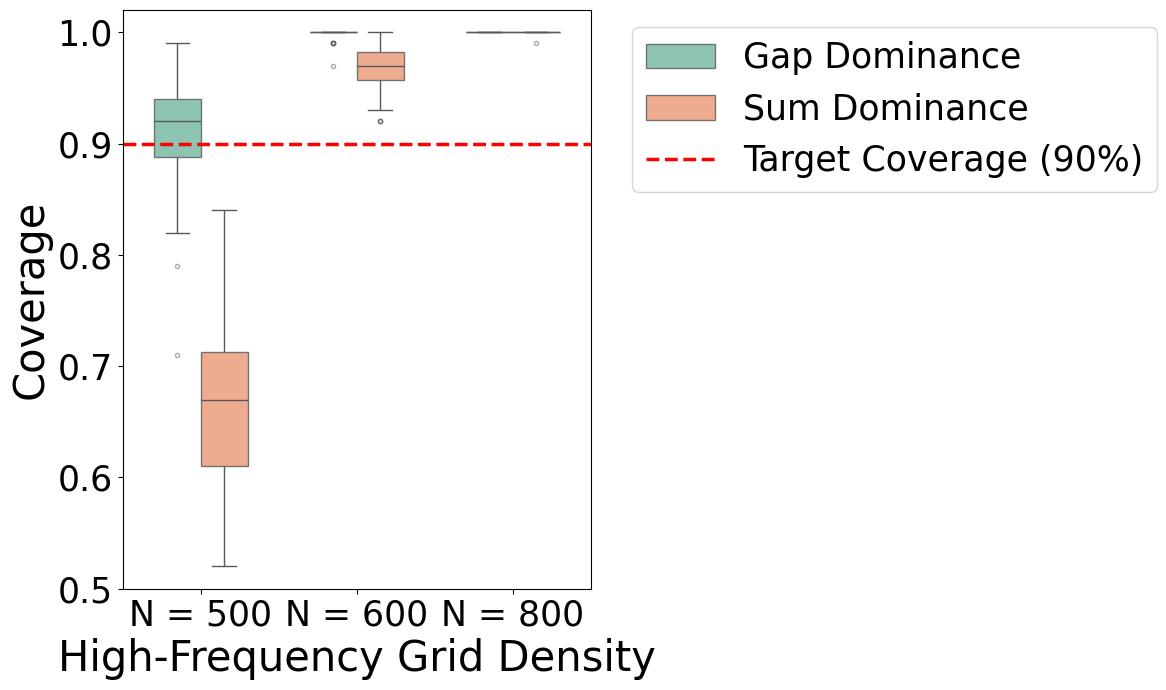}\caption{Coverage at different sampling frequencies.}
    \end{subfigure}
    \hfill
    \begin{subfigure}[b]{0.49\textwidth}
\includegraphics[width=\textwidth]{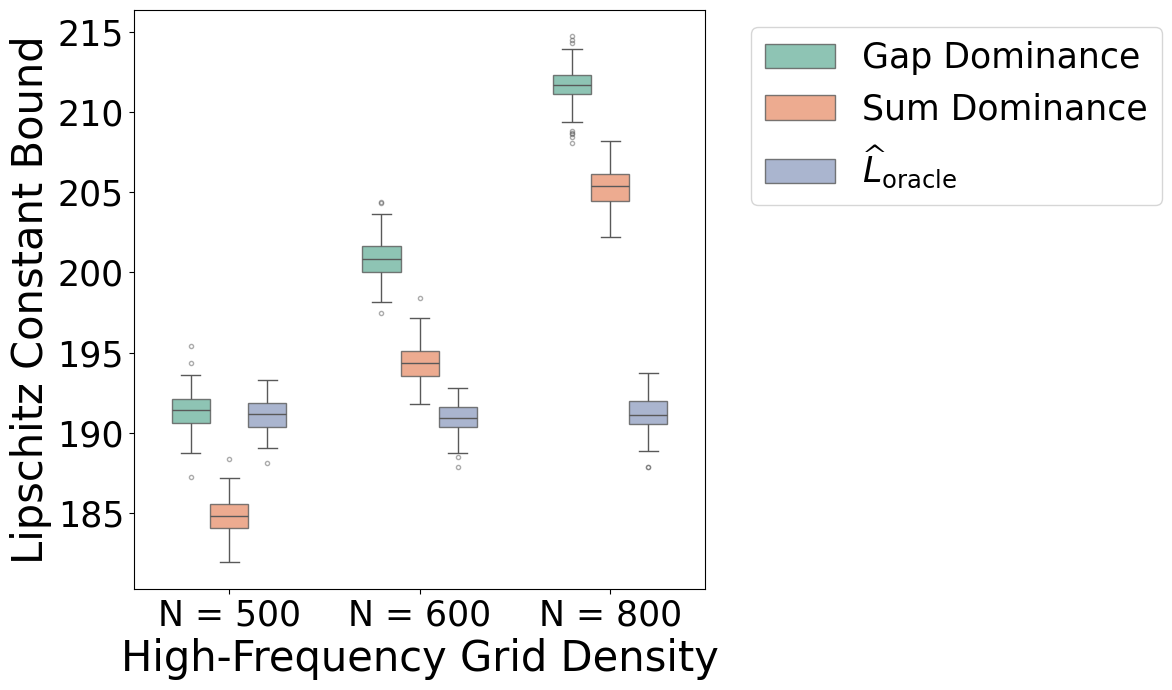}
        \caption{Predicted Lipschitz upper bounds.}
    \end{subfigure}
    \caption{Experiment results using 100 calibration datasets of 200 trajectories each.}\label{coverages_quantiles_regime1}
\end{figure}

\subsubsection{Downsampling Adaptive Sampling Times}

In this setting, we take $\boldsymbol{T}_{\mathrm{full}}$ to be given by the adaptive sampling times provided by the solver. We create four distinct observable subsampled regimes $\boldsymbol{T} \subset \boldsymbol{T}_{\mathrm{full}}$.

Our subsampling consists of uniformly discarding some percentage of the sampling times under which we evaluate our bounds, thus inducing varying levels of sampling sparsity. By specifying exact retention fractions $\eta \in (0, 1)$, for a given set of sampling times $\boldsymbol{T}_{\mathrm{full}}$, we take $ N_{\text{keep}}:=\lceil\eta N_{\text{full}}\rceil, N_{\text{full}}:=|\boldsymbol{T}_{\mathrm{full}}|$ and obtain subsampled times:

$$\boldsymbol{T}_{\eta}= \left\{ t_{\lfloor i \cdot \rho \rfloor} \in \boldsymbol{T}_{\mathrm{full}} \;\middle\vert{}\; i = 0, \dots, N_{\text{keep}}-1 \right\}, \text{ where } \rho := \frac{N_{\text{full}} - 1}{N_{\text{keep}} - 1}.$$
We take the following labeling for each downsampling regime:
$\texttt{Dense}+$ ($\eta = 0.90$), which  drops $10\%$ of points uniformly across time, $\texttt{Dense}$ ($\eta = 0.80$), $\texttt{Sparse}$ ($\eta = 0.70$), and  $\texttt{Sparse}+$ ($\eta = 0.60$), dropping $40\%$ of the sampling times from $P_{\boldsymbol{T}_{\mathrm{full}}}$. In Fig. \ref{adaptive_sampling_times2} we include histograms of the number of collected sampling times for  $2,000$ trajectories under these subsamplig regimes.

\begin{figure}[H]
    \centering
\includegraphics[width=0.5\linewidth]{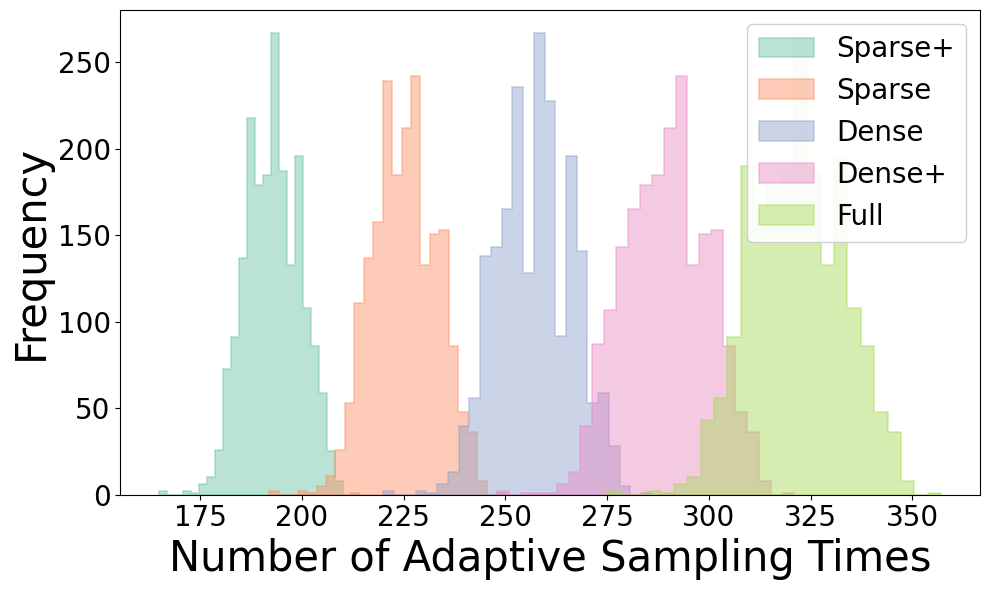}
    \caption{Number of collected sampling times at different induced sparsity levels.}
\label{adaptive_sampling_times2}
\end{figure}

In Fig. \ref{gaps_oracle2} we include validation of the gap stochastic dominance assumption using samples from oracle continuous trajectories. In Fig. \ref{gaps_surrogate2} we use our surrogate approach following the empirical approach from Appendix \ref{appendix_empirical_validation}. Here we observe deviations between our empirical approach and the oracle empirical CDFs, yet we see that the empirical validation remains faithful for capturing stochastic dominance assumption violations for the $\texttt{Dense}+$ regime.

\begin{figure}[H]
\centering \includegraphics[width=0.6\linewidth]{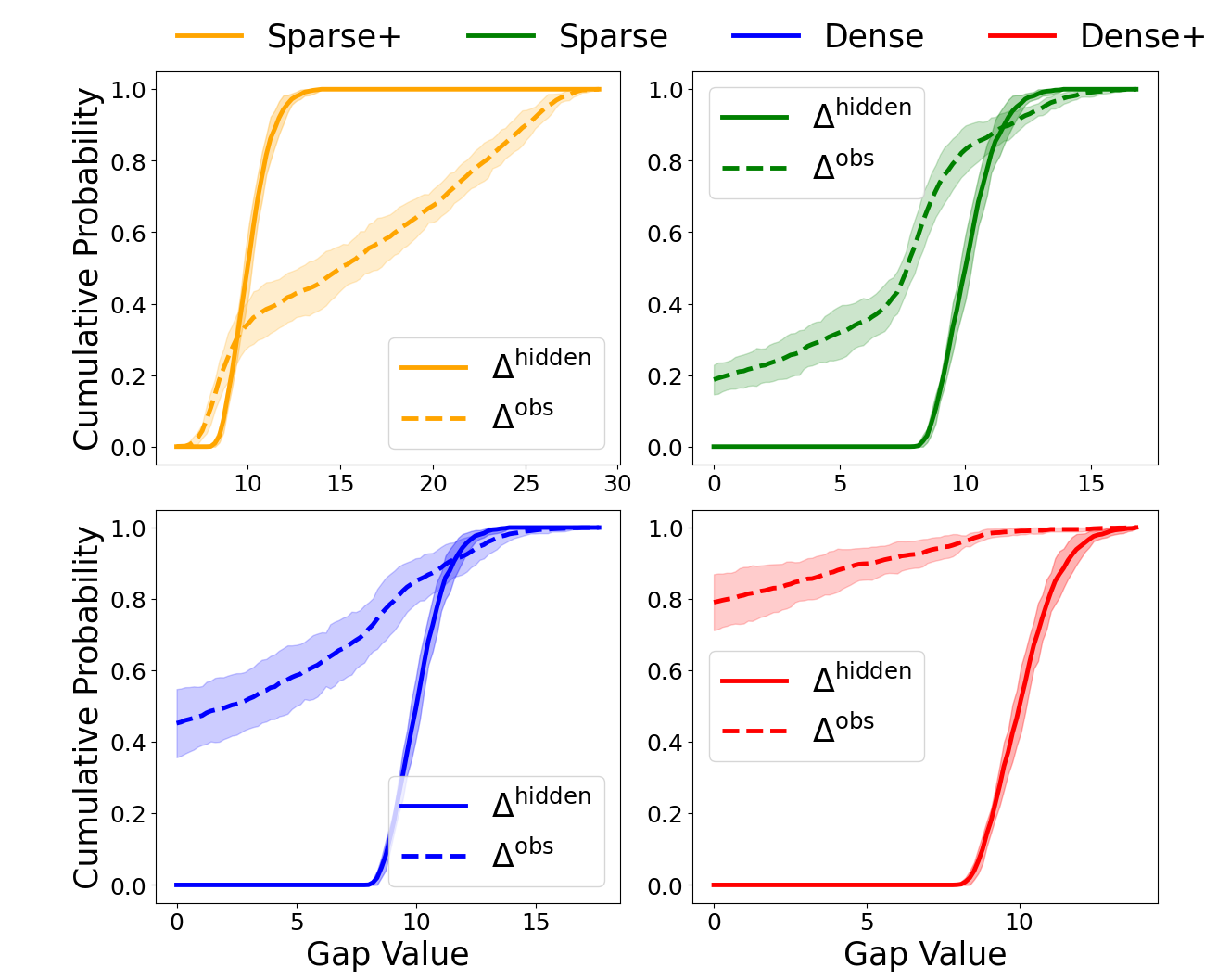}
\caption{Gap dominance verification using  continuous trajectories.}\label{gaps_oracle2}
\end{figure}

\begin{figure}[H]
\centering \includegraphics[width=0.6\linewidth]{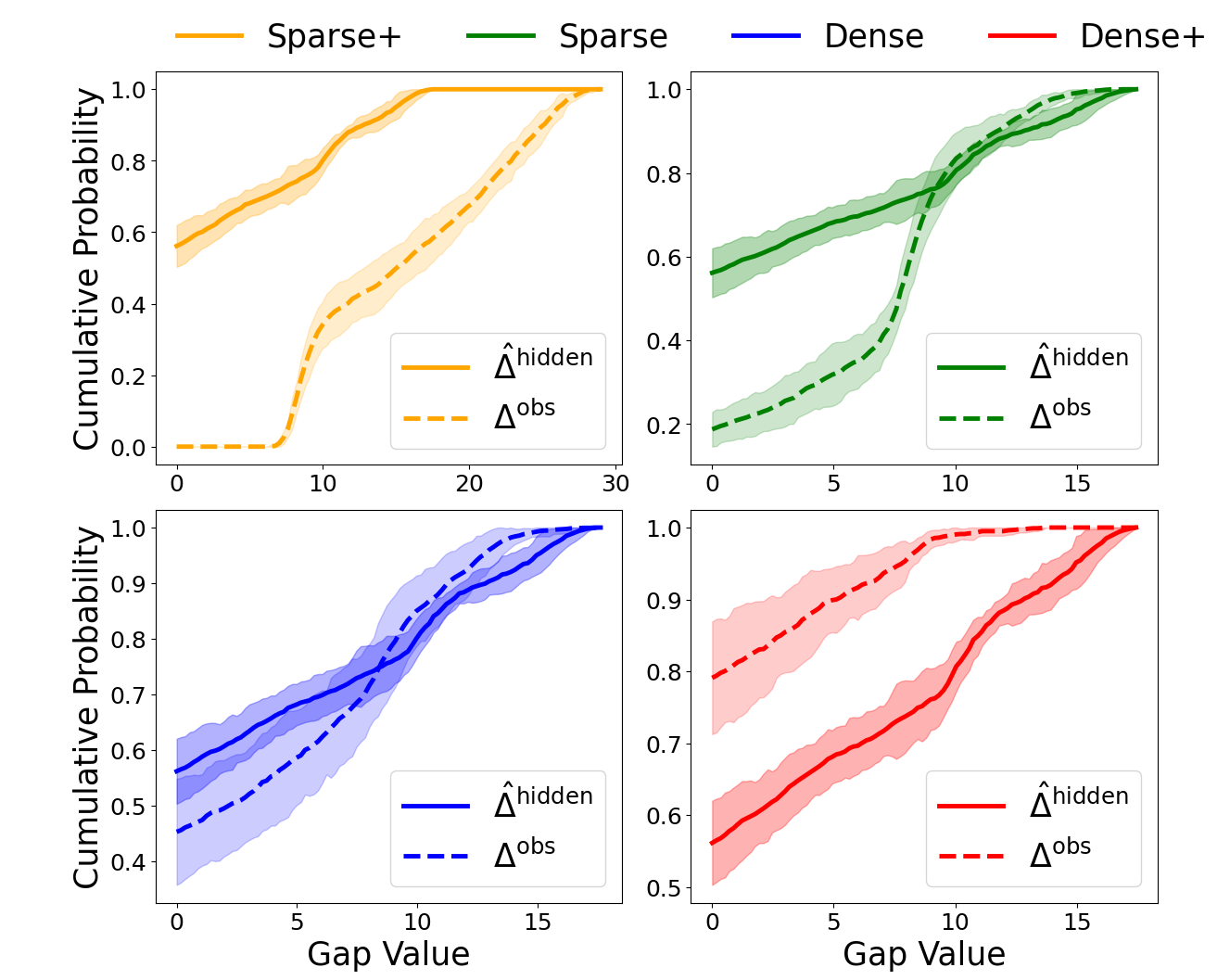}
\caption{Gap dominance verification using surrogate approach  ${\hat{\Delta}^{\mathrm{hidden}}}$.}\label{gaps_surrogate2}
\end{figure}

We replicate this procedure to verify Assumption \ref{sum_stochastic_dominance_assumption} in Fig. \ref{sum_oracle2} and Fig. \ref{sum_surrogate2}.  Although the empirical validation approach captures the effect of subsampling, we observe deviations between the curves, providing slightly optimistic bounds that favor the satisfaction of stochastic dominance where the corresponding oracle plots show that it is not satisfied.

\begin{figure}[H]
\centering \includegraphics[width=0.6\linewidth]{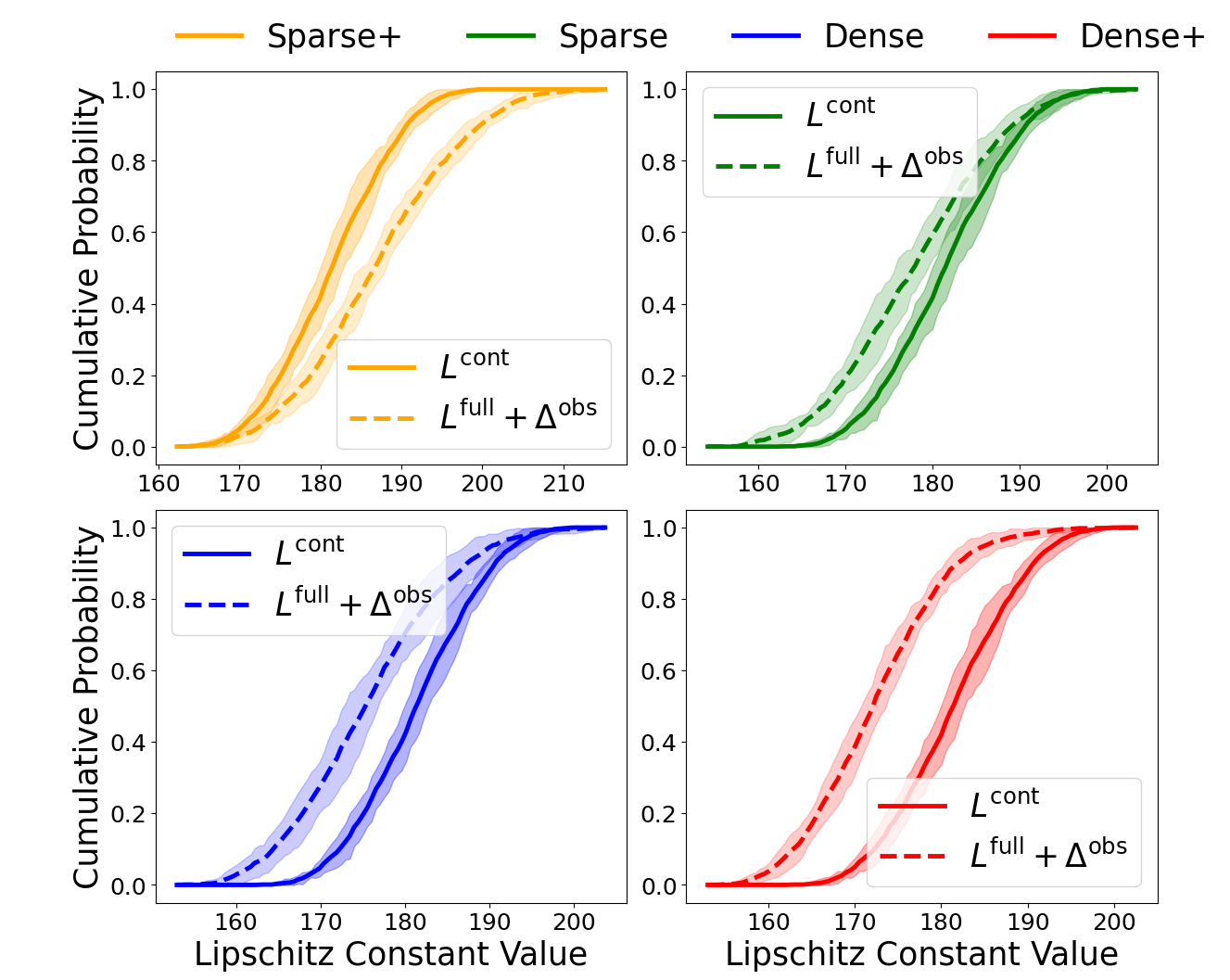}
\caption{Sum dominance verification with access to continuous trajectories.}\label{sum_oracle2}
\end{figure}

\begin{figure}[H]
\centering \includegraphics[width=0.6\linewidth]{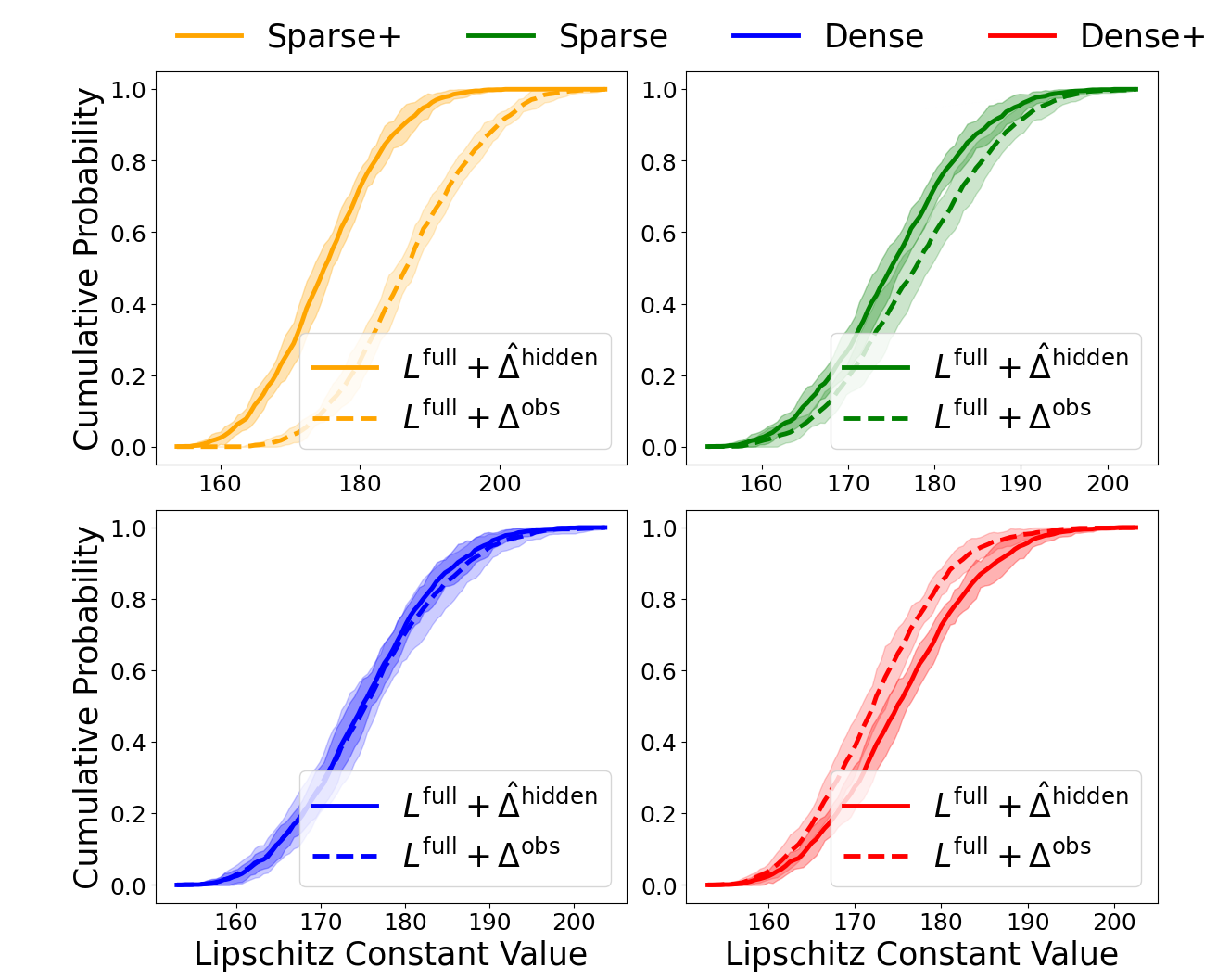}
\caption{Sum dominance using hidden gap surrogates ${\hat{\Delta}^{\mathrm{hidden}}}$.}\label{sum_surrogate2}
\end{figure}

In Fig. \ref{coverages_quantiles_regime2} we observe the same situation as in the previous sampling regime: despite the fact that  Assumption \ref{stochastic_dominance_assumption}
is violated, we still get reliable coverage for the estimation of a valid Lipschtitz constant, even in the extreme case of $\texttt{Dense}+$.

\begin{figure}[H]
    \centering
    \begin{subfigure}[b]{0.49\textwidth}
\includegraphics[width=\textwidth]{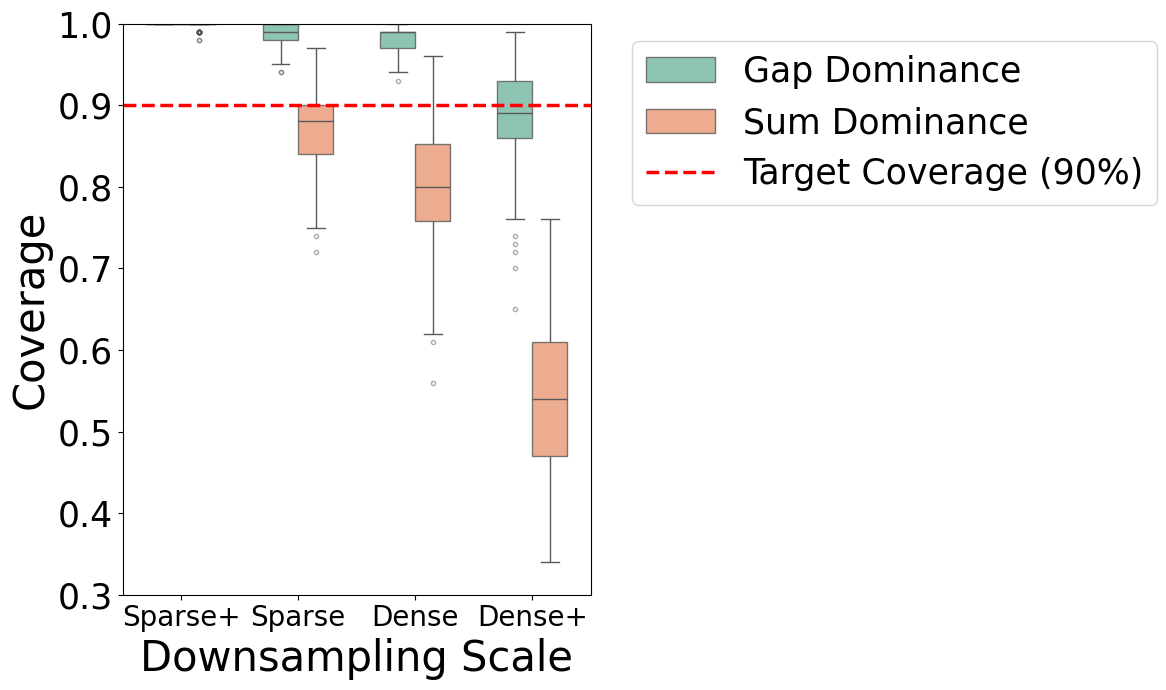}
        \caption{Coverage at different subsampling regimes.}
    \end{subfigure}
    \hfill
    \begin{subfigure}[b]{0.49\textwidth}
\includegraphics[width=\textwidth]{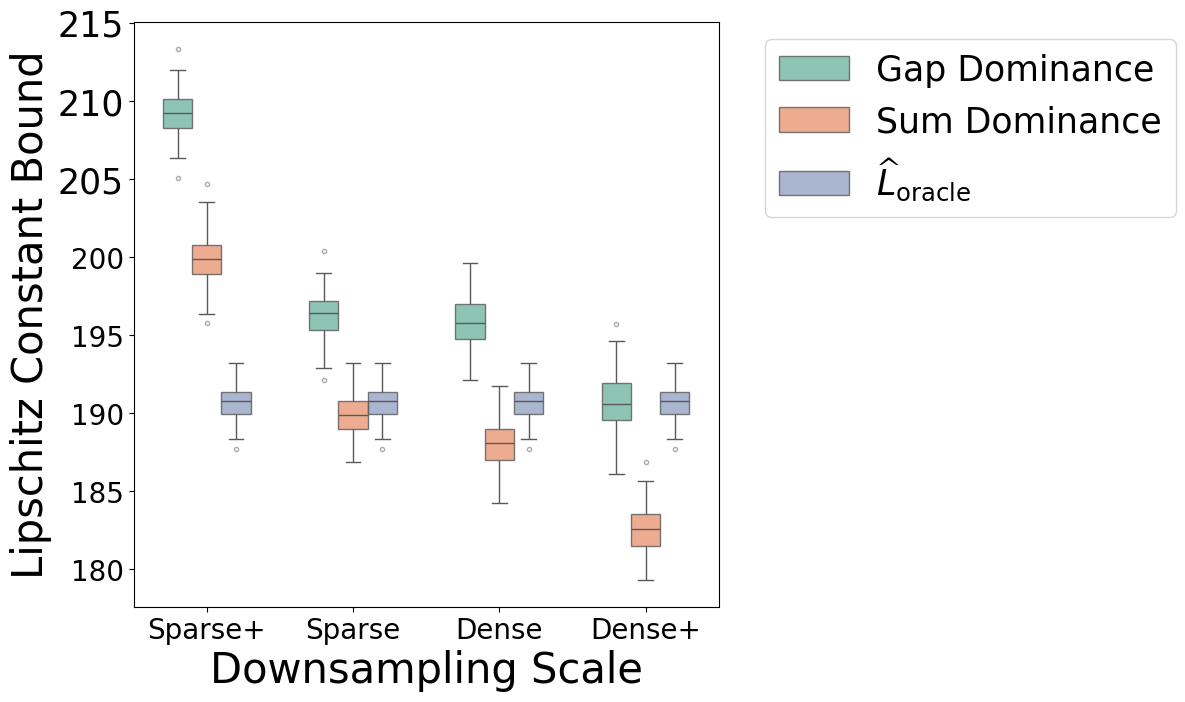}
        \caption{Predicted  Lipschitz upper bounds.}
    \end{subfigure}
    \caption{Experiment results using 100 calibration datasets of 200 trajectories each.}\label{coverages_quantiles_regime2}
\end{figure}


\end{document}